\documentclass[onecolumn,letter]{IEEEtran}

\usepackage{latexsym}
\usepackage{amssymb}
\usepackage{amsmath}
\usepackage{multirow}

\usepackage{color}
\usepackage{bm}
\usepackage[dvips]{graphicx}

\newtheorem{thm}    {Theorem}
\newtheorem{lem}[thm]{Lemma}
\newtheorem{cor}[thm]{Corollary}
\newtheorem{proposition}[thm]{Proposition}
\newtheorem{rem}     {Remark}

 \newenvironment{proofof}[1]{\vspace*{5mm} \par \noindent
         \quad{\it Proof of #1:\hspace{2mm}}}{\endproof
}

\newcommand{\sX}{\mathsf{X}}
\newcommand{\sZ}{\mathsf{Z}}
\newcommand{\sW}{\mathsf{W}}

\def\sM{\mathsf{M}}
\def\sL{\mathsf{L}}

\def\argmax{\mathop{\rm argmax}}
\def\mix{\mathop{\rm mix}}

\def\normal{\mathop{\rm normal}}

\def\Ker{\mathop{\rm Ker}}

\def\FF{\mathbb F}

\def\rG{\mathop{\rm G}}
\def\rH{\mathop{\rm H}}
\def\rq{\mathop{\rm q}}

\newcommand{\bR}{\mathbb{R}}

\newcommand{\cH}{{\cal H}}

\def\cA{{\cal A}}

\def\rE{{\rm E}}
\def\rP{{\rm P}}

\newcommand{\Tr}{{\rm Tr}\,}

\newcommand{\bX}{{\bf X}}
\newcommand{\bbZ}{{\bf Z}}
\newcommand{\bY}{{\bf Y}}

\makeatletter
\newcommand{\lleq}{\mathrel{\mathpalette\gl@align<}}
\newcommand{\ggeq}{\mathrel{\mathpalette\gl@align>}}
\newcommand{\gl@align}[2]{
\vbox{\baselineskip\z@skip\lineskip\z@
\ialign{$\m@th#1\hfil##\hfil$\crcr#2\crcr{}_{{}_{(=)}}\crcr}}}
\makeatother

\def\Label#1{\label{#1}\ [\ \text{#1}\ ]\ }
\def\Label{\label}

\begin{document}
\title{Large deviation analysis for quantum security 
via smoothing of R\'{e}nyi entropy of order 2}
\author{
Masahito Hayashi
\thanks{
M. Hayashi is with Graduate School of Mathematics, Nagoya University, 
Furocho, Chikusaku, Nagoya, 464-8602, Japan, and
Centre for Quantum Technologies, National University of Singapore, 3 Science Drive 2, Singapore 117542.
(e-mail: masahito@math.nagoya-u.ac.jp)
This paper was
presented in part at 
 The 7th Conference on Theory of Quantum Computation, Communication, and Cryptography (TQC2012), Koshiba Hall, The University of Tokyo, Tokyo, Japan, 17-19, May (2012),
and in part at The International Symposium on Quantum Information and Quantum Logic, Zhejiang University, Hangzhou, China, 10-13, August (2012),
and in part at 
2nd Annual conference on Quantum Cryptography (QCRYPT 2012), Singapore, September 10-14, (2012),
and in part at Japan-Singapore Workshop on Multi-user Quantum Networks, Centre for Quantum Technologies, National University of Singapore, Singapore, 17-20 September (2012).}}
\date{}
\maketitle

\begin{abstract}
It is known that
the security evaluation can be done by smoothing of 
R\'{e}nyi entropy of order 2
in the classical and quantum settings
when we apply universal$_2$ hash functions.
Using the smoothing of R\'{e}nyi entropy of order 2,
we derive security bounds for $L_1$ distinguishability and
modified mutual information criterion under the classical and quantum setting,
and have derived these exponential decreasing rates.
These results are extended to the case when 
we apply $\varepsilon$-almost dual universal$_2$ hash functions.
Further, we apply this analysis to the secret key generation 
with error correction.
\end{abstract}

\section{Introduction}
\subsection{Overview}
When a random number is correlated to the third party,
the random number is not secure.
In this case, in order to amplify the privacy, 
one can apply a hash function to the original random number.
This process is called privacy amplification or secret key extraction.  
Bennett et al. \cite{BBCM} and H\r{a}stad et al. \cite{ILL} proposed to use universal$_2$ hash functions for privacy amplification
and derived two universal hashing lemma, which provides an upper bound for 
leaked information based on R\'{e}nyi entropy of order $2$.
In the quantum setting, Renner and K\"onig \cite{R-K} showed that 
the trace norm of the difference between the real state and the ideal state is universally composable. 
Hence, we use the trace norm and call it the $L_1$ distinguishability criterion.
Renner \cite{Renner} extended two universal hashing lemma to the quantum case
and evaluated the $L_1$ distinguishability criterion with universal$_2$ hash functions
based on a quantum version of conditional R\'{e}nyi entropy of order $2$.
In order to apply Renner's two universal hashing lemma to a realistic setting, 
Renner \cite{Renner} attached the smoothing to min entropy,
which is a lower bound on the above quantum version of conditional R\'{e}nyi entropy of order $2$.
That is, he proposed to maximize the min-entropy among the sub-states whose 
trace norm distance to the true state is less
than a given threshold.
However, it is not easy to find the maximizing sub-state.
Instead of the rigorous maximization of min entropy under this condition, 
we can consider a lower bound of the maximum of min entropy.
In the following, we say that this type lower bound or the method based on this type lower bound 
is an {\it approximate smoothing} of min entropy.
In contrast with an approximate smoothing,
we say that the tight value of min entropy under the given condition
or the method based on the tight value is the {\it rigorous smoothing} of min entropy.

Indeed, the same difficulty still holds even for the maximum of R\'{e}nyi entropy of order $2$
under the same condition.
Hence, we can consider an {\it approximate smoothing} of R\'{e}nyi entropy of order $2$.
Considering an approximate smoothing of R\'{e}nyi entropy of order $2$,
the previous paper \cite{H-tight} derived an upper bound of the $L_1$ distinguishability criterion after an application of universal$_2$ hash functions in the classical setting.
In the $n$-fold independent and identical case, 
the upper bound yields a lower bound of the exponential decreasing rate of the $L_1$ distinguishability criterion.
The obtained lower bound is tight with no side information \cite{H-tight}.
The same fact is also shown with classical side information by combination of \cite{H-arxiv} and the forthcoming paper \cite{W-H2}.
This fact shows that the approximate smoothing gives a sub-distribution that is sufficiently close to the sub-distribution maximizing 
the R\'{e}nyi entropy of order $2$ in the classical setting.
However, no study treats the approximate smoothing of R\'{e}nyi entropy of order $2$ in the quantum case.
One of the purposes of this paper is to attach the approximate smoothing of R\'{e}nyi entropy of order $2$
and to evaluate the $L_1$ distinguishability criterion in the quantum case.

Further, when we employ the rigorous smoothing of min entropy instead of approximate smoothing of R\'{e}nyi entropy of order $2$,
we can derive another lower bound of the exponential decreasing rate of the $L_1$ distinguishability criterion.
When there is no side information,
it has been shown in \cite{H-tight} that
the lower bound based on the rigorous smoothing of min entropy is not tight, 
i.e., strictly weaker than 
the bound based on the approximate smoothing of R\'{e}nyi entropy of order $2$ given in \cite{H-tight}.
Further, the paper \cite{H-arxiv} showed the same fact when the side information classical.
Due to this superiority of approximate smoothing of R\'{e}nyi entropy of order $2$ over the rigorous smoothing of min entropy,
it is natural to extend the bound given by \cite{H-tight} to the quantum case.

The security of secret key generation by universal$_2$ hush function 
has been discussed mainly in the cryptography community
and has not been studied in the information theory community
while the problem can be described by information theoretic quantity.
However, the mutual information has not been discussed in this topic
while the mutual information has been widely accepted as the criterion of information security 
by so many papers \cite{CN,DM,AC93,Mau93,M94,M03}.
In fact, the security of the wiretap channel model 
has been mainly discussed with the mutual information 
among information theory community \cite{Wyner,CK79,Csiszar,Hayashi}.
Watanabe \cite{Watanabe} gave an interesting example in the classical setting, in which, 
the mutual information is not close to zero while 
the $L_1$ distinguishability criterion is close to zero.
His example suggests the demand of the convergence of the mutual information.
Therefore, 
it is needed to evaluate 
the security based on the mutual information 
as well as 
the security based on the $L_1$ distinguishability criterion
because so many recent literatures \cite{CKbook,Shikata,BTV,YPS,BKS,ZKVB,BB,PB,LLB,WO,NP,NYBNR,TNG} still accept the mutual information.

However, the mutual information does not reflect the uniformity while it reflects the independence.
In order to address the uniformity as well as leaked information,
we need the modification of mutual information, which is called 
the modified mutual information criterion
and is explained in Subsection \ref{cqs2-2}.
As is shown in Appendix \ref{s8-24}, 
if we suppose several natural conditions for the security criterion,
it is limited to the modified mutual information criterion.
Hence, it is needed to evaluate the modified mutual information criterion as well as the $L_1$ distinguishability criterion.

In fact, 
when one of two security criteria goes to zero exponentially,
the other also goes to zero exponentially due to the relations given in Subsection \ref{cqs2-2}.
Hence,
the asymptotic key generation rate does not depend on the choice of 
the security criterion.
However, 
the relations given in Subsection \ref{cqs2-2} cannot decide 
one of their exponential decreasing rates from the other exponent.
Hence, we need to consider both exponents separately.

\subsection{Main results}
As our result, first, we obtain upper bounds of the above two kinds of secrecy criteria
when Alice and Bob share the same random number and Eve has a correlated quantum state 
by using approximate smoothing of R\'{e}nyi entropy of order $2$
(Theorems \ref{Lem14}, and \ref{Lem12-q}).
This problem is called the secret key generation without errors.
Then, in the independent and identical distributed (i.i.d.) case,
we obtain lower bounds on the exponential decreasing rate 
of the above two kinds of secrecy criteria (Theorems \ref{t-3-16-2}).
We also show that 
the obtained lower bound for universal composable criterion 
is tight in a typical example, in which,
the leaked information is given as a pure state and can be regarded as the environment of Pauli channel.
This fact suggests the superiority of our method even in the quantum setting.

Further, we apply this result to the case
when there exist errors between Alice's and Bob's random variables
and Eve has a correlated quantum state (Theorems \ref{L3-20-7}, and \ref{L12-31-2}).
This problem is called the secret key generation with error correction.
The classical case has been treated by
Ahlswede \& Csisz\'{a}r\cite{AC93}, Maurer\cite{Mau93}, 
and Muramatsu\cite{MUW05} et al.
Renner \cite{Renner} treated the quantum case 
while he did not discuss 
the exponential decreasing rate.
Our analysis derives the exponential decreasing rate
even for the secret key generation with error correction (Theorems \ref{t3-20-20} and \ref{p3-16-1c}).
For derivation of these results,
we need to invent several information quantities and several original technical lemmas, which are given in Section \ref{cqs2}.

Further, we should note that 
the presentation style of this paper has is different from 
that of existing researches \cite{RW,Renner,TSSR11}.
with respect to security evaluation in the single-shot form.
These papers\cite{RW,Renner,TSSR11} bound the length of generated keys 
when the amount of leaked information is fixed.
In contrast, this paper bounds the amount of leaked information
when the length of generated keys is fixed. 
The latter style is useful for evaluation of the exponential decreasing rate.

\subsection{Generalization of main results}
Recently, Tomamichel et al. \cite{TSSR11}
extended two universal hashing lemma, i.e., 
they showed the security with a larger class of hash functions,
which is the class of $\varepsilon$-almost universal$_2$ hash functions in the sense of \cite{Carter,WC81}
when $\varepsilon$ is close to $1$ while they \cite{TSSR11} used a different terminology.
Tsurumaru et al \cite{Tsuru} proposed the concept ``$\varepsilon$-almost dual universal$_2$ hash functions"
for linear universal$_2$ hash functions,
which are defined as the dual functions of $\varepsilon$-almost universal$_2$ hash functions.
They also showed that the $\varepsilon$-almost dual universal$_2$ hash functions contain the original universal$_2$ hash functions when $\varepsilon=2$.
Tsurumaru et al \cite{Tsuru} showed the security of $\varepsilon$-almost dual universal$_2$ hash functions
when $\varepsilon$ increases polynomially with respect to the coding length
while Tomamichel et al. \cite{TSSR11} showed the security of $\varepsilon$-almost universal$_2$ hash functions when $\varepsilon$ is close to $1$.
Tsurumaru et al \cite{Tsuru} also gave an insecure example for $2$-almost universal$_2$ hash functions
over the finite field $\FF_2$.
This example suggests that $\varepsilon$-almost dual universal$_2$ hash functions
have a larger expandability than $\varepsilon$-almost universal$_2$ hash functions.
Further, the forthcoming paper \cite{H-T} gives concrete examples of $\varepsilon$-almost universal$_2$ hash functions 
that have a smaller calculation amount and a smaller number of random variables
than the concatenation of Toeplitz matrix and the identity matrix, which is a typical example of universal$_2$ hash functions.
Hence, it is useful from a applied viewpoint to evaluate the security with $\varepsilon$-almost dual universal$_2$ hash functions.

On the other hand,
Dodis and Smith \cite{DS05} proposed the concept 
``$\delta$-biased family" for a family of random variables.
The concept ``$\varepsilon$-almost dual universal$_2$ hash functions"
can be converted to a part of ``$\delta$-biased family"\cite{DS05,Tsuru}.
Indeed, 
Dodis et al.\cite{DS05} and Fehr et al.\cite{FS08}
showed a security lemma (Proposition \ref{Lem6-1-q}).
Employing this conversion and the above security lemma by \cite{DS05,FS08},
we derive a variant of two universal hashing lemma for 
``$\varepsilon$-almost dual universal$_2$ hash functions"
while Tsurumaru et al \cite{Tsuru} 
showed the security for this class of hash function
by evaluating the virtual decoding phase error probability
by using the relation between the virtual phase error correction and privacy amplification.
The variant can be regarded as a kind of generalization of
two universal hashing lemma by Renner \cite{Renner}.
Replacing the role of two universal hashing lemma by Renner \cite{Renner} by this variant,
we can extend the above result for universal$_2$ hash functions
to the case of the application of ``$\varepsilon$-almost dual universal$_2$ hash functions",
which is a wider class of hash functions
than universal$_2$ hash functions (Lemma \ref{Lem9-q} and
Theorems \ref{t-3-16-3}, \ref{L3-20-8}, \ref{L12-31-2}, \ref{t3-20-21}, and \ref{t3-20-11b}).

\subsection{Relation with second order analysis}
In the 
i.i.d. case,
when the rate of generated random numbers 
is smaller than the entropy rate 
(or conditional entropy rate) of the original information source,
it is possible to generate the random variable, in which,
the $L_1$ distinguishability criterion
approaches zero asymptotically.
In the realistic setting, we can manipulate only a finite size operation.
In order to treat the performance in the finite length setting,
we have two kinds of formalism for the i.i.d. setting.

The first one is the second order formalism, in which,
we focus on the asymptotic expansion up to the second order $\sqrt{n}$
of the length of generated keys $l_{n}$ 
as $l_{n}= H n + C\sqrt{n}+o(\sqrt{n})$
with the constant constraint for the $L_1$ distinguishability criterion.
The second one is the exponent formalism,
in which, we fixed the generation rate $R:= l_{n}/n$
and evaluate the exponential decreasing rate of convergence of 
the $L_1$ distinguishability criterion.
In the exponent formalism,
it is not sufficient to show that the security parameter goes zero exponentially,
and it is required to explicitly give lower and/or upper bounds for the exponential decreasing rate.
The exponent formalism has been studied by 
various information theoretical problems,
e.g., channel coding\cite{Gal,SGB}, source coding\cite{C-source,CKbook,Han-source}, and wire-tap channel\cite{Hayashi,H-leaked}.
In the quantum case,
the same topic has been studied also in 
channel coding\cite{expo-chan}, source coding\cite{Q-source}, wire-tap channel\cite{q-wire}, and entanglement concentration \cite{HKMMW,q-concent}.
As the second order formalism, 
the optimal coding length with the fixed error probability
has been derived up to the second order $\sqrt{n}$ in the various setting \cite{strassen,Hay1,Pol}
in the case of classical channel coding.
The previous paper \cite{Hay1} treats the secret key generation 
with the second order formalism
based on the information spectrum approach \cite{Han}, 
which is closely related to $\epsilon$-smooth min-entropy.
Then, another previous paper \cite{TH}
discusses the randomness extraction with quantum side information
with the second order formalism
by using the relation with 
$\epsilon$-smooth min-entropy \cite{Renner} and 
quantum versions of the information spectrum \cite{H2001,NH}.
The classical case of the result \cite{TH} can be regarded as
a finite-length bound based on smoothing of min entropy.
Note that, as is mentioned by Han \cite{Han},
the information spectrum approach can not yield the optimal exponent of error probability in the channel coding.
This fact suggests that we have to treat the exponent formalism with a method different from the second order formalism.

Since the secret key generation by universal$_2$ hash functions has been studied mainly in 
the cryptography community,
it has not been studied with the exponent formalism sufficiently
while the exponential decreasing rate is a standard topic in the information theory community.
Since the exponential decreasing rate of 
the decoding error probability in the source coding  
is characterized by R\'{e}nyi entropy in the classical \cite{C-source} and the quantum \cite{Q-source} case,
many information theoretical people might be interested in 
whether a similar characterization holds in the secret key generation.

Recently, the previous paper \cite{H-tight} 
derived an exponential decreasing rate of leaked information
in the $L_1$ distinguishability criterion in the classical setting.
The tightness of the rate is shown in the forthcoming paper \cite{W-H2}.
Based on the results \cite{H-tight,TH},
another recent paper \cite{W-H1} numerically dealt with 
the $L_1$ distinguishability criterion in the independently and identically distribution of the binary distribution with the finite-length setting.
It compared the bound based on the second order formalism and the bound based on the exponent formalism in this setting.
It numerically showed that the comparative merits between both depend on 
the length of the data and 
the required amount of the $L_1$ distinguishability criterion.
That is, when the length of the data is not so many and 
the required amount of the $L_1$ distinguishability criterion is too small,
the bound based on the exponent formalism is better than
the bound based on the second order formalism.
Indeed, 
when the required amount of the $L_1$ distinguishability criterion is too small,
the convergence of the second order rate is not uniform.
Hence, the second order formalism does not necessarily work properly 
for an approximation of the finite-length case.
In such a case, from a mathematical viewpoint, we often take the limit of the length of generated keys 
under the condition that 
the required amount of the $L_1$ distinguishability criterion
depends on the length of the data
because such a limit often gives a better approximation of the finite-length case.
The exponent formalism is a particular case of this type of limit.
The numerical analysis in \cite{W-H1} shows 
the importance of the exponent formalism 
when the required amount of the $L_1$ distinguishability criterion is too small at least in the classical case.

While the paper \cite{TH} derives a finite-length bound 
achieving the optimal second order rate by using smoothing of min entropy,
the bound in the classical case
requires the evaluation of the tail probability,
which causes the following drawback.
In the case of binary distribution,
the tail probability can be numerically calculated.
Otherwise, its calculation is not easy 
when the data has a huge size.
Hence, 
we often apply the Berry-Esseen theorem (the central limit theorem).
However, the convergence of Berry-Esseen theorem is not so good when the tail probability is too small.
Instead of Berry-Esseen theorem, we often apply Chernoff bound, which essentially gives the exponential decreasing rate.
This is because Chernoff bound gives a smaller upper bound of the tail probability than Berry-Esseen theorem in this case.
When the tail probability is bounded by Chernoff bound,
this type bound essentially gives an exponential decreasing upper bound based on 
an approximate smoothing of min entropy.
This fact suggests the importance of the exponent formalism 
when the data has a huge size.
We have the similar importance of the exponent formalism in the quantum case
because the numerical calculation based on the bound given in \cite{TH}
is more difficult in the quantum case except for the special example given in \cite{TH}.
Hence, we need to discuss the finite-length bound given in \cite{TH}
from the exponent formalism. 
As is shown in the paper \cite{H-arxiv},
the upper bound by the rigorous smoothing of min entropy does not give the optimal exponential decreasing rate when the side information is classical.
That is,
the finite-length bound given in \cite{TH}
cannot attain the optimal exponent,
and 
the smoothing of R\'{e}nyi entropy of order 2 is 
required for the optimal exponent.
Therefore, this paper addresses only the smoothing of R\'{e}nyi entropy of order 2 under the exponent formalism,

\subsection{Organization}
Now, we give the outline of the preliminary parts.
In Section \ref{cqs2}, we introduce the information quantities for 
evaluating the security and derive several useful inequalities
for the quantum case.
We also give a clear definition for security criteria.
In section \ref{cqs3}, 
we introduce several class of hash functions (universal$_2$ hash functions
and $\varepsilon$-almost dual universal$_2$ hash functions).
We clarify the relation between 
$\varepsilon$-almost dual universal$_2$ hash functions
and $\delta$-biased family.
We also derive an $\varepsilon$-almost dual universal$_2$ version of Renner's two universal hashing lemma \cite[Lemma 5.4.3]{Renner}(Lemma \ref{Lem6-3-q}) based on 
Lemma for $\delta$-biased family given by Dodis et al.\cite{DS05} and Fehr et al.\cite{FS08}
in the classical and quantum setting.
These parts give the definitions for concepts and quantities describing the main results.
The latter preliminary parts are more technical and used for proofs of the main results.
In section \ref{cqs4}, under the universal$_2$ condition or
the $\varepsilon$-almost dual universal$_2$ condition,
we evaluate the $L_1$ distinguishability criterion
and the modified mutual information based on R\'{e}nyi entropy of order $2$
for the quantum setting.

Next, we outline the main results. 
In Section \ref{s4-1}, 
we obtain a suitable bound for the quantum setting in the single-shot setting
by attaching an approximate smoothing of R\'{e}nyi entropy of order 2 
to the evaluation obtained in the previous section.
In Section \ref{s4-1-b}, 
we derive an exponential decreasing rate for both criteria 
for the quantum setting
when we simply apply hash functions and 
there is no error between Alice's and Bob's information.

In Section \ref{s5}, 
we proceed to 
the secret key generation with error correction
for the quantum setting.
In this case, 
we need error correction as well as the privacy amplification.
We derive Gallager bound for the error probability in this setting.
We also derived upper bounds for the $L_1$ distinguishability criterion
and the modified mutual information for a given sacrifice rate.
Based on these upper bounds,
we derive the exponential decreasing rates for both criteria.

In Section \ref{s10}, we apply our result to
the QKD case. That is, the state is given by the quantum communication via 
Pauli channel, which is a typical case in quantum key distribution.
For this example, we showed that our approximate smoothing is tight 
in the sense of exponents.
This evaluation is shown in Appendix \ref{aL8-16-6}.

\begin{table}[htb]
  \caption{Summary of obtained results.}
\begin{center}
  \begin{tabular}{|l|l|l|c|c|} \hline
task& setting & hash functions   &  L1  & MMI \\ \hline
\multirow{6}{*}{PV} & \multirow{3}{*}{single-shot} & universal$_2$ 
& 
\!\!\!\!\begin{tabular}{c}
{\it (\ref{12-5-1-q}) in Lemma \ref{Lem8-q}}\\
(\ref{12-5-6-nb}) in Corollary \ref{c3-29-2-q} \\
(\ref{8-26-13-g}), (\ref{8-26-13-g2}) in Theorem \ref{Lem14} 
\end{tabular}
\!\!\!\!\!\!\!\!
&
\!\!\!\!\begin{tabular}{c}
(\ref{12-5-2-q}) in Lemma \ref{Lem8-q}\\
(\ref{8-26-13-g2b}) in Corollary \ref{c3-29-1-q} \\
(\ref{12-5-6-n}) in Theorem \ref{Lem12-q} 
\end{tabular}
\!\!\!\!\!\!\!\!
\\ \cline{3-5}
& & 
\!\!\!\!\begin{tabular}{l}
$\varepsilon$-almost\\
dual universal$_2$ 
\end{tabular}\!\!\!\! &
\!\!\!\!\begin{tabular}{c}
(\ref{4-17-2}) and (\ref{4-17-2b}) in Lemma \ref{Lem9-q}\\
(\ref{8-26-13-g}), (\ref{8-26-13-g2}) in Theorem \ref{Lem14} 
\end{tabular}
\!\!\!\!\!\!\!\!
&
\!\!\!\!\begin{tabular}{c}
(\ref{12-5-2-a-q}) and (\ref{12-5-2-a-qb}) in Lemma \ref{Lem9-q}\\
(\ref{12-5-6-n}) in Theorem \ref{Lem12-q} 
\end{tabular}
\!\!\!\!\!\!\!\!
\\ \cline{2-5}
& \multirow{3}{*}{exponent} & universal$_2$ 
& (\ref{12-18-9}) in Theorem \ref{t-3-16-2}  
& 
\!\!\!\!\begin{tabular}{c}
(\ref{12-18-7-q}) in Theorem \ref{t-3-16-2}\\
(\ref{12-18-9-x})
\end{tabular}\!\!\!\!
\\ \cline{3-5}
&&
\!\!\!\!\begin{tabular}{l}
$P(n)$-almost\\
dual universal$_2$ 
\end{tabular}\!\!\!\!
& (\ref{12-18-9-a}) in Theorem \ref{t-3-16-3} & 
\!\!\!\!\begin{tabular}{c}
(\ref{12-18-7-a-q}) in Theorem \ref{t-3-16-3}\\
(\ref{12-18-9-x})
\end{tabular}
\!\!\!\!
\!\!\!\!\\ \hline
\multirow{6}{*}{%
\!\!\!\!\begin{tabular}{l}
PV \&\\ 
fixed EC
\end{tabular}%
} &
\multirow{3}{*}{single-shot} &
 universal$_2$ 
& (\ref{12-23-2-q}) in Theorem \ref{L3-20-7} 
& (\ref{12-23-3-q2}) in Theorem \ref{L3-20-7}
\\ \cline{3-5}
& & 
\!\!\!\!\begin{tabular}{l}
$\varepsilon$-almost\\
dual universal$_2$ 
\end{tabular}\!\!\!\!
& (\ref{12-23-4-q}) in Theorem \ref{L3-20-8} 
& (\ref{12-23-3-q2b}) in Theorem \ref{L3-20-8}\\ \cline{2-5}
&\multirow{3}{*}{exponent} & universal$_2$ 
& (\ref{12-23-6-q}) in Theorem \ref{t3-20-20} & (\ref{12-23-7-q}) in Theorem \ref{t3-20-20}\\ \cline{3-5}
& &
\!\!\!\!\begin{tabular}{l}
$P(n)$-almost\\
dual universal$_2$ 
\end{tabular}\!\!\!\!
& (\ref{12-23-6-q}) in Theorem \ref{t3-20-21} 
& (\ref{12-23-7-q}) in Theorem \ref{t3-20-21}\\ \hline
\multirow{6}{*}{%
\!\!\!\!\begin{tabular}{l}
PV \&\\ 
randomized EC
\end{tabular}\!\!\!%
} 
& 
\multirow{3}{*}{single-shot} & 
universal$_2$ & \multirow{6}{*}{no improvement}  
& (\ref{12-23-5-4}) of Theorem \ref{L12-31-2} \\ \cline{3-3} \cline{5-5} 
& &
\!\!\!\!\begin{tabular}{l}
$\varepsilon$-almost\\
dual universal$_2$ 
\end{tabular}\!\!\!\!
&   & (\ref{12-23-5-q-2}) of Theorem \ref{L12-31-2} \\ \cline{2-3} \cline{5-5} 
&\multirow{3}{*}{exponent} & universal$_2$ 
&   & (\ref{12-31-2-c}) in Theorem \ref{p3-16-1c}\\  \cline{3-3} \cline{5-5} 
& &
\!\!\!\!\begin{tabular}{l}
$P(n)$-almost\\
dual universal$_2$ 
\end{tabular}\!\!\!\!
&  & (\ref{12-31-2-c}) in Theorem \ref{t3-20-11b}\\ \hline
  \end{tabular}
\end{center}

\vspace{2ex}

Roman letters express obtained results.
Italic letters express existing results or results with the same performance as existing results.
PV is privacy amplification.
EC is error correction.
L1 is the $L_1$ distinguishability criterion.
MMI is the modified mutual information criterion. 
$P(n)$ is a polynomial.
\Label{table1}
\end{table}

\section{Preparation}\Label{cqs2}
\subsection{Information quantities for single system}
\subsubsection{Case of sub-states}
In order to discuss the security problem in the quantum systems, 
we prepare several information quantities in the single quantum system.
In the following,
a non-negative Hermitian matrix $\rho$ 
is called a sub-state when $\Tr \rho \le 1$.
First, we define the following quantities:
\begin{align}
D(\rho\|\sigma) &:= \Tr \rho (\log \rho-\log \sigma) \\
\psi(s|\rho\|\sigma) & := \log \Tr \rho^{1+s} \sigma^{-s} \\
\underline{\psi}(s|\rho\|\sigma) & := \log \Tr \rho^{\frac{1+s}{2}} \sigma^{-s/2}\rho^{\frac{1+s}{2}} \sigma^{-s/2} .
\end{align}
Then, we obtain the following lemma:
\begin{lem}\Label{L3-26-1}
The functions $s \mapsto
\psi(s|\rho\|\sigma),
\underline{\psi}(s|\rho\|\sigma)$
are convex.
In particular, they are strictly convex
when $\rho$ and $\sigma$ are not completely mixed.
\end{lem}
The proof of Lemma \ref{L3-26-1} is given in Appendix \ref{sL3-26-1}.

Lemma \ref{L3-26-1} yields the following lemma.
\begin{lem}\Label{L21-1}
$\frac{\psi(s|\rho\|\sigma)}{s}$
and
$\frac{\underline{\psi}(s|\rho\|\sigma)}{s}$
are monotonically increasing with respect to $s$
in $(0,\infty)$ and $(-\infty)$.
In particular, they are
strictly monotonically increasing with respect to $s$
when $\rho$ and $\sigma$ are not completely mixed.
\end{lem}

For any quantum operation $\Lambda$,
the following information processing inequalities
\begin{align}
D(\Lambda(\rho)\|\Lambda(\sigma)) & \le D(\rho\|\sigma) , \quad
\psi(s|\Lambda(\rho)\|\Lambda(\sigma)) \le \psi(s|\rho\|\sigma)
\Label{8-21-7-q} 
\end{align}
hold for $s\in (0,1]$\cite[(5,30),(5.41)]{Hayashi-book}.
However, this kind of inequality does not fold for $\underline{\psi}(s|\rho\|\sigma)$ 
in general.

\begin{lem}\Label{L31}
The relation
\begin{align}
\underline{\psi}(s|\rho\|\sigma) & \le \psi(s| \rho\|\sigma) \Label{8-26-2} 
\end{align}
holds for $s\in (0,1]$.
\end{lem}

Lemma \ref{L31} is shown in Appendix \ref{sL31}.
For the latter discussion, we define the pinching map,
which is used for our proof of another lemma.
For a given Hermitian matrix $X$, 
we focus on its spectral decomposition 
$X= \sum_{i=1}^v x_i E_i$,
where $v$ is the number of the eigenvalues of $X$.
Then, 
the pinching map ${\cal E}_X$ is defined as
\begin{align}
{\cal E}_X (\rho):=\sum_{i} E_i \rho E_i.
\end{align}
Then, the inequality 
\begin{align}
\rho \le v {\cal E}_{\sigma}(\rho).
\Label{8-15-23}
\end{align}
holds\cite[Lemma 3.8]{Hayashi-book},\cite{H2001}.
Inequality (\ref{8-15-23}) is used in the proof of Lemma \ref{L31}.


\subsubsection{Case of normalized states}
When $\rho$ and $\sigma$ are normalized states,
we can show several additional useful properties as follows.
In this case, the inequality
$D(\rho\|\sigma) \ge 0$ holds.
The equality holds if and only if $\rho=\sigma$.

Since 
$\psi(0|\rho\|\sigma)=0$
and 
$\underline{\psi}(0|\rho\|\sigma)=0$,
we have
$\lim_{s\to 0} \frac{1}{s}\psi(s|\rho\|\sigma)=D(\rho\|\sigma)$
and 
$\lim_{s\to 0} \frac{1}{s}\underline{\psi}(s|\rho\|\sigma)=D(\rho\|\sigma)$.
Hence, 
Lemma \ref{L21-1} yields the following lemma.
\begin{lem}\Label{L21}
When $\rho$ and $\sigma$ are normalized states, we have
\begin{align}
-\psi(-s|\rho\|\sigma) \le s D (\rho\|\sigma) & \le \psi(s|\rho\|\sigma) \\
-\underline{\psi}(-s|\rho\|\sigma)\le s D (\rho\|\sigma) & \le \underline{\psi}(s|\rho\|\sigma)
\end{align}
for $s>0$.
\end{lem}

\subsection{Information quantities in composite system}
\subsubsection{Case of joint sub-state}
Next, we prepare several information quantities in 
the composite system ${\cal H}_A \otimes {\cal H}_E$,
in which,  ${\cal H}_A$ is a classical system spanned by the basis $\{|a\rangle\}$.
A composite sub-state $\rho$ is called a {\it c-q} sub-state
when it has a form $\rho_{A,E}=
\sum_a P_A(a)|a\rangle \langle a| \otimes \rho_{E|a}$,
in which the conditional state $\rho_{E|a}$ is normalized.
For a given c-q state $\rho_{A,E}$, 
we define the sub-states 
$\rho_E:= \Tr_A \rho_{A,E}$ and $\rho_A:= \Tr_E \rho_{A,E}$.
Then, we define the normalized states
$\rho_{E,\normal}:= \rho_{E}/\Tr \rho_{E}$ and $\rho_{A,\normal}:= \rho_{A}/\Tr\rho_{A}$.
Then, 
the von Neumann entropies 
and 
R\'{e}nyi entropies of order $1+s$
are given as
\begin{align*}
H(A,E|\rho_{A,E}) &:= -\Tr \rho_{A,E} \log \rho_{A,E} \\
H_{1+s}(A,E|\rho_{A,E}) &:=\frac{-1}{s}\log \Tr \rho_{A,E}^{1+s} \\
\end{align*}
with $s \in \bR \setminus \{0\}$.

Quantum versions of the conditional entropy and the min entropy, 
and 
two kinds of quantum versions of conditional R\'{e}nyi entropy 
of order $1+s$ are given as
\begin{align*}
H(A|E|\rho_{A,E}) := H(A,E|\rho_{A,E})-H(E|\rho_{E,\normal}) 
\end{align*}
and
\begin{align*}{H}_{\min}(A|E|\rho_{A,E}) 
:=& - \log \| (I_A \otimes \rho_{E,\normal}^{-1/2})
\rho_{A,E} (I_A \otimes \rho_{E,\normal}^{-1/2}) \|, \\
H_{1+s}(A|E|\rho_{A,E}) 
:=& \frac{-1}{s} \log 
\Tr \rho_{A,E}^{1+s} (I_A \otimes  \rho_{E,\normal}^{-s}), \\
\overline{H}_{1+s}(A|E|\rho_{A,E}) 
:=& \frac{-1}{s} \log 
\Tr \rho_{A,E}^{\frac{1+s}{2}} (I_A \otimes  \rho_{E,\normal}^{-s/2}) \rho_{A,E}^{\frac{1+s}{2}} (I_A \otimes  \rho_{E,\normal}^{-s/2})
\end{align*}
with $s \in \bR \setminus \{0\}$.
These quantities can be written in the following way:
\begin{align}
H(A|E|\rho_{A,E}) 
&= \log |{\cal A}| - D(\rho_{A,E}\|\rho_{\mix,A} \otimes \rho_{E,\normal}) \\
H_{1+s}(A|E|\rho_{A,E}) 
&= \log |{\cal A}| - \frac{1}{s}\psi(s|\rho_{A,E}\|\rho_{\mix,A} \otimes \rho_{E,\normal}) \\
\overline{H}_{1+s}(A|E|\rho_{A,E})
&=
\log |{\cal A}|- \frac{1}{s}\underline{\psi}(s|\rho_{A,E}\|\rho_{\mix,A} \otimes \rho_{E,\normal}) ,
\end{align}
where $\rho_{\mix,A}$ is the completely mixed state on ${\cal H}_A$.
When we replace $\rho_{E,\normal}$ by another normalized state $\sigma_E$ on $\cH_E$, 
we obtain the following generalizations:
\begin{align*}
H(A|E|\rho_{A,E}\|\sigma_E) 
&:= \log |{\cal A}| - D(\rho_{A,E}\|\rho_{\mix,A} \otimes \sigma_E) \\
H_{1+s}(A|E|\rho_{A,E}\|\sigma_E) 
&:= \log |{\cal A}| - \frac{1}{s}\psi(s|\rho_{A,E}\|\rho_{\mix,A} \otimes \sigma_E) \\
\overline{H}_{1+s}(A|E|\rho_{A,E}\|\sigma_E)
&:=
\log |{\cal A}|- \frac{1}{s}\underline{\psi}(s|\rho_{A,E}\|\rho_{\mix,A} \otimes \sigma_E) \\
 {H}_{\min}(A|E|\rho_{A,E}\|\sigma_{E}) 
&:= - \log \| (I_A \otimes \sigma_{E,\normal}^{-1/2})
\rho_{A,E} (I_A \otimes \sigma_E^{-1/2}) \|.
\end{align*}
Lemma \ref{L31} implies that
\begin{align}
\overline{H}_{1+s}(A|E|\rho_{A,E}\|\sigma_E) 
\ge H_{1+s}(A|E|\rho_{A,E}\|\sigma_E) 
\end{align}
for $s \in (0,1]$.
Using Lemma \ref{L21-1}, we obtain the following lemma.
\begin{lem}\Label{L22-1}
$H_{1+s}(A|E|\rho_{A,E}\|\sigma_E)$
and
$\overline{H}_{1+s}(A|E|\rho_{A,E}\|\sigma_E)$
are monotonically decreasing with respect to $s$
in $(0,\infty)$ and $(-\infty,0)$.
In particular, 
they are strictly monotonically decreasing with respect to $s$
in $(0,\infty)$ and $(-\infty,0)$
when $\rho_{A,E}$ and $\sigma_E$ are not completely mixed.
\end{lem}
Further, since 
\begin{align*}
&e^{-\overline{H}_{2}(A|E|\rho_{A,E}\|\sigma_E)} 
= \Tr \rho_{A,E} 
(I_A \otimes \sigma_E^{-1/2})
\rho_{A,E} (I_A \otimes \sigma_E^{-1/2}) \\
\le &
\|(I_A \otimes \sigma_E^{-1/2})
\rho_{A,E} (I_A \otimes \sigma_E^{-1/2})\|
= e^{- {H}_{\min}(A|E|\rho_{A,E}\|\sigma_E)},
\end{align*}
Lemma \ref{L22-1} implies 
the relation 
$\overline{H}_{1+s}(A|E|\rho_{A,E}\|\sigma_E)\ge {H}_{\min}(A|E|\rho_{A,E}\|\sigma_E) $ 
for $s\in (0,1]$.
A similar relation 
${H}_{1+s}(A|E|\rho_{A,E}\|\sigma_E)\ge {H}_{\min}(A|E|\rho_{A,E}\|\sigma_E) $ 
has been shown for $s\in (0,1]$ in \cite{precise}.

When we apply a quantum operation $\Lambda$ on ${\cal H}_E$, 
since
it does not act on the classical system ${\cal A}$,
(\ref{8-21-7-q}) implies that
\begin{align}
H(A|E||\Lambda(\rho_{A,E})\|\Lambda(\sigma_E)) &\ge H(A|E|\rho_{A,E}\|\sigma_E) \\
H_{1+s}(A|E|\Lambda(\rho_{A,E})\|\Lambda(\sigma_E)) &\ge H_{1+s}(A|E|\rho_{A,E}\|\sigma_E ) \Label{8-15-12-qq} .
\end{align}
When we apply the function $f$ to the classical random number $a \in \cA$,
$H(f(A),E|\rho_{A,E}) \le H(A,E|\rho_{A,E})$, i.e., 
\begin{align}
H(f(A)|E|\rho_{A,E}) \le H(A|E|\rho_{A,E}).\Label{8-14-1}
\end{align}

\subsubsection{Case of joint normalized state}
When the joint state $\rho_{A,E}$ is normalized,
we can show several additional useful properties as follows.
In this case, 
since $D(\rho_E\|\sigma_E) \ge 0$, 
we obtain
\begin{align}
H(A|E|\rho_{A,E}\|\sigma_E) 
=
H(A|E|\rho_{A,E})+D(\rho_E\|\sigma_E) 
\ge 
H(A|E|\rho_{A,E}) \Label{12-31-5}.
\end{align}
Further, using Lemma \ref{L21}, we obtain the following lemma.
\begin{lem}\Label{L22}
In particular,
\begin{align}
H_{1-s}(A|E|\rho_{A,E}\|\sigma_E) \ge H(A|E|\rho_{A,E}\|\sigma_E) & \ge H_{1+s}(A|E|\rho_{A,E}\|\sigma_E), \\
\overline{H}_{1-s}(A|E|\rho_{A,E}\|\sigma_E)
\ge H(A|E|\rho_{A,E}\|\sigma_E) & \ge \overline{H}_{1+s}(A|E|\rho_{A,E}\|\sigma_E) \Label{8-15-14} 
\end{align}
for $s>0$.
\end{lem}

Now, we introduce another kind of conditional R\'{e}nyi entropy for a joint normalized state as
\begin{align*}
H_{1+s}^{\rG}(A|E|\rho_{A,E})
&:=-\frac{1+s}{s}\log 
\Tr_E (\Tr_A \rho_{A,E}^{1+s})^{\frac{1}{1+s}}.
\end{align*}
This quantity can be expressed as
\begin{align*}
H_{1+s}^{\rG}(A|E|P_{A,E})
=-\frac{1+s}{s}\phi(\frac{s}{1+s}|A|E|\rho_{A,E})
\end{align*}
by using the Gallager type function \cite{H-tight}:
\begin{align*}
\phi(s|A|E|\rho_{A,E})
&:=\log 
\Tr_E (\Tr_A 
\rho_{A,E}^{1/(1-s)})^{1-s} 
=\log 
\Tr_E 
(\sum_a P_A(a)^{1/(1-s)} \rho_{E|a}^{1/(1-s)} )^{1-s} .
\end{align*}
Taking the limit $s\to 0$,
we obtain
\begin{align}
& 
\lim_{s\to 0}
H_{1+s}^{\rG}(A|E|P_{A,E}) 
= \lim_{s\to 0}
\frac{\phi(s|A|E|\rho_{A,E})}{s} 
= 
\frac{d \phi(s|A|E|\rho_{A,E})}{ds}|_{s=0} \nonumber \\
=
&
H(E|A|\rho_{A,E})
-H(E|\rho_{A,E})
+H(A|\rho_{A,E}) 
= -H(A|E|\rho_{A,E})\Label{1-5-2}.
\end{align}
Then, we obtain the following lemma:
\begin{lem}\Label{cor1-q}
The relation
\begin{align}
\max_{\sigma} 
H_{1+s}(A|E|\rho_{A,E}\|\sigma_E)
= 
H_{1+s}^{\rG}(A|E|P_{A,E}) 
\Label{8-26-8}
\end{align}
holds for $s \in (-1,\infty)$.
The maximum can be realized when $\sigma_E= 
(\Tr_A \rho_{A,E}^{1+s})^{1/(1+s)}/\Tr_E (\Tr_A \rho_{A,E}^{1+s})^{1/(1+s)}$.
\end{lem}
The proof of Lemma \ref{cor1-q} is given in Appendix \ref{scor1-q}.

As a corollary of Lemma \ref{cor1-q}, 
we have the following.
\begin{cor}\Label{c-3-27}
The map $s \mapsto H_{1+s}^{\rG}(A|E|\rho_{A,E})$
is monotonically decreasing for $s \in (-1,\infty)$.
In particular, 
it is strictly decreasing for $s \in (-1,\infty)$
when $\rho_{A,E}$ is not completely mixed.
\end{cor}

\begin{proof}
For $-1<s<t$,
we choose $\sigma_E$ such that
$H_{1+t}(A|E|\rho_{A,E}\|\sigma_E)=H_{1+t}^{\rG}(A|E|\rho_{A,E})$.
Since $s \mapsto H_{1+s}(A|E|\rho_{A,E}) $ is monotonically decreasing (Lemma \ref{L22-1}), 
\begin{align}
& H_{1+t}^{\rG}(A|E|\rho_{A,E})
=
H_{1+t}(A|E|\rho_{A,E}\|\sigma_E)
\le
H_{1+s}(A|E|\rho_{A,E}\|\sigma_E)
\Label{3-27-5} \\
\le &
H_{1+s}^{\rG}(A|E|\rho_{A,E})\nonumber
\end{align}
for $s<t$.
In particular, when $\rho_{A,E}$ is not completely mixed,
Inequality (\ref{3-27-5}) is strict.
Hence, 
the function 
is strictly decreasing for $s \in (-1,\infty)$.
\end{proof}

Given a state $\rho_{A,B,E}$ on ${\cal H}_A\otimes 
{\cal H}_B \otimes {\cal H}_E$,
Lemma \ref{cor1-q} yields that
\begin{align}
& e^{-s H_{1+s}^{\rG}(A|B,E|\rho_{A,B,E} )} 
\le 
\min_{\sigma_E}
e^{-s H_{1+s}(A|B,E|\rho_{A,B,E}\| 
\rho_{\mix,B} \otimes \sigma_E )} \nonumber \\
= &
d_B^s 
\min_{\sigma_E} e^{-s H_{1+s}(A,B|E|\rho_{A,B,E}\| \sigma_E )} 
= 
d_B^s 
e^{-s H_{1+s}^{\rG}(A,B|E|\rho_{A,B,E} )}.
\end{align}
That is, $t=\frac{s}{1+s}\in (0,1)$ satisfies that
\begin{align}
&-t H_{\frac{1}{1-t}}^{\rG}(A|B,E|\rho_{A,B,E} )
= -\frac{s}{1+s} H_{1+s}^{\rG}(A|B,E|\rho_{A,B,E} ) \nonumber \\
\le & \frac{s}{1+s} \log d_B
-\frac{s}{1+s} H_{1+s}^{\rG}(A,B|E|\rho_{A,B,E} )
= t \log d_B
-t H_{\frac{1}{1-t}}^{\rG}(A,B|E|\rho_{A,B,E} ).
\Label{12-20-5-q}
\end{align}

Using the Lemma \ref{cor1-q}, we obtain the following lemma.
\begin{lem}\Label{l4-b}
Given a c-q sub state $\rho_{A,E}=\sum_{a}P_A(a)|a\rangle \langle a| \otimes \rho_{E|a}$,
any TP-CP map $\Lambda$ on ${\cal H}_E$ satisfies that
\begin{align*}
H_{1+s}^{\rG}(A|E|\rho_{A,E}  ) \le H_{1+s}^{\rG}(A|E|\Lambda (\rho_{A,E}  ))
\end{align*}
for $1 \ge s \ge 0$.
\end{lem}

\begin{proof}
Due to (\ref{8-15-12-qq}) and Lemma \ref{cor1-q},
we obtain
\begin{align*}
&
s H_{1+s}^{\rG}(A|E|\rho_{A,E}  )
=
\max_{\sigma_E}
s H_{1+s}(A|E|\rho_{A,E}\|\sigma_E) \\
\le &
\max_{\sigma_E}
s H_{1+s}(A|E|\Lambda (\rho_{A,E})\|\Lambda(\sigma_E)) 
\le 
\max_{\sigma_E}
s H_{1+s}(A|E|\Lambda (\rho_{A,E})\|\sigma_E) 
= 
s H_{1+s}^{\rG}(A|E|\Lambda(\rho_{A,E})  ).
\end{align*}
\end{proof}

\subsection{Criteria for secret random numbers}\Label{cqs2-2}
\subsubsection{Case of joint sub-state}
Next, we introduce criteria for the amount of the information leaked from the secret random number $A$ to $E$ for joint sub-state $\rho_{A,E}$.
Using the trace norm, we can evaluate the secrecy for the state $\rho_{A,E}$ as follows:
\begin{align}
d_1(A|E|\rho_{A,E} ):=\| \rho_{A,E} - \rho_A \otimes \rho_{E} \|_1.
\end{align}
Taking into account the randomness, 
Renner \cite{Renner} defined the following criteria for security of a secret random number:
\begin{align}
d_1'(A|E|\rho_{A,E}):=
\| \rho_{A,E} - \rho_{\mix,A} \otimes \rho_{E} \|_1.
\end{align}
It is known that the quantity is universally composable \cite{R-K}.
We call it the {\it $L_1$ distinguishability criterion}.

Renner\cite{Renner} defined the conditional $L_2$-distance from uniform of $\rho_{A,E}$
relative to a normalized state $\sigma_E$ on ${\cal H}_E$:
\begin{align}
&\underline{d_{2}}(A|E|\rho_{A,E}\|\sigma_E) 
:=
\Tr 
((I \otimes \sigma_E^{-1/4}) (\rho_{A,E} -\rho_{\mix,A} \otimes \rho_E )(I \otimes \sigma_E^{-1/4}) )^2 \nonumber \\
=&
\Tr 
((I \otimes \sigma_E^{-1/4}) \rho_{A,E} (I \otimes \sigma_E^{-1/4}) )^2 
-
\frac{1}{|{\cal A}|}\Tr (\sigma_E^{-1/4} \rho_E \sigma_E^{-1/4})^2
=
e^{-\overline{H}_{2}(A|E|\rho_{A,E}\|\sigma_E)}
-
\frac{1}{|{\cal A}|}
e^{\underline{\psi}(1|\rho_E\|\sigma_E)} \Label{eq-7-23-1}.
\end{align}
Using this value, 
we can evaluate $d_1'(A|E|\rho_{A,E})$ as follows \cite[Lemma 5.2.3]{Renner}
\begin{align}
d_1'(A|E|\rho_{A,E})
\le
\sqrt{|{\cal A}|}
\sqrt{\underline{d_{2}}(A|E|\rho_{A,E}|\sigma_E)}
\Label{4-17-1}.
\end{align}

\subsubsection{Case of joint normalized state}
In the remaining part of this subsection, we assume that
the state $\rho_{A,E}$ is a normalized state.
The correlation between the classical system ${\cal A}$ and the quantum system $\cH_E$ 
can be evaluated by the mutual information
\begin{align}
I(A:E|\rho) &:= D( \rho \| \rho_A \otimes \rho_E) .
\end{align}
This quantity has been adopted by many literatures \cite{CKbook,Shikata,BTV,YPS,BKS,ZKVB,BB,PB,LLB,WO,NP,NYBNR,TNG,CN,DM,AC93,Mau93,M94,M03} as a criteria of independence.
In order to take account into uniformity as well as independence,
we modify the mutual information by using the completely mixed state $\rho_{\mix,A}$ on ${\cal A}$:
\begin{align}
I'(A|E|\rho_{A,E}) &:= D( \rho_{A,E} \| \rho_{\mix,A} \otimes \rho_E) ,
\end{align}
which is called the modified mutual information and satisfies
\begin{align}
I'(A|E|\rho_{A,E}) = I(A:E|\rho_{A,E}) + D(\rho_A\|\rho_{\mix,A} )
\end{align}
and
\begin{align}
H(A|E|\rho_{A,E} ) = -I'(A|E|\rho_{A,E}) +\log |{\cal A}| .
\end{align}
This quantity $I(A:E|\rho_{A,E})$ 
represents the amount of information leaked by $E$,
and the remaining quantity $D(\rho_A\|\rho_{\mix,A} )$
describes the difference of the random number $A$ from the uniform random number.
So, if the quantity $I'(A|E|\rho_{A,E})$ is small,
we can conclude that the random number A has less correlation with E and is close to the uniform random number. In
particular, if the quantity $I'(A|E|\rho_{A,E})$ goes to zero, 
the mutual information $I(A:E|\rho_{A,E})$ goes to zero, and the marginal
distribution $\rho_A$ goes to the uniform distribution.
In this paper, we can adopt 
the quantity $I'(A|E|\rho_{A,E})$ as a criterion for qualifying the secret random number.
The detail validity of the quantity $I'(A|E|\rho_{A,E})$ is given in Appendix \ref{s8-24}.

Using the quantum version of Pinsker inequality,
we obtain
\begin{align}
d_1(A|E|\rho_{A,E} )^2  &\le 2 I(A|E|\rho_{A,E}) \Label{8-19-14-a-q} \\
d_1'(A|E|\rho_{A,E} )^2 &\le 2 I'(A|E|\rho_{A,E}).\Label{8-19-14-q}
\end{align}
Conversely, we can evaluate $I(A:E|\rho_{A,E})$ and $I'(A|E|\rho_{A,E})$ by using $d_1(A|E|\rho_{A,E} )$ and $d_1'(A|E|\rho_{A,E} )$
in the following way.
When $\rho_{A,E}$ is a normalized c-q state,
applying the Fannes inequality, we obtain
\begin{align}
0 \le & I(A:E|\rho_{A,E}) = H(A| \rho_{A,E})+ H(E| \rho_{A,E})- H(A,E|\rho_{A,E}) 
=  H(A,E|\rho_A\otimes \rho_E)- H(A,E|\rho_{A,E}) \nonumber \\
= & \sum_a P_A(a)  H(E|\rho_{E} )- H(E|P\rho_{E|a}) \nonumber \\
\le & \sum_a P_A(a)  \eta (  \| \rho_{E|a} -  \rho_{E} \|_1 ,\log d_E) 
=  \eta (  \| \rho_{A,E} - \rho_A \otimes \rho_{E} \|_1,\log d_E) 
= \eta( d_1(A|E|\rho_{A,E} ),\log d_E)
\Label{8-26-9-a-q}
\end{align}
where $d_E$ is the dimension of ${\cal H}_E$.
Similarly, we obtain
\begin{align}
0 \le & I'(A|E|\rho_{A,E}) = H(A| \rho_{\mix,A})+ H(E| \rho_{A,E})- H(A,E|\rho_{A,E}) 
= H(A,E|\rho_{\mix,A} \otimes \rho_E)- H(A,E|\rho_{A,E}) \nonumber\\
\le & \eta (  \| \rho_{\mix,A} \otimes \rho_E- \rho_{A,E} \|_1, \log |{\cal A}| d_E) 
=  \eta( d_1'(A|E|\rho_{A,E} ), \log |{\cal A}| d_E).
\Label{8-26-9-q}
\end{align}

\section{Ensemble of Hash functions}\Label{cqs3}
\subsection{Ensemble of general hash functions}
In this section, we focus on an ensemble $\{f_{\bX}\}$ of hash functions $f_{\bX}$ from ${\cal A}$ to ${\cal B}$,
where $\bX$ is a random variable identifying the function $f_{\bX}$.
In this case, the total information of Eve's system is written as 
the composite system of ${\cal H}_E$ and $\bX$.
By using the state
$\rho_{f_{\bX}(A),E,\bX}:=
\sum_{  a\in f_{\bX}^{-1}(b) ,x }P_{\bX}(x)P_A(a)
|b\rangle \langle b| \otimes \rho_{E|a} \otimes |x\rangle \langle x|$,
the $L_1$ distinguishability criterion
is written as
\begin{align}
& d_1'(f_{\bX}(A)|E,\bX|\rho_{f_{\bX}(A),E,\bX} )
=
\|\rho_{f_{\bX}(A),E,\bX}- \rho_{\mix,B} \otimes \rho_{E,\bX}\|_1 \nonumber \\
=&
\sum_{x} P_{\bX}(x)
\|\rho_{f_{\bX=x}(A),E}- \rho_{\mix,B} \otimes \rho_{E}\|_1
=
\rE_{\bX}
\|P_{f_{\bX}(A),E}-  \rho_{\mix,B} \otimes \rho_{E} \|_1 .
\end{align}
Then, the modified mutual information is written as
\begin{align}
& I'(f_{\bX}(A)|E,\bX|\rho_{f_{\bX}(A),E,\bX} ) 
=
D(\rho_{f_{\bX}(A),E,\bX}\| \rho_{\mix,B} \otimes \rho_{E,\bX}) \nonumber\\
=&
\sum_{x} P_{\bX}(x)
D(\rho_{f_{\bX=x}(A),E}\| \rho_{\mix,B} \otimes \rho_{E}) 
=
\rE_{\bX}
D(\rho_{f_{\bX}(A),E}\| \rho_{\mix,B} \otimes \rho_{E}) .
\end{align}

We say that a function ensemble 
$\{f_{\bX}\}$ is {\it $\varepsilon$-almost universal$_2$} \cite{Carter,WC81,Tsuru}, 
if, for any pair of different inputs $a_1$,$a_2$, 
the collision probability of their outputs is upper bounded as
\begin{equation}
{\rm Pr}\left[f_{\bX}(a_1)=f_{\bX}(a_2)\right]
\le \frac{\varepsilon}{|{\cal B}|}.
\Label{eq:def-universal-2}
\end{equation}
The parameter $\varepsilon$ appearing in (\ref{eq:def-universal-2}) is shown to be confined in the region
\begin{equation}
\varepsilon\ge\frac{|\cA|-|{\cal B}|}{|\cA|-1},
\Label{eq:epsilon-lower-bound}
\end{equation}
and in particular, 
an ensemble $\{f_{\bX}\}$ with $\varepsilon=1$ is simply called a {\it universal$_2$} function ensemble.

Two important examples of universal$_2$ hash function ensembles 
are the Toeplitz matrices (see, e.g., \cite{MNP90}), and multiplications over a finite field (see, e.g., \cite{Carter,BBCM}).
A modified form of the Toeplitz matrices is also shown to be universal$_2$, which is given by a concatenation $(X, I)$ of the Toeplitz matrix $X$ and the identity matrix $I$ \cite{H-leaked}.
The (modified) Toeplitz matrices are particularly useful in practice, because there exists an efficient multiplication algorithm using the fast Fourier transform algorithm with complexity $O(n\log n)$ (see, e.g., \cite{MatrixTextbook}).

The following proposition holds for any {\it universal$_2$} function ensemble.
\begin{proposition}[{Renner \cite[Lemma 5.4.3]{Renner}}]\Label{Lem5-q}
Given any composite c-q sub-state $\rho_{A,E}$ on $\cH_A \otimes \cH_E$
and any normalized state $\sigma_E$ on $\cH_E$,
any universal$_2$ ensemble of hash functions $f_{\bX}$ 
from $\cA$ to $\{1, \ldots, \sM\}$
satisfies
\begin{align}
\rE_{\bX} \underline{d_{2}}(f_{\bX}(A)|E|\rho_{A,E}\|\sigma_E)
\le
e^{-\overline{H}_{2}(A|E|\rho_{A,E}\|\sigma_E)}.
\end{align}
More precisely, the inequality
\begin{align}
& \rE_{\bX} 
e^{-\overline{H}_{2}(f_{\bX}(A)|E|\rho_{A,E}\|\sigma_E)} 
\le 
(1-\frac{1}{\sM} )e^{-\overline{H}_{2}(A|E|\rho_{E}\|\sigma_E)} 
 +
\frac{1}{\sM} 
e^{\underline{\psi}(1|\rho_{A,E}\|\sigma_E)} 
\end{align}
holds.
\end{proposition}

\subsection{Ensemble of linear hash functions}
Tsurumaru and Hayashi\cite{Tsuru} focus on linear functions over the finite field $\FF_2$. 
Now, we treat the case of linear functions over a finite field $\FF_q$,
where $q$ is a power of a prime number $p$.
That is, the following contents are generalization of the arguments given in \cite{Tsuru}.
Further, the contents withe respect to the modified mutual information
are not given in \cite{Tsuru} even with $q=2$.
We assume that sets ${\cal A}$, ${\cal B}$ are $\FF_q^n$, $\FF_q^m$ respectively with $n\ge m$, 
and $f$ are linear functions over $\FF_q$.
Note that, in this case, there is a kernel $C$ corresponding to a given linear function $f$, 
which is a vector space of 
the dimension $n-m$ or more.
Conversely, when given a vector subspace $C \subset\FF_q^n$ of 
the dimension $n-m$ or more, 
we can always construct a linear function
\begin{equation}
f_{C}: \FF_q^n\to \FF_q^n/C \cong \FF_q^l,\ \ l\le m .
\Label{eq:def-tilde-f}
\end{equation}
That is,
we can always identify a linear hash function $f_{C}$ and a code $C$.

When $C_{\bX}=\Ker f_{\bX}$,
the definition of $\varepsilon$-universal$_2$ function 
ensemble of (\ref{eq:def-universal-2}) takes the form
\begin{equation}
\forall x\in \FF_q^n\setminus\{0\},\ \ {\rm Pr}\left[f_{\bX}(x)=0\right]\le q^{-m}\varepsilon,
\end{equation}
which is equivalent with
\begin{equation}
\forall x\in \FF_q^n\setminus\{0\},\ \ {\rm Pr}\left[x\in C_{\bX}\right]\le q^{-m}\varepsilon.
\end{equation}
This shows that the ensemble of kernel $\{C_{\bX}\}$ contains sufficient 
information for determining if a function ensemble $\{f_{\bX}\}$ 
is $\varepsilon$-almost universal$_2$ or not.

For a given ensemble of codes $\{C_{\bX}\}$,
we define its minimum (respectively, maximum) dimension 
as $t_{\min}:=\min_{\bX}\dim C_{\bX}$ (respectively, $t_{\max}:=\max_{\bX}\dim C_{\bX}$).
Then, we say that a linear code ensemble $\{C_{\bX}\}$ 
of minimum (or maximum) dimension $t$ is an {\it $\varepsilon$-almost universal$_2$} code ensemble, 
if the following condition is satisfied
\begin{equation}
\forall x\in \FF_q^n\setminus\{0\},\ \ {\rm Pr}\left[x\in C_{\bX}\right]\le q^{t-n}\varepsilon.
\Label{eq:C-r-upperbound}
\end{equation}
In particular, if $\varepsilon=1$, we call $\{C_{\bX}\}$ a {\it universal$_2$} code ensemble.


\subsection{Dual universality of a code ensemble}
\Label{sec:approx-duality}
Based on Tsurumaru and Hayashi\cite{Tsuru},
we define several variations of the universality of an ensemble of error-correcting codes 
and the linear functions as follows.
First,
we define the dual code ensemble $\{C_{\bX}\}^\perp$ of a given linear code ensemble $\{C_{\bX}\}$ as the set of all dual codes of $C_{\bX}$. That is, $\{C_{\bX}\}^\perp=\{C_{\bX}^\perp\}$.
We also introduce the notion of dual universality as follows.
We say that a code ensemble $\{C_{\bX}\}$ in $\FF_q^n$
is {\it $\varepsilon$-almost dual universal$_2$}
with minimum dimension $t$ (with maximum dimension $t$), 
if the dual ensemble ${\cal C}^\perp$ is $\varepsilon$-almost universal$_2$
with maximum dimension $n-t$ (with minimum dimension $n-t$).
Hence, 
We say that 
a linear function ensemble $\{f_{\bX}\}$ from $\FF_q^n$ to $\FF_q^m$
is $\varepsilon$-almost dual universal$_2$, 
if the kernels $C_{\bX}$ of $f_{\bX}$ forms an $\varepsilon$-almost dual universal$_2$ 
code ensemble with minimum dimension $n-m$.
This condition is equivalent with the condition that
the ensemble of the linear spaces spanned 
by the generating matrix of $f_{\bX}$ 
forms an $\varepsilon$-almost universal$_2$ 
code ensemble with maximum dimension $m$.
An explicit example of a dual universal$_2$ function ensemble (with $\varepsilon=1$) can be given by the modified Toeplitz matrices mentioned earlier \cite{H-qkd2}, 
i.e., a concatenation $(X, I)$ of the Toeplitz matrix $X$ and the identity matrix $I$.
This example is particularly useful in practice because it is both universal$_2$ and dual universal$_2$, 
and also because there exists an efficient algorithm with complexity $O(n\log n)$.

With these preliminaries, we can present the following theorem as $\FF_q$ extension of \cite[Corollary 2]{Tsuru}:

\begin{proposition}
\Label{thm:almost-universal2}
An $\varepsilon$-almost universal$_2$ surjective liner hash function ensemble $\{f_{\bX}\}$ 
from $\FF_q^n$ to $\FF_q^m$
is $q(1-q^{m}\varepsilon)+(\varepsilon-1)q^{n-m}$-almost dual universal$_2$ liner hash function ensemble.
\end{proposition}
As a special case, we obtain the following.
\begin{cor}
\Label{thm:almost-universal}
Any universal$_2$ linear function ensemble $\{f_{\bX}\}$
over a finite filed $\FF_q$
is 
$q$-almost dual universal$_2$ function ensemble.
\end{cor}

\subsection{Permuted code ensemble}
In order to treat an example of $\varepsilon$-almost universal$_2$ functions,
we consider the case when the distribution is invariant under permutations of the order in $\FF^n_q={\cal A}^n$.
Now, $S_n$ denotes the symmetric group of degree $n$,
and $\sigma(i)=j$ means that $\sigma\in S_n$ maps $i$ to $j$, where $i,j\in\{1,\dots,n\}$.
The code $\sigma(C)$ is 
defined by 
$ \{ x^\sigma:=(x_{\sigma(1)},\dots,x_{\sigma(n)}) |x=(x_1,\dots,x_n)\in C\}$.
Then, 
we introduce the permuted code ensemble 
$\{ \sigma(C) \}_{\sigma \in S_n}$ of a code $C$.
In this ensemble, $\sigma$ obeys the uniform distribution on $S_n$

For an element $x=(x_1,\dots,x_n) \in \FF_q^n$, we can define the empirical distribution $p_x$ on $\FF_q$
as $p_x(a):= \#\{i|x_i=a\}/n$.
So, we denote the set of the empirical distributions on $\FF_q^n={\cal A}^n$ by $T_{n,{\cal A}}$.
The cardinality $|T_{n,{\cal A}}|$ is bounded by $(n+1)^{q-1}$.
Similarly, we define $T_{n,{\cal A}}^+:=T_{n,{\cal A}} \setminus \{1_0\}$,
where $1_0$ is the deterministic distribution on $0\in \FF_q$.
For given a code $C\subset \FF_q^n$, we define 
$\varepsilon_p(C):=
\frac{q^n \#\{x\in C| p_x= p \}}{|C| \#\{x\in \FF_q^n| p_x= p \}}$
and
$\varepsilon(C):=\max_{p \in T_{n,{\cal A}}^+} \varepsilon_p(C).$
Then, we obtain the following lemma.
\begin{lem}
The permuted code ensemble 
$\{ \sigma(C) \}_{\sigma \in S_n}$ of a code $C$
is $\varepsilon(C)$-almost universal$_2$.
\end{lem}

\begin{proof}
For any non-zero element $x'\in \FF_q^n$,
we fix an empirical distribution $p:=p_{x'}$.
Then, 
$x'$ belongs to $\sigma(C)$ with the probability
$\frac{\#\{x\in C| p_x= p \}}{\#\{x\in \FF_q^n| p_x= p \}}$.
That is,
the probability that $x'$ belongs to $\sigma(C)$ 
is less than $\frac{\varepsilon(C)|C|}{q^n}$.
\end{proof}

\begin{lem}\Label{ht2}
For any $t \le n$,
there exists a $t$-dimensional code $C \in \FF_q^n$ such that
\begin{align}
\varepsilon(C)
< (n+1)^{q-1}. \Label{12-20-1}
\end{align}
\end{lem}

\begin{proof}
Let $\{C_{\bX}\}_{\bX}$ be a universal$_2$ code ensemble.
Then,
any $p\in T_{n,{\cal A}}^+$ satisfies 
$\rE_{\bX} \varepsilon_p(C_{\bX})\le 1$. 
The Markov inequality yields
\begin{eqnarray}
\Pr\{
\varepsilon_p(C_{\bX}) \ge |T_{n,{\cal A}}| \}
\le \frac{1}{|T_{n,{\cal A}}|}
\end{eqnarray}
and thus
\begin{align}
\Pr\{
\exists p\in T_{n,{\cal A}}^+ ,~\varepsilon_p(C_{\bX}) \ge |T_{n,{\cal A}}| \}
\le \frac{|T_{n,{\cal A}}|-1}{|T_{n,{\cal A}}|}.
\end{align}
Hence,
\begin{align}
\Pr\{
\forall p\in T_{n,{\cal A}}^+ ,~\varepsilon_p(C_{\bX}) < |T_{n,{\cal A}}| \}
\ge \frac{1}{|T_{n,{\cal A}}|}.
\end{align}
Therefore, 
there exists a code $C$ satisfying the desired condition (\ref{12-20-1}).
\end{proof}

\subsection{$\delta$-biased ensemble}
Next, according to Dodis and Smith\cite{DS05},
we introduce $\delta$-biased ensemble of random variables $\{W_{\bX}\}$.
For a given $\delta>0$,
an ensemble of random variables $\{W_{\bX}\}$ on $\FF_q^n$
is called {\it $\delta$-biased}
when the inequality
\begin{align}
\rE_{\bX} (\rE_{W_{\bX}} (-1)^{x\cdot W_{\bX}})^2 \le \delta^2
\end{align}
holds for any $x\in \FF_q^n$.

We denote the random variable subject to the uniform distribution on a code $C\in \FF_q^n$,
by $W_C$.
Then,
\begin{align}
\rE_{W_C} (-1)^{x\cdot W_{C}}
=
\left\{
\begin{array}{ll}
0 & \hbox{ if } x \notin C^{\perp} \\
1 & \hbox{ if } x \in C^{\perp} .
\end{array}
\right. \Label{3-23-1}
\end{align}
Using the above relation, 
as is suggested in \cite[Case 2]{DS05},
we obtain the following lemma. 

\begin{lem}\Label{Lem6-0}
When
a code ensemble $\{C_{\bX}\}$ in $\FF_q^n$ 
is $\varepsilon$-almost dual universal
with minimum dimension $t$,
the ensemble of random variables $\{W_{C_{\bX}}\}$ in $\FF_q^n$
is $\sqrt{\varepsilon q^{-t}}$-biased.
\end{lem}

\begin{proof}
$\{C^{\perp}_{\bX}\}$ is $\varepsilon$-almost universal
with maximum dimension $n-t$ in $\FF_q^n$.
Hence, for any $x \in \FF_q^n$,
the probability $\Pr \{x \in C^{\perp}_{\bX}\}$ is less than
$\varepsilon q^{-t}$.
Thus, (\ref{3-23-1}) guarantees that
the ensemble of random variables $\{W_{C_{\bX}}\}$ in $\FF_q^n$
is $\sqrt{\varepsilon q^{-t}}$-biased.
\end{proof}

In the following, we treat the case of ${\cal A}=\FF_q^n$.
Given 
a composite state $\rho_{A,E}$ on ${\cal H}_A \otimes {\cal H}_E$
and a distribution $P_{W}$ on ${\cal A}$,
as a quantum generalization of $P_{A,E} * P_{W}$,
we define another composite state 
$\rho_{A,E}* P_{W}:=
\sum_{w}P_{W}(w) \sum_{a}P_A(a)| a+w\rangle\langle a+w|\otimes \rho_{a}^E$.
Then, using this concept,
Fehr and Schaffner \cite{FS08} obtain the following 
proposition as a quantum extension of Lemma 4 of Dodis and Smith\cite{DS05}.
Their proof is based on discrete Fourier transform and is easy to understand.

\begin{proposition}[{\cite[Theorem 3.2]{FS08}}]\Label{Lem6-1-q}
For any c-q sub-state $\rho_{A,E}$ on ${\cal H}_A \otimes {\cal H}_E$
and any normalized state $\sigma_E$ on ${\cal H}_E$,
a $\delta$-biased
ensemble of random variables $\{W_{\bX}\}$ on ${\cal A}=\FF_q^n$
satisfies
\begin{align}
\rE_{\bX} \underline{d_{2}}(A |E| \rho_{A,E}* P_{W_{\bX}} \|\sigma_E)
\le
\delta^2
e^{-\overline{H}_{2}(A|E|\rho_{A,E}\|\sigma_E)}. \Label{3-23-4}
\end{align}
More precisely,
\begin{align}
& \rE_{\bX} \underline{d_{2}}(A |E|\rho_{A,E} * P_{W_{\bX}} \| \sigma_E )
\le 
\delta^2(1-\frac{1}{q^n})
e^{-\overline{H}_{2}(A|E|\rho_{A,E} \| \sigma_E)}.\Label{Lem6-1-q-eq2}
\end{align}
\end{proposition}

Using the above proposition, we can show the following lemma.
\begin{lem}\Label{Lem6-3-q}
Given a c-q sub-state $\rho_{A,E}$ on ${\cal H}_A \otimes {\cal H}_E$
and a normalized state $\sigma_E$ on ${\cal H}_E$.
When $\{C_{\bX}\}$ is an 
$\varepsilon$-almost dual universal$_2$ code ensemble
with minimum dimension $t$,
the ensemble of hash functions $\{f_{C_{\bX}}\}$ 
satisfies
\begin{align}
\rE_{\bX} \underline{d_{2}}(f_{C_{\bX}}(A)|E| \rho_{A,E}\|\sigma_E)
\le
\varepsilon 
e^{-\overline{H}_{2}(A|E|\rho_{A,E}\|\sigma_E )}.
\Label{12-5-9-q}
\end{align}
More precisely,
\begin{align}
& \rE_{\bX} 
e^{-\overline{H}_{2}(f_{C_{\bX}}(A)|E|\rho_{A,E}\|\sigma_E)} 
\le 
\varepsilon 
(1-\frac{1}{q^n} )
e^{-\overline{H}_{2}(A|E|\rho_{A,E}\|\sigma_E )} 
+
\frac{1}{q^{n-t}} 
e^{\underline{\psi}(1|\rho_{A,E}\|\sigma_E)} .
\Label{Lem6-3-q-eq2}
\end{align}
In other words,
an $\varepsilon$-almost dual universal$_2$ function family 
$\{f_{\bX}\}$
from $\FF_q^n$ to $\FF_q^{n-t}$ satisfies (\ref{12-5-9-q}) and (\ref{Lem6-3-q-eq2}).
\end{lem}
\begin{proof}
Due to Lemma \ref{Lem6-0} and (\ref{3-23-4}),
we obtain
\begin{align}
\rE_{\bX} \underline{d_{2}}(A |E|
\rho_{A,E}* P_{W_{C_{\bX}}} \|\sigma_E)
\le
\varepsilon q^{-t}
e^{-\overline{H}_{2}(A|E|\rho\|\sigma_E)}.\Label{12-18-2}
\end{align}

Now, we focus on the relation ${\cal A} \cong {\cal A}/C \times C\cong f_{C} \times C$
for any code $C$.
Then, we obtain
\begin{align*}
& \tilde{\rho}(W_{C})=
\sum_{w \in C}
q^{-t}
\sum_{a}P_A(a)| a+w\rangle \langle a+w|\otimes \rho_{a}^E 
=
\sum_{w \in C}
q^{-t}|w\rangle \langle w|
\otimes 
\sum_{[a] \in {\cal A}/C}
P_{A}([a])| [a]\rangle \langle [a]|\otimes \rho_{[a]}^E \\
= & \sum_{w \in C}
q^{-t}|w\rangle \langle w|
\otimes 
\rho_{f_{C}(A),E}.
\end{align*}
Thus, (\ref{12-18-2}) implies
\begin{align}
& \underline{d_{2}}(A |E|
\rho_{A,E}* P_{W_{C}} \|\sigma_E) 
=
q^{-t}
\underline{d_{2}}(f_{C}(A) |E|\rho_{f_{C}(A),E}\|\sigma_E) 
=
q^{-t}
\underline{d_{2}}(f_{C}(A) |E|\rho_{A,E}\|\sigma_E).\Label{3-23-2}
\end{align}
Therefore, 
\begin{align*}
\rE_{\bX} q^{-t}
\underline{d_{2}}(f_{C_{\bX}}(A) |E|\rho_{A,E}\|\sigma_E)
\le
\varepsilon q^{-t}
e^{-\overline{H}_{2}(A|E|\rho_{A,E}\|\sigma_E )},
\end{align*}
which implies (\ref{12-5-9-q}).

Similarly, Lemma \ref{Lem6-0}, (\ref{Lem6-1-q-eq2}), and (\ref{3-23-2}) imply that
\begin{align*}
& \rE_{\bX} q^{-t}
\underline{d_{2}}(f_{C_{\bX}}(A) |E| \rho_{A,E}\|\sigma_E) 
\le 
\varepsilon q^{-t} (1-\frac{1}{q^n})
e^{-\overline{H}_{2}(A|E|\rho_{A,E}\|\sigma_E)}.
\end{align*}
Since $\rE_{\bX}
\underline{d_{2}}(f_{C_{\bX}}(A) |E|\rho_{A,E}\|\sigma_E) 
=
\rE_{\bX} 
e^{-\overline{H}_{2}(f_{C_{\bX}}(A)|E|\rho_{A,E}\|\sigma_E)}
-\frac{1}{q^{n-t}} e^{\psi(1|\rho_{E} \| \sigma_E)}$,
we have (\ref{Lem6-3-q-eq2}).
\end{proof}

\section{Security bounds with R\'enyi entropy of order 2}\Label{cqs4}
Next, we consider the quantum case
for the security bound based on the R\'{e}nyi entropy of order 2.
Renner\cite[Lemma 5.2.3]{Renner} 
essentially
evaluated 
$\rE_{\bX} d_1'(f_{\bX}(A)|E|\rho_{A,E}) $
by using $\rE_{\bX} \underline{d_{2}}(f_{\bX}(A)|E|\rho_{A,E}\|\sigma_E )$ as follows.

\begin{lem}\Label{Lem7-1-q}
Given
a composite c-q sub-state $\rho_{A,E}$ on $\cH_A \otimes \cH_E$
and 
a normalized state $\sigma_E$ on ${\cal H}_E$,
any ensemble of hash functions $f_{\bX}$ 
from $\cA$ to $\{1, \ldots, \sM\}$
satisfies
\begin{align*}
& \rE_{\bX} d_1'(f_{\bX}(A)|E|\rho_{A,E}) 
\le 
\sM^{\frac{1}{2}}
\sqrt{\rE_{\bX} \underline{d_{2}}(f_{\bX}(A)|E|\rho_{A,E}\|\sigma_E)} 
\end{align*}
Further, the inequalities used in proof of Renner\cite[Corollary
5.6.1]{Renner} imply that
\begin{align*}
&\rE_{\bX} d_1'(f_{\bX}(A)|E|\rho_{A,E}) 
\le 
2 \|\rho_{A,E}-\rho_{A,E}'\|_1
+\rE_{\bX} d_1'(f_{\bX}(A)|E|\rho_{A,E}') \\
\le &
2 \|\rho_{A,E}-\rho_{A,E}'\|_1
+\sM^{\frac{1}{2}}
\sqrt{\rE_{\bX} \underline{d_{2}}(f_{\bX}(A)|E|\rho_{A,E}'\|\sigma_E)}.
\end{align*}
\end{lem}

Applying the same discussion to the von Neumann entropy,
we can evaluate 
the average of
the modified mutual information criterion
by using $\rE_{\bX} \underline{d_{2}}(f_{\bX}(A)|E|\rho_{A,E}\|\sigma_E )$ as follows.

\begin{lem}\Label{Lem7-q}
Assume that 
$\rho_{A,E}$ is a normalized composite c-q state $\rho_{A,E}$ on $\cH_A \otimes \cH_E$.
Any ensemble of hash functions $f_{\bX}$ 
from $\cA$ to $\{1, \ldots, \sM\}$
satisfies
\begin{align}
&\rE_{\bX} I'(f_{\bX}(A)|E|P_{A,E}) 
\le 
\log (1+ \sM\rE_{\bX} \underline{d_{2}}(f_{\bX}(A)|E| \rho_{A,E} ) )  
\Label{12-6-2-q} \\
\le &
\sM\rE_{\bX} \underline{d_{2}}(f_{\bX}(A)|E| \rho_{A,E} ) .
\Label{12-6-2-q+} 
\end{align}
Further, 
when a composite c-q sub-state $\rho_{A,E}'$ satisfies $\rho_E'\le \rho_E$
and $\rho_A'\le \rho_A$,
\begin{align}
\rE_{\bX} I'(f_{\bX}(A)|E|\rho_{A,E}) 
\le &
2 \eta( \| \rho_{A,E}-\rho_{A,E}' \|_1 ,\log \tilde{\sM} ) 
+ \log (1+ \sM\rE_{\bX} \underline{d_{2}}(f_{\bX}(A)|E|\rho_{A,E}'\|\rho_E ) )  ,
\Label{12-6-4-q} \\
\le &
2 \eta( \| \rho_{A,E}-\rho_{A,E}' \|_1 ,\log \tilde{\sM} ) 
+ \sM\rE_{\bX} \underline{d_{2}}(f_{\bX}(A)|E|\rho_{A,E}'\|\rho_E )   ,
\Label{12-6-4-q+}
\end{align}
where $\tilde{\sM}:=\max \{\sM, d_E\}$.
\end{lem}

\begin{proof}
The inequality $\underline{\psi}(1|\rho_{E}' \| \rho_E ) \le 0$
holds because $\rho_E'\le \rho_E$.
Since
\begin{align}
& \underline{d_{2}}(f_{\bX}(A)|E|\rho_{A,E}'\|\rho_E ) 
=  e^{-\overline{H}_2(f_{\bX}(A)|E|\rho_{A,E}'\|\rho_E )}
-\frac{1}{\sM}
e^{\underline{\psi}(1|\rho_{E}' \| \rho_E )} 
\ge  e^{-\overline{H}_2(f_{\bX}(A)|E|\rho_{A,E}' \|\rho_E )}
-\frac{1}{\sM},
\Label{3-28-1}
\end{align}
we have
\begin{align*}
e^{-\overline{H}_2(f_{\bX}(A)|E|\rho_{A,E}' \|\rho_E )} 
\le &
\underline{d_{2}}(f_{\bX}(A)|E|\rho_{A,E}' \|\rho_E )
+\frac{1}{\sM}.
\end{align*}
Taking the logarithm, we obtain
\begin{align}
&
-\log \sM
 +\log (1+ \sM \underline{d_{2}}(f_{\bX}(A)|E|\rho_{A,E}' \|\rho_E ) ) 
\ge  -\overline{H}_2(f_{\bX}(A)|E|\rho_{A,E}' \|\rho_E) 
\ge -H(f_{\bX}(A)|E| \rho_{A,E}' \|\rho_E) .\Label{12-19-1}
\end{align}
Substituting $\rho_{A,E}$ to $\rho_{A,E}'$,
we obtain
$H(f_{\bX}(A)|E| \rho_{A,E} \|\rho_E) =H(f_{\bX}(A)|E| \rho_{A,E}) $
and
\begin{align*}
& I'(f_{\bX}(A)|E|\rho_{A,E}) 
= \log \sM - H(f_{\bX}(A)|E|\rho_{A,E}) 
\le  
\log (1+ \sM \underline{d_{2}}(f_{\bX}(A)|E| \rho_{A,E}  ) ).
\end{align*}
Since the function $x \mapsto \log (1+x)$ is concave,
we obtain
\begin{align*}
& \rE_{\bX} I'(f_{\bX}(A)|E|\rho_{A,E}) 
\le 
\log (1+
\sM \rE_{\bX} \underline{d_{2}}(f_{\bX}(A)|E|\rho_{A,E}  ) ),
\end{align*}
which implies (\ref{12-6-2-q}).
The inequality $\log (1+x) \le x$ yields (\ref{12-6-2-q+}).

Fannes inequality guarantees that
\begin{align*}
&
|
\Tr (\rho_E -\rho_E') \log \rho_E | 
\le 
\eta( \| \rho_E -{\rho'}_E \|_1 \},\log d_E ) 
\le 
\eta ( \| \rho_{A,E} -\rho_{A,E}' \|_1 ,\log \tilde{\sM}),
\end{align*}
and
\begin{align*}
& |H(E|f_{\bX}(A)|\rho_{A,E}\|\rho_A)
-
H(E|f_{\bX}(A)|\rho_{A,E}'\|\rho_A)| 
=
|\sum_{b}P_{f_{\bX}(A)}(b) H(E|\rho_{E|f_{\bX}(A)=b})-H(E|{\rho'}_{E|f_{\bX}(A)=b})| \\
\le &
\sum_{b}P_{f_{\bX}(A)}(b) \log d_E \| \rho_{E|f_{\bX}(A)=b}-{\rho'}_{E|f_{\bX}(A)=b} \|_1 
= 
\log d_E  \| \rho_{f_{\bX}(A),E} -{\rho'}_{f_{\bX}(A),E} \|_1 \\
\le &
\eta ( \| \rho_{A,E} -\rho_{A,E}' \|_1 ,\log d_E) 
\le 
\eta ( \| \rho_{A,E} -\rho_{A,E}' \|_1 ,\log \tilde{\sM}).
\end{align*}
Since the condition $\rho_A' \le \rho_A$ implies 
$-\Tr (\rho_{f_{\bX}(A)} -{\rho'}_{f_{\bX}(A)}) \log \rho_{f_{\bX}(A)}\ge 0$,
we have
\begin{align}
& H(f_{\bX}(A)|E|\rho_{A,E}\|\rho_E) 
-
H(f_{\bX}(A)|E|\rho_{A,E}'\|\rho_E) 
=
H(f_{\bX}(A),E|\rho_{A,E})+\Tr \rho_E \log \rho_E 
-
H(f_{\bX}(A),E|\rho_{A,E}')-\Tr {\rho'}_E \log \rho_E  \nonumber\\
=&
H(E|f_{\bX}(A)|\rho_{A,E}\|\rho_A)
-
H(E|f_{\bX}(A)|\rho_{A,E}'\|\rho_A) 
-\Tr (\rho_{f_{\bX}(A)} -{\rho'}_{f_{\bX}(A)}) \log \rho_{f_{\bX}(A)}
+\Tr (\rho_E -{\rho'}_E) \log \rho_E  \nonumber \\
\ge & 
H(E|f_{\bX}(A)|\rho_{A,E}\|\rho_A)
-
H(E|f_{\bX}(A)|\rho_{A,E}'\|\rho_A) 
+\Tr (\rho_E -{\rho'}_E) \log \rho_E  
\ge 
-2 \eta ( \| \rho_{A,E} -\rho_{A,E}' \|_1 ,\log \tilde{\sM}).
\Label{12-6-5-q}
\end{align}

Therefore, (\ref{12-6-5-q}) and (\ref{12-19-1}) imply that
\begin{align*}
& I'(f_{\bX}(A)|E|\rho_{A,E}) 
= \log \sM - H(f_{\bX}(A)|E|\rho_{A,E}) 
\le 
2 \eta( \| \rho_{A,E}-\rho_{A,E}' \|_1 ,\log \tilde{\sM}) 
+
\log \sM - H(f_{\bX}(A)|E|\rho_{A,E}'\|\rho_E ) \\
\le & 
2 \eta( \| \rho_{A,E}-\rho_{A,E}' \|_1 ,\log \tilde{\sM} ) 
+ \log (1+ \sM \underline{d_{2}}(f_{\bX}(A)|E|\rho_{A,E}' \|\rho_E ) ).
\end{align*}
Therefore, taking the expectation of $\bX$, we obtain (\ref{12-6-4-q}),
which implies (\ref{12-6-4-q+}).

In this proof,
the condition $\rho_E' \le \rho_E$ is crucial 
because Inequality (\ref{3-28-1}) cannot be shown without this condition.
\end{proof}

Now, we evaluate the security by combining Proposition \ref{Lem5-q} and Lemmas \ref{Lem7-1-q} and \ref{Lem7-q}.
For this purpose, we introduce the quantities
\begin{align*}
\Delta_{d,2}(\sM,\varepsilon|\rho_{A,E})
&:=\min_{\sigma_E}
\min_{\rho_{A,E}'}
2 \|\rho_{A,E}-\rho_{A,E}'\|_1 
+\sqrt{\varepsilon} \sM^{\frac{1}{2}}
e^{-\frac{1}{2}\overline{H}_{2}(A|E|\rho_{A,E}'\|\sigma_E)} \\
&=
\min_{\sigma_E}
\min_{\epsilon_1>0} 2 \epsilon_1
+\sqrt{\varepsilon} \sM^{\frac{1}{2}}
e^{-\frac{1}{2}\overline{H}_{2}^{\epsilon_1}(A|E|\rho_{A,E}\|\sigma_E)},\\
\Delta_{I,2}(\sM,\varepsilon|\rho_{A,E})
&:=\min_{\sigma_E}
\min_{\rho_{A,E}': \rho_{E}'\le \sigma_E,}
\eta( \|\rho_{A,E}-\rho_{A,E}'\|_1  ,\log \tilde{\sM})
+ \varepsilon \sM e^{-\overline{H}_{2}(A|E|\rho_{A,E}' \|\rho_E )} \\
&=
\min_{\epsilon_1>0} \eta( \epsilon_1 ,\log \tilde{\sM}) 
+ \varepsilon \sM e^{-\overline{H}_{2}^{\epsilon_1}(A|E|\rho_{A,E} \|\rho_E )} ,
\end{align*}
where $\tilde{\sM}:=\max \{\sM, d_E\}$ and
\begin{align}
\overline{H}_{2}^{\epsilon_1}(A|E|\rho_{A,E}\|\sigma_E)
:=&
\max_{\rho_{A,E}': \|\rho_{A,E}-\rho_{A,E}'\|_1 \le \epsilon_1 } \overline{H}_{2}(A|E|\rho_{A,E}'\|\sigma_E) \\
\overline{H}_{2}^{\epsilon_1}(A|E|\rho_{A,E})
:=&
\max_{\rho_{A,E}': \|\rho_{A,E}-\rho_{A,E}'\|_1 \le \epsilon_1, \rho_{E}'\le \rho_E , \rho_{A}'\le \rho_A } \overline{H}_{2}(A|E|\rho_{A,E}'\|\rho_E). 
\end{align}
Note that $\overline{H}_{2}^{\epsilon_1}(A|E|\rho_{A,E})$ is different from $\overline{H}_{2}^{\epsilon_1}(A|E|\rho_{A,E}\|\rho_E)$
because the definition of $\overline{H}_{2}^{\epsilon_1}(A|E|\rho_{A,E})$
has additional constraints for $\rho_{A,E}'$.
Then, we obtain the following lemma under the universal$_2$ condition.

\begin{lem}\Label{Lem8-q}
Given a normalized state $\sigma_E$ on ${\cal H}_E$ and
c-q sub-states $\rho_{A,E}$, 
any universal$_2$ ensemble of hash functions $f_{\bX}$ 
from $\cA$ to $\{1, \ldots, \sM\}$
satisfies
\begin{align}
\rE_{\bX} d_1'(f_{\bX}(A)|E|\rho_{A,E}) 
\le & 
\sM^{\frac{1}{2}}
e^{-\frac{1}{2}\overline{H}_{2}(A|E|\rho_{A,E}\|\sigma_E)}\nonumber \\
\rE_{\bX} d_1'(f_{\bX}(A)|E|\rho_{A,E}) 
\le &
\Delta_{d,2}(\sM,1|\rho_{A,E})
\Label{12-5-1-q}
\end{align}
When $\rho_{A,E}$ is a normalized c-q state, it satisfies
\begin{align}
\rE_{\bX} I'(f_{\bX}(A)|E|\rho_{A,E} ) 
\le &
\sM e^{-\overline{H}_{2}(A|E|\rho_{A,E} )} \nonumber \\
\rE_{\bX} I'(f_{\bX}(A)|E|\rho_{A,E} ) 
\le &
\Delta_{I,2}(\sM,1|\rho_{A,E}),
\Label{12-5-2-q}
\end{align}
\end{lem}

While the above evaluations of the $L_1$ distinguishability criterion
has been shown in Renner\cite[Corollary 5.6.1]{Renner},
those of the modified mutual information criterion 
have not been shown until now.

Further, 
since $\overline{H}_2(A|E|\rho_{A,E}' \|\sigma_E )\ge
H_{\min}(A|E|\rho_{A,E}' \|\sigma_E )$, 
Renner\cite{Renner} introduced the idea to replace
$\overline{H}_2(A|E|\rho_{A,E}' \|\sigma_E )$
by the min entropy
$H_{\min}(A|E|\rho_{A,E}' \|\sigma_E )$
in (\ref{12-5-1-q}).
For this purpose, based on $H_{\min}(A|E|\rho_{A,E}\|\sigma_E)$,
Renner\cite{Renner} introduced 
$\epsilon_1$-smooth min entropy as
\begin{align}
H_{\min}^{\epsilon_1}(A|E|\rho_{A,E}\|\sigma_E)
:=
\max_{ \|\rho_{A,E}-\rho_{A,E}'\|_1 \le \epsilon_1 } H_{\min}(A|E|\rho_{A,E}'\|\sigma_E).
\end{align}
Then,
Renner\cite[Corollary 5.6.1]{Renner} obtained another upper bound:
\begin{align}
\rE_{\bX} d_1'(f_{\bX}(A)|E|\rho_{A,E}) 
\le 
\Delta_{d,\min}(\sM,\varepsilon|\rho_{A,E}),
\Label{12-5-1-q-min}
\end{align}
where
\begin{align*}
\Delta_{d,\min}(\sM,\varepsilon|\rho_{A,E})
&:=\min_{\sigma_E}
\min_{\epsilon_1>0} 2 \epsilon_1
+\sqrt{\varepsilon} \sM^{\frac{1}{2}}
e^{-\frac{1}{2}\overline{H}_{2}^{\epsilon_1}(A|E|\rho_{A,E}\|\sigma_E)}.
\end{align*}
That is, he proposed to evaluate $\Delta_{d,\min}(\sM,1|\rho_{A,E})$ instead of 
$\Delta_{d,2}(\sM,1|\rho_{A,E})$. 
However, the bound $\Delta_{d,2}(\sM,1|\rho_{A,E})$ gives a strictly better bound in the following sense.

When there is no side information, i.e., the state is given as a distribution $P_A$ on ${\cal A}$,
the previous paper \cite{H-tight} showed that
\begin{align*}
\frac{-1}{n}\log \Delta_{d,\min}(e^{nR},1|P_{A}^{n})
<
\frac{-1}{n}\log \Delta_{d,2}(e^{nR},1|P_{A}^{n}).
\end{align*}
Further, when the side information is classical, i.e., the state is given as a joint distribution $P_{A,E}$ on the joint system,
the paper \cite{H-arxiv} showed that
\begin{align*}
\frac{-1}{n}\log \Delta_{d,\min}(e^{nR},1|P_{A,E}^{n})
<
\frac{-1}{n}\log \Delta_{d,2}(e^{nR},1|P_{A,E}^{n}).
\end{align*}
That is, in these cases, $\Delta_{d,2}(e^{nR},1|\rho_{A,E}^{\otimes n})$ gives a strictly better exponential decreasing rate.
Hence, we focus on the bounds based on R\'{e}nyi entropy of order 2 rather than those based on min entropy.

Since the function $x \mapsto \eta (x,y)$ is concave,
combing Inequality (\ref{8-26-9-q}), we obtain the following corollary.
\begin{cor}\Label{c3-29-1-q}
Any universal$_2$ ensemble of hash functions $f_{\bX}$ 
from $\cA$ to $\{1, \ldots, \sM\}$
and 
any normalized c-q state $\rho_{A,E}$ on $\cH_A \otimes \cH_E$
satisfy
\begin{align}
\rE_{\bX} I'(f_{\bX}(A)|E|\rho_{A,E} ) 
\le 
\eta( \Delta_{d,2}(\sM,1|\rho_{A,E})
, \log |{\cal A}| d_E).
\Label{8-26-13-g2b}
\end{align}
for $s \in (0,1]$.
\end{cor} 

Since the function $x \mapsto \sqrt{x}$ is concave,
combing Inequality (\ref{8-19-14-q}), we obtain the following corollary.
\begin{cor}\Label{c3-29-2-q}
Any universal$_2$ ensemble of hash functions $f_{\bX}$ 
from $\cA$ to $\{1, \ldots, \sM\}$
and 
any normalized c-q state $\rho_{A,E}$ on $\cH_A \otimes \cH_E$
satisfy
\begin{align}
\rE_{\bX} d_1'(f_{\bX}(A)|E|\rho_{A,E} ) 
\le 
2 
\sqrt{\Delta_{I,2}(\sM,1|\rho_{A,E})}
\Label{12-5-6-nb} 
\end{align}
for $s \in (0,1]$.
\end{cor} 

Similarly,
combining Lemmas \ref{Lem6-3-q}, \ref{Lem7-1-q}, and \ref{Lem7-q},
under the $\varepsilon$-almost dual universal$_2$ condition
and employing the same discussion as Corollaries \ref{c3-29-1-q} and \ref{c3-29-2-q},
we can evaluate the average of both security criteria 
as follows.

\begin{lem}\Label{Lem9-q}
Given a normalized state $\sigma_E$ on ${\cal H}_E$ and
c-q sub-states $\rho_{A,E}$ on $\cH_A \otimes \cH_E$.
When an ensemble of linear hash functions $\{f_{\bX}\}_{\bX}$ from $\cA$ to $\{1, \ldots, \sM\}$
is $\varepsilon$-almost dual universal$_2$,
we obtain
\begin{align}
\rE_{\bX} d_1'(f_{\bX}(A)|E|\rho_{A,E}) 
\le & 
\sqrt{\varepsilon}
\sM^{\frac{1}{2}}
e^{-\frac{1}{2}\overline{H}_{2}(A|E|\rho_{A,E}\|\sigma_E )}, \nonumber\\
\rE_{\bX} d_1'(f_{\bX}(A)|E|\rho_{A,E} ) 
\le &
\Delta_{d,2}(\sM,\varepsilon|\rho_{A,E})
\Label{4-17-2}.
\end{align}
When $\rho_{A,E}$ is a normalized c-q state, we have
\begin{align}
\rE_{\bX} d_1'(f_{\bX}(A)|E|\rho_{A,E}) 
\le & 
2 \sqrt{\Delta_{I,2}(\sM,\varepsilon|\rho_{A,E})} \Label{4-17-2b},\\
\rE_{\bX} I'(f_{\bX}(A)|E|\rho_{A,E} ) 
\le &
\varepsilon \sM e^{-\overline{H}_{2}(A|E|\rho_{A,E} )} , \nonumber \\
\rE_{\bX} I'(f_{\bX}(A)|E|\rho_{A,E} ) 
\le &
\Delta_{I,2}(\sM,\varepsilon|\rho_{A,E}),
\Label{12-5-2-a-q} \\
\rE_{\bX} I'(f_{\bX}(A)|E|\rho_{A,E} ) 
\le &
\eta( \Delta_{d,2}(\sM,\varepsilon|\rho_{A,E}), \log |{\cal A}| d_E).
\Label{12-5-2-a-qb}
\end{align}
\end{lem}

Hence, the quantities $\rE_{\bX} d_1'(f_{\bX}(A)|E|\rho_{A,E}) $ and $\rE_{\bX} I'(f_{\bX}(A)|E|\rho_{A,E} )$
can be evaluated by bounding the quantities
$\Delta_{d,2}(\sM,\varepsilon|\rho_{A,E})$
and $\Delta_{I,2}(\sM,\varepsilon|\rho_{A,E})$.
In the next section, we derive upper bounds of these quantities.


\section{Secret key generation with no error: Single-shot case}\Label{s4-1}
In this section.
in order to evaluate the security of secret key generation with no error in the single-shot case,
we evaluate the upper bounds 
$\Delta_{d,2}(\sM,\varepsilon|\rho_{A,E})$,
and 
$\Delta_{I,2}(\sM,\varepsilon|\rho_{A,E})$.

\subsection{$L_1$ distinguishability criterion}\Label{s4-1-2}
In order to describe our upper bound of $\Delta_{d,2}(\sM,\varepsilon|\rho_{A,E})$,
we introduce two notations.
We denote the number of eigenvalues by $v(\sigma_E)$,
and define the real number $\lambda(\sigma_E) := \log a_1-\log a_0$
by using 
the maximum eigenvalue $a_1$
and  
the minimum eigenvalue $a_0$ of $\sigma_E$.
Then, we obtain the following theorem:
\begin{thm}\Label{Lem14}
Given 
any c-q sub-state $\rho_{A,E}$ on 
$\cH_A \otimes \cH_E$
and any normalized state $\sigma_E$ on ${\cal H}_E$,
we have
\begin{align}
\Delta_{d,2}(\sM,\varepsilon|\rho_{A,E}) 
\le & 
(4+ \sqrt{\varepsilon v(\sigma_E)}) \sM^{s/2} 
e^{-\frac{s}{2}H_{1+s}(A|E| \rho_{A,E} \|\sigma_E )}\Label{8-26-13-f} \\
\Delta_{d,2}(\sM,\varepsilon|\rho_{A,E}) 
\le & 
(4+ \sqrt{\varepsilon \lceil \lambda(\sigma_E) \rceil}) \sM^{s/2} 
e^{-\frac{s}{2}H_{1+s}(A|E| \rho_{A,E} \|\sigma_E )+\frac{s}{2}}\Label{8-26-13-f2}
\end{align}
for $s \in (0,1]$.
Further,
when $\rho_{A,E}$ is normalized,
\begin{align}
\Delta_{d,2}(\sM,\varepsilon|\rho_{A,E})
\le &
(4+ \sqrt{\varepsilon v_s }) \sM^{s/2} 
e^{\frac{-s}{2}H_{1+s}^{\rG}(A|E| \rho_{A,E} )} \Label{8-26-13-g} \\
\Delta_{d,2}(\sM,\varepsilon|\rho_{A,E}) 
\le &
(4+ \sqrt{\varepsilon \lceil \lambda_s \rceil}) \sM^{s/2} 
e^{\frac{-s}{2}H_{1+s}^{\rG}(A|E| \rho_{A,E} )+\frac{s}{2}} 
\Label{8-26-13-g2}
\end{align}
for $s \in (0,1]$,
where
$v_s:= v(\Tr_A \rho_{A,E}^{1+s}/\Tr \rho_{A,E}^{1+s}) $
and $\lambda_s:=\lambda ( \Tr_A \rho_{A,E}^{1+s}/\Tr \rho_{A,E}^{1+s})$.
\end{thm}

Indeed, 
the number $v(\sigma_E)$ in crease at most polynomially when $\sigma_E$
is i.i.d.
However, otherwise, it does not generally behaves polynomially 
with respect to the system size
when the system size increases.
On the other hand, the number $\lambda(\sigma_E)$ 
is decided only by the ratio between the maximum and the minimum eigenvalues.
In many cases, we can expect that
the number $\lambda(\sigma_E)$ behaves linearly 
with respect to the system size
when the system size increases.

\begin{proofof}{Theorem \ref{Lem14}}
When $\rho_{A,E}'=P\rho_{A,E} P$ with a projection $P$, we have
$\|\rho_{A,E}'-\rho_{A,E}\|_1 \le 2 \sqrt{\Tr \rho_{A,E} (I-P)}$.
Any projection $P$ satisfies
\begin{align}
\Delta_{d,2}(\sM,\varepsilon|\rho_{A,E})
\le & 4 \sqrt{\Tr \rho_{A,E} (I-P)} + \sM^{1/2} e^{-\frac{1}{2}\overline{H}_{2}(A|E|P\rho_{A,E} P \|\sigma_E)}.
\Label{9-13-1}
\end{align}
We choose $P= \{ {\cal E}_{\sigma_E} (\rho_{A,E}) 
-\frac{1}{\sM} I \otimes \sigma_E \le 0 \}$,
where we simplify ${\cal E}_{I_A \otimes \sigma_E}$ to ${\cal E}_{\sigma_E}$.
Since 
$P$ is commutative with $I \otimes \sigma_E$,
\begin{align}
& \Tr \rho_{A,E} (I-P) =
\Tr \rho_{A,E} {\cal E}_{\sigma_E}(I-P) =
\Tr {\cal E}_{\sigma_E} (\rho_{A,E}) (I-P) 
\le \Tr {\cal E}_{\sigma_E} (\rho_{A,E})^{1+s} \sM^{s} 
(I\otimes \sigma_E^{-s} ) (I-P) \nonumber \\
\le & \Tr {\cal E}_{\sigma_E} (\rho_{A,E})^{1+s} \sM^{s} 
(I\otimes \sigma_E^{-s} ) 
=  \sM^s e^{-s H_{1+s}(A|E|{\cal E}_{\sigma_E} (\rho_{A,E}) 
\|\sigma_E )}
\Label{9-13-4}.
\end{align}
Further, using (\ref{8-15-23}), we have
\begin{align*}
& e^{-\overline{H}_{2}(A|E|P \rho_{A,E} P  \|\sigma_E)} 
= \Tr P \rho_{A,E} P \sigma_E^{-1/2} P 
\rho_{A,E} P \sigma_E^{-1/2} \nonumber \\
\le & v \Tr P {\cal E}_{\sigma_E} (\rho_{A,E}) P 
\sigma_E^{-1/2} P \rho_{A,E} P \sigma_E^{-1/2} 
=  v e^{-H_{2}(A|E|P {\cal E}_{\sigma_E} (\rho_{A,E}) P  \|\sigma_E)} .
\end{align*}
Thus,
\begin{align}
& \sM e^{-\overline{H}_{2}(A|E|P \rho_{A,E} P  \|\sigma_E)} 
\le v \sM e^{-H_{2}(A|E|P {\cal E}_{\sigma_E} (\rho_{A,E}) P  \|\sigma_E)} 
=  v \Tr {\cal E}_{\sigma_E} (\rho_{A,E})^{2} 
\sM (I\otimes \sigma_E^{-1} ) P \nonumber \\
\le & v \Tr {\cal E}_{\sigma_E} (\rho_{A,E})^{1+s} \sM^{s} (I\otimes \sigma_E^{-s} ) P 
\le v \Tr {\cal E}_{\sigma_E} (\rho_{A,E})^{1+s} \sM^{s} 
(I\otimes \sigma_E^{-s} ) 
= v 
\sM^s e^{-s H_{1+s}(A|E|{\cal E}_{\sigma_E}(\rho_{A,E})\|\sigma_E )}. \Label{9-13-5}
\end{align}
Substituting (\ref{9-13-4}) and (\ref{9-13-5}) into RHS of (\ref{9-13-1}),
we obtain
\begin{align}
&\Delta_{d,2}(\sM,\varepsilon|\rho_{A,E}) 
\le 
(4 +\sqrt{\varepsilon  v}) 
\sM^{s/2} 
e^{-\frac{s}{2} H_{1+s}(A|E|{\cal E}_{\sigma_E}(\rho_{A,E})\|\sigma_E )} 
\le 
(4 +\sqrt{\varepsilon v}) \sM^{s/2} e^{-\frac{s}{2}H_{1+s}(A|E| \rho_{A,E} \|\sigma_E )}.
\nonumber 
\end{align}
Hence, we obtain (\ref{8-26-13-f}).

Next, we show (\ref{8-26-13-f2}).
For this purpose, we choose a positive integer $l$.
For the given $\lambda=\lambda (\sigma_E)$, we define $\sigma_E'$ by the following procedure.
First, we diagonalize $\sigma_E$ as $\sigma_E=\sum_y s_y |u_y \rangle \langle u_y|$.
We define $s'_y:= a_0 e^{\lambda i}$ 
when $\log s_y \in ( \log a_0+ \frac{\lambda}{l}(i-1), 
\log a_0+ \frac{\lambda}{l}i ]$ for $i=1, \ldots, l$.
We define $s'_y:= a_0 e^{\lambda }$
when $s_y=a_0$. 
Hence, 
$\sigma_E \le \sigma_E'\le e^{\frac{\lambda}{l}} \sigma_E$
and $1 \le \Tr \sigma_E' \le e^{\frac{\lambda}{l}}$.
Then, $e^{-\frac{\lambda}{l}} \sigma_E \le \frac{\sigma_E'}{\Tr \sigma_E'}$.
Since 
$
e^{-\frac{s}{2}H_{1+s}(A|E| \rho_{A,E} \|
\frac{\sigma_E'}{\Tr \sigma_E'} )} 
\le
e^{-\frac{s}{2}H_{1+s}(A|E| \rho_{A,E} \|
e^{-\frac{\lambda}{l}} \sigma_E )} $,
Inequality (\ref{8-26-13-f}) implies that
\begin{align*}
&
\Delta_{d,2}(\sM,\varepsilon|\rho_{A,E}) 
\le 
(4 +\sqrt{\varepsilon l}) \sM^{s/2} e^{-\frac{s}{2}H_{1+s}(A|E| \rho_{A,E} \|
\frac{\sigma_E'}{\Tr \sigma_E'} )} \\
\le &
(4 +\sqrt{\varepsilon l}) \sM^{s/2} 
e^{-\frac{s}{2}H_{1+s}(A|E| \rho_{A,E} \|
e^{-\frac{\lambda}{l}} \sigma_E  )} 
= 
(4 +\sqrt{\varepsilon l}) 
e^{\frac{s\lambda}{2l}} 
\sM^{s/2} 
e^{-\frac{s}{2}H_{1+s}(A|E| \rho_{A,E} \|
\sigma_E  )} .
\end{align*}
Substituting $\lceil \lambda \rceil $ into $l$, 
we obtain (\ref{8-26-13-f2}).

Applying Lemma \ref{cor1-q},
we obtain (\ref{8-26-13-g})
from (\ref{8-26-13-f}) with $\sigma_E = \frac{\Tr_A \rho_{A,E}^{1+s}}{\Tr_{AE} \rho_{A,E}^{1+s}}$.
Similarly, (\ref{8-26-13-f}) yields (\ref{8-26-13-g2}).
Therefore, we obtain Theorem \ref{Lem14}.
\end{proofof}

\begin{rem}\Label{r3-30-1}
In our proof of the above theorems,
the state $\rho_{A,E}'$ is chosen by 
the information-spectrum-smoothing of the pinched state ${\cal E}_{\sigma_E}(\rho_{A,E})$.
Since the choice in \cite{TH} is also characterized by 
the the information-spectrum-smoothing of the pinched state,
our choice is the same as the choice in \cite{TH}.
\end{rem}

\subsection{Modified mutual information}\Label{s-3-27-1}
The bound $\Delta_{I,2}(\sM,\varepsilon|\rho_{A,E})$
can be evaluated by using the conditional R\'{e}nyi entropy 
$H_{1+s}(A|E| \rho_{A,E} \|\sigma_E )$
as follows.
\begin{thm}\Label{Lem12-q}
\begin{align}
\Delta_{I,2}(\sM,\varepsilon|\rho_{A,E}) 
\le & 
2 \eta(2 \sM^{\frac{s}{2-s}} e^{-\frac{s}{2-s} H_{1+s}(A|E|
{\cal E}_{\rho_E}(\rho_{A,E}) )},
v \varepsilon/4 +\log \tilde{\sM} ) \Label{12-5-6-n-2} \\
\le & 
2 \eta(2 \sM^{\frac{s}{2-s}} e^{-\frac{s}{2-s} H_{1+s}(A|E|\rho_{A,E})},
v \varepsilon/4 +\log \tilde{\sM} ) \Label{12-5-6-n} 
\end{align}
for $s \in (0,1]$,
where $\tilde{\sM}:=\max\{\sM,d_E\}$
and $v$ is the number of eigenvalues of $\rho_E$.
\end{thm}

\begin{proofof}{Theorem \ref{Lem12-q}}
When $\rho_{A,E}'=P\rho_{A,E} P$ with a projection $P$,
$\|\rho_{A,E}'-\rho_{A,E}\|_1 \le 2 \sqrt{\Tr \rho_{A,E} (I-P)}$.
Any projection $P$ satisfies
\begin{align}
\Delta_{I,2}(\sM,\varepsilon|\rho_{A,E}) 
\le & 2 \eta ( 2 \sqrt{\Tr \rho_{A,E} (I-P)} ,\log \tilde{\sM})
+ \varepsilon \sM e^{-\overline{H}_{2}(A|E|P\rho_{A,E} P \|\rho_E )}.
\Label{9-13-1-q}
\end{align}
We choose 
$P= \{{\cal E}_{\rho_E} (\rho_{A,E}) 
-\frac{1}{\sM'} I \otimes \rho_E \le 0 \}$
with arbitrary real number $\sM'$.
Since $P$ is commutative with 
$I_A \otimes \rho_E$ and $\rho_A \otimes I_E$,
the sub-state $\rho_{A,E}'= P \rho_{A,E} P$ satisfies 
$\rho_E' \le \rho_E$ and $\rho_A' \le \rho_A$.
Hence, we can apply Lemmas \ref{Lem7-q} and \ref{Lem9-q}.

Further, since $P$ is commutative with $I_A \otimes \rho_E$,
similar to (\ref{9-13-4}) and (\ref{9-13-5}), we obtain
\begin{align}
\Tr \rho_{A,E} (I-P) 
=\Tr {\cal E}_{\rho_E}(\rho_{A,E}) (I-P) 
\le {\sM'}^s e^{-s H_{1+s}(A|E|{\cal E}_{\rho_E}(\rho_{A,E}) )}
\end{align}
and
\begin{align}
\sM e^{-\overline{H}_{2}(A|E|P \rho_{A,E} P \|\rho_E)} 
\le & v \sM {\sM'}^{s-1} e^{-s H_{1+s}(A|E|{\cal E}_{\rho_E}(\rho_{A,E}) )}.
\end{align}
We choose $\sM':= 
e^{-\frac{s}{2-s}H_{1+s}(A|E|{\cal E}_{\rho_E}(\rho_{A,E}) )}
\sM^{\frac{2}{2-s}}$.
Then, we obtain
\begin{align}
& \Tr \rho_{A,E} (I-P) 
\le  \sM^{\frac{2s}{2-s}} 
e^{-\frac{2s}{2-s} H_{1+s}(A|E|{\cal E}_{\rho_E}(\rho_{A,E}) )}
\Label{8-26-6-n}
\end{align}
and
\begin{align}
& \sM e^{-\overline{H}_{2}(A|E|P \rho_{A,E} P \|\rho_E )} 
\le  v \sM^{\frac{s}{2-s}} 
e^{-\frac{s}{2-s} H_{1+s}(A|E|{\cal E}_{\rho_E}(\rho_{A,E}) )}.
\Label{8-26-7-n}
\end{align}
Substituting (\ref{8-26-6-n}) and (\ref{8-26-7-n}) to (\ref{9-13-1-q}),
we obtain (\ref{12-5-6-n-2}).
Then, (\ref{12-5-6-n}) follows from (\ref{8-15-12-qq}).
Therefore, we obtain Theorem \ref{Lem12-q}.
\end{proofof}

\section{Secret key generation with no error: Asymptotic case}\Label{s4-1-b}
\subsection{Approximate smoothing of R\'{e}nyi entropy of order 2}
Next, we consider the quantum case when our state
is given by the $n$-fold independent and identical state 
$\rho_{A,E}$, i.e., $\rho_{A,E}^{\otimes n}$.
In this case, we focus on 
the optimal generation rate
\begin{align*}
& G(\rho_{A,E}) 
:=
\sup_{\{(f_n,\sM_n)\}}
\left\{\left.
\lim_{n\to\infty} \frac{\log \sM_n}{n}
\right|
d_1'(f_{n}(A_n)|E_n | \rho_{A,E}^{\otimes n} )
\to 0 
\! \right\}.
\end{align*}
Due to Theorem \ref{Lem14},
when the generation rate $R= \lim_{n\to\infty} \frac{\log \sM_n}{n}$ is smaller than $H(A|E)$,
there exists a sequence of functions $f_n:{\cal A}\to \{1,\ldots, e^{nR} \}$ such that
\begin{align}
& d_1'(f_n(A)|E|\rho_{A,E}^{\otimes n}) 
\le (4+\sqrt{v_n}) 
e^{\frac{-s}{2}H_{1+s}^{\rG}(A|E|\rho_{A,E}^{\otimes n} )+\frac{nsR}{2}} 
= (4+\sqrt{v_n}) e^{n( \frac{-s}{2}H_{1+s}^{\rG}(A|E|\rho_{A,E} )+\frac{sR}{2}) } ,
\Label{8-14-4}
\end{align}
where $v_n$ is the number of eigenvalues of $(\Tr_A \rho_{A,E}^{1+s})^{\otimes n}$,
which is a polynomial increasing for $n$.
Since $\lim_{s \to 0}\frac{-1}{2}H_{1+s}^{\rG}(A|E|\rho_{A,E} )=H(A|E|\rho_{A,E}) $,
there exists a number $s \in (0,1]$ such that
$\frac{s}{2}H_{1+s}^{\rG}(A|E|\rho_{A,E} )-\frac{sR}{2}>0$.
Thus, the right hand side of (\ref{8-14-4}) goes to zero exponentially.
Conversely,
due to (\ref{8-14-1}),
any sequence of functions $f_n: {\cal A}^n \mapsto \{1, \ldots, e^{nR} \}$ satisfies that
\begin{align}
\lim_{n \to \infty} \frac{H(f_{n}(A)|E|\rho_{A,E}^{\otimes n})}{n} \le 
\frac{H(A|E|\rho_{A,E}^{\otimes n})}{n} = H(A|E|\rho_{A,E}).
\end{align}
Therefore, 
\begin{align}
\lim_{n \to \infty} 
\frac{I'(f_{n}(A)|E|\rho_{A,E}^{\otimes n})}{n} 
&= R- \lim_{n \to \infty} \frac{H(f_{n}(A)|E|\rho_{A,E}^{\otimes n})}{n} 
\ge 
R- H(A|E|\rho_{A,E}) \Label{8-14-5}.
\end{align}
That is, when $R>H(A|E|\rho_{A,E})$,
$\frac{I'(f_{n}(A)|E|\rho_{A,E}^{\otimes n})}{n} $ does not go to zero.
Due to (\ref{8-19-14-q}),
$d_1'(f_{n}(A)|E|\rho_{A,E}^{\otimes n})$ does not go to zero.
Hence, we can recover the result by \cite{D-W} as
\begin{align}
G(\rho_{A,E}) =H(A|E|\rho_{A,E}).
\end{align}

In order to treat the speed of this convergence,
we focus on 
the {\it exponentially decreasing rate (exponent)} of 
$d_1'(f_n(A)|E|\rho_{A,E}^{\otimes n})$ for a given $R$.
As another criterion,
we also focus on a variant 
$I'(f_{n}(A_n)|E_n | \rho_{A,E}^{\otimes n} )
=
I(f_{n}(A_n):E_n | \rho_{A,E}^{\otimes n} )
+
D(\rho_{f_{n}(A_n)} \| \rho_{\mix,f_{n}(A_n)})$
of the mutual information.

For this purpose, we evaluate the exponential deceasing rates of upper bounds.
For a given polynomial $P(n)$,
Theorems \ref{Lem14} and \ref{Lem12-q} yield that
\begin{align}
 \liminf_{n \to \infty} \frac{-1}{n}\log \Delta_{d,2}(e^{nR},P(n)|\rho_{A,E}^{\otimes n}) 
\ge &
e_{\rG,\rq}(\rho_{A,E}|R)  ,
\Label{12-18-9b} \\
 \liminf_{n \to \infty} \frac{-1}{n}\log \Delta_{I,2}(e^{nR},P(n)|\rho_{A,E}^{\otimes n}) 
\ge & e_{\rH,\rq}(\rho_{A,E}|R) 
\Label{12-18-7-qb},
\end{align}
where
\begin{align*}
e_{\rG,\rq}(\rho_{A,E}|R) 
:= &
\max_{0 \le s \le 1} 
\frac{s}{2}H_{1+s}^{\rG}(A|E|\rho_{A,E})-\frac{s}{2}R 
=  \max_{0 \le t \le \frac{1}{2}} 
\frac{t}{2(1-t)} (H_{\frac{1}{1-t}}^{\rG}(A|E|\rho_{A,E})-R) ,\\
e_{\rH,\rq}(\rho_{A,E}|R) 
:=& \max_{0 \le s \le 1} \frac{s}{2-s} ( H_{1+s}(A|E|\rho_{A,E} ) -R). 
\end{align*}
When the side information is classical, i.e., the state is given as a joint distribution $P_{A,E}$ on the joint system,
the equation 
\begin{align}
& \lim_{n \to \infty} \frac{-1}{n}\log \Delta_{d,2}(e^{nR},\varepsilon|P_{A,E}^{n}) 
= 
\max_{0 \le s \le 1/2} t (H^{\rG}_{\frac{1}{1-t}}(A|E|P_{A,E}) - R)\Label{8-21-1}
\end{align}
is shown by combination of \cite{H-arxiv} and the forthcoming paper \cite{W-H2}\footnote{The part $\ge$ is shown in \cite{H-arxiv}.
The part $\le$ with $\varepsilon=1$ is shown in \cite{W-H2}.
The LHS is monotonically decreasing for $\varepsilon$. Hence, we have (\ref{8-21-1}) with $\varepsilon \ge 1$.}.
Hence, our evaluation (\ref{12-18-9b}) is not tight in general.
However, Equality in (\ref{12-18-9b}) holds in a special case given in Subsection \ref{s10-2} as Lemma \ref{L8-16-6}.
Since the example given in Subsection \ref{s10-2} is very natural in the quantum case,
our evaluation is useful in the quantum setting.

Applying Lemma \ref{Lem8-q}, we obtain the following theorem.
\begin{thm}\Label{t-3-16-2}
When a function ensemble $f_{\bX^n}$ from $\cA^n$ to $\{1,\ldots, \lfloor e^{nR} \rfloor\}$
is universal$_2$,
\begin{align}
\liminf_{n \to \infty} \frac{-1}{n}\log 
\rE_{\bX_n} d_1'(f_{\bX^n}(A_n)|E_n | \rho_{A,E}^{\otimes n} ) 
\ge &
e_{\rG,\rq}(\rho_{A,E}|R)  ,
\Label{12-18-9} \\
\liminf_{n \to \infty} \frac{-1}{n}\log 
\rE_{\bX_n} I'(f_{\bX^n}(A_n)|E_n | \rho_{A,E}^{\otimes n} ) 
\ge & e_{\rH,\rq}(\rho_{A,E}|R) 
\Label{12-18-7-q}.
\end{align}
\end{thm}

Similarly, using Lemma \ref{Lem9-q}, we obtain the following theorem.
\begin{thm}\Label{t-3-16-3}
When an ensemble of linear functions  $f_{\bX^n}$ from $\cA^n$ to $\{1,\ldots, \lfloor e^{nR} \rfloor\}$
is $P(n)$-almost dual universal$_2$,
we have
\begin{align}
\liminf_{n \to \infty} \frac{-1}{n}\log \rE_{\bX_n} 
d_1'(f_{\bX^n}(A_n)|E_n | \rho_{A,E}^{\otimes n} ) 
\ge &
e_{\rG,\rq}(\rho_{A,E}|R)
\Label{12-18-9-a}, \\
\liminf_{n \to \infty} \frac{-1}{n}\log \rE_{\bX_n} I'(f_{\bX^n}(A_n)|E_n | \rho_{A,E}^{\otimes n} ) 
\ge & e_{\rH,\rq}(\rho_{A,E}|R) .
\Label{12-18-7-a-q}
\end{align}
In particular, 
when codes $C_n$ satisfies condition (\ref{12-20-1}),
we have
\begin{align}
\liminf_{n \to \infty} \frac{-1}{n}\log 
d_1'(f_{C_n}(A_n)|E_n | \rho_{A,E}^{\otimes n} )
\ge &
e_{\rG,\rq}(\rho_{A,E}|R)
\Label{12-18-9-b} \\
\liminf_{n \to \infty} \frac{-1}{n}\log  
I'(f_{C_n}(A_n)|E_n | \rho_{A,E}^{\otimes n} )
\ge & e_{\rH,\rq}(\rho_{A,E}|R) .
\Label{12-18-7-b-q}
\end{align}
\end{thm}


\subsection{Comparison for exponents}
Now, we compare exponents given in Theorem \ref{t-3-16-2}
with exponents derived by Corollaries \ref{c3-29-1-q} and \ref{c3-29-2-q}.
When a function ensemble $f_{\bX^n}$ from $\cA^n$ to $\{1,\ldots, \lfloor e^{nR} \rfloor\}$
is universal$_2$,
Corollary \ref{c3-29-2-q} yields the inequality
\begin{align}
\liminf_{n \to \infty} \frac{-1}{n}\log \rE_{\bX_n} 
d_1'(f_{\bX^n}(A_n)|E_n | \rho_{A,E}^{\otimes n} ) 
\ge &
\frac{1}{2}e_{\rH,\rq}(\rho_{A,E}|R).
\Label{12-18-7-q-x}
\end{align}
Similarly,
Corollary \ref{c3-29-1-q} yields the inequality
\begin{align}
\liminf_{n \to \infty} \frac{-1}{n}\log 
\rE_{\bX_n} I'(f_{\bX^n}(A_n)|E_n | \rho_{A,E}^{\otimes n} ) 
\ge & e_{\rG,\rq}(\rho_{A,E}|R) ,
\Label{12-18-9-x}
\end{align}
under the same condition for a function ensemble $f_{\bX^n}$.

In order to compare (\ref{12-18-7-q-x}) and (\ref{12-18-9-x})
with (\ref{12-18-9}) and (\ref{12-18-7-q}), respectively,
we prepare the following lemma for
two exponents $e_{\rH,\rq}(\rho_{A,E}|R)$ and $e_{\rG,\rq}(\rho_{A,E}|R)$.
\begin{lem}\Label{8-29-10}
We obtain
\begin{align}
\frac{1}{2}e_{\rH,\rq}(\rho_{A,E}|R) \le e_{\rG,\rq}(\rho_{A,E}|R) \Label{12-21-30-q} .
\end{align}
Further,
when the relations
\begin{align}
{H}_{1+s}(A|E|\rho_{A,E} ) 
= H_{\frac{1}{1-s}}^{\rG}(A|E|\rho_{A,E} ) 
\Label{ineq-7-23-2}
\end{align}
and
\begin{align}
R \ge R(2/3):=
\left. \frac{(2-s)^2}{2}\frac{d}{ds} 
\frac{s}{2-s}H_{1+s}(A|E|\rho_{A,E} )
\right|_{s=\frac{2}{3}}
\Label{8-29-1}
\end{align}
hold,
we obtain a stronger inequality
\begin{align}
e_{\rH,\rq}(\rho_{A,E}|R) \le  e_{\rG,\rq}(\rho_{A,E}|R) .\Label{12-21-31-q}
\end{align}
\end{lem}

Hence, we can conclude that (\ref{12-18-9}) is better than (\ref{12-18-7-q-x}).
Similarly, under the condition in Lemma \ref{8-29-10}, 
(\ref{12-18-9-x}) is better than (\ref{12-18-7-q}).
However, the relation between 
(\ref{12-18-9-x}) and (\ref{12-18-7-q}) is not clear in general, now.
The condition (\ref{ineq-7-23-2}) seems too restrictive.
However a typical example given in Section \ref{s10}
satisfies the condition.
Hence, Lemma \ref{8-29-10} is often useful.

Therefore, when the number $n$ is sufficiently large,
Inequalities (\ref{8-26-13-g}) and (\ref{8-26-13-g2}) are
better evaluations for the average $\rE_{\bX_n} d_1'(f_{\bX^n}(A_n)|E_n | \rho_{A,E}^{\otimes n} )$
of the $L_1$ distinguishability criterion
s
than 
Corollary \ref{c3-29-2-q}.
In this case, if (\ref{12-21-31-q}) holds,
Corollary \ref{c3-29-1-q} gives
a better evaluation for 
the average of the modified mutual information criterion
$\rE_{\bX_n} I'(f_{\bX^n}(A_n)|E_n | \rho_{A,E}^{\otimes n} ) $ than Inequality (\ref{12-5-6-n}).

\begin{proofof}{Lemma \ref{8-29-10}}
Lemma \ref{cor1-q} yields that
\begin{align}
& \frac{1}{2}e_{\rH,\rq}(\rho_{A,E}|R) 
=
\max_{0 \le s \le 1} \frac{1}{2-s}( \frac{s}{2} H_{1+s}(A|E|\rho_{A,E} ) -\frac{s}{2}  R) 
\le 
\max_{0 \le s \le 1} \frac{1}{2-s} 
(\frac{s}{2}H_{1+s}^{\rG}(A|E|\rho_{A,E} ) -\frac{s}{2}  R) 
\nonumber \\
\le &
\max_{0 \le s \le 1} 
\frac{s}{2} (H_{1+s}^{\rG}(A|E|\rho_{A,E} ) - R) \Label{12-22-1-q} 
=  e_{\rG,\rq}(\rho_{A,E}|R),
\end{align}
where 
Inequality (\ref{12-22-1-q}) follows from 
the non-negativity of the RHS of (\ref{12-22-1-q}) and
the inequality $\frac{1}{2-s} \le 1$.

Next, we show (\ref{12-21-31-q}).
Assume that the relations
(\ref{8-29-1}) and (\ref{ineq-7-23-2}) hold.
We choose $\mu(s):= s H_{1+s}(A|E|\rho_{A,E} )$.
Then, $\mu'(s)\ge 0$ and $\mu''(s) \le 0$.
Defining
$R(s):=\frac{(2-s)^2}{2} \frac{d}{ds} \frac{\mu(s) }{2-s}
= \frac{\mu(s)}{2} 
+\frac{2-s}{2}\mu'(s)$,
we have
\begin{align*}
& \frac{d}{ds}
\frac{1}{2-s}(\mu(s) -s  R)
=
\frac{d}{ds} (\frac{\mu(s) }{2-s})
-\frac{2}{(2-s)^2} R 
= 
\frac{d}{ds} (\frac{\mu(s) }{2-s})
-\frac{2}{(2-s)^2} R(s)
+(R(s)-R)\frac{2}{(2-s)^2} 
= 
\frac{2(R(s)-R)}{(2-s)^2} .
\end{align*}
Since
$\frac{d}{ds} R(s)=\frac{2-s}{2}\mu''(s)\le 0$,
the maximum 
$\max_s \frac{1}{2-s}(\mu(s) -s  R)$
can be attained only when $R=R(s)$.
Hence, when (\ref{8-29-1}) holds, 
\begin{align}
e_{\rH,\rq}(\rho_{A,E}|R)
=& \max_{0 \le s \le 1} \frac{s}{2-s}(H_{1+s}(A|E|\rho_{A,E} ) - R) 
= \max_{0 \le s \le 2/3} \frac{s}{2-s}(H_{1+s}(A|E|\rho_{A,E} ) -R) .\nonumber
\end{align}
Now, we choose $t$ by $\frac{t}{2-2t}=\frac{s}{2-s} $.
Then, $0 \le t \le 1/2$ and $t \le s$
when $0 \le s \le 2/3$.
Hence,
$H_{1+s}(A|E|\rho_{A,E} ) -R
\le
H_{1+t}(A|E|\rho_{A,E} ) -R$,
which implies that
\begin{align}
& \max_{0 \le s \le 2/3} \frac{s}{2-s}(H_{1+s}(A|E|\rho_{A,E} ) -R) 
\le
\max_{0 \le t \le 1/2} \frac{t}{2-2t}(H_{1+t}(A|E|\rho_{A,E} ) - R) .\nonumber
\end{align}
Since the relation (\ref{ineq-7-23-2}) holds,
\begin{align}
 \max_{0 \le t \le 1/2} \frac{t}{2-2t}(H_{1+t}(A|E|\rho_{A,E} ) - R) 
= & \max_{0 \le t \le 1/2} 
\frac{t}{2-2t}( H_{\frac{1}{1-t}}^{\rG}(A|E|\rho_{A,E} ) - R) 
=  e_{\rG,\rq}(\rho_{A,E}|R).\nonumber
\end{align}
\end{proofof}

\subsection{Non i.i.d. case}\Label{s4-1-0-b-qb}
Finally, we consider our bounds 
when the state $\rho_{A,E}^{(n)}$ is given as non i.i.d. state on the system 
$({\cal H}_A \otimes {\cal H}_E)^{\otimes n}$.

In this case,
the speeds of increase of 
$v$ and $v_s$
are not polynomial with respect to the size $n$ of the system, in general.
Hence, when $n$ is sufficiently large,
the factor $v$ and $v_s$ are not negligible.

However, when the minimum eigenvalue of $\rho_{A,E}^{(n)}$
is greater than $c^{n}$ with a constant $c>0$,
the minimum eigenvalue of 
$\Tr_A (\rho_{A,E}^{(n)})^{1+s}/\Tr  (\rho_{A,E}^{(n)})^{1+s}$
is greater than $c^{(1+s)n}$.
Hence, $\lambda_s$ increases linearly with respect to $n$.
Thus,
when the key generation rate is $R$,
the upper bound (\ref{8-26-13-g2}) for the $L_1$ distinguishability criterion
has the factor of the order $O(\sqrt{n})$ with the term 
$e^{\frac{nsR}{2}-
\frac{s}{2}H_{1+s}^{\rG}(A|E| \rho_{A,E}^{(n)} )+\frac{s}{2}}  $.
For the modified mutual information criterion,
Theorem \ref{Lem12-q} cannot derive a good bound
because the factor $v$ does not behave polynomially.
Instead of Theorem \ref{Lem12-q}, Corollary \ref{c3-29-1-q} gives a better upper bound,
which has the factor of the order $O(n^{3/2})$ with the term 
$e^{\frac{nsR}{2}-
\frac{s}{2}H_{1+s}^{\rG}(A|E| \rho_{A,E}^{(n)} )+\frac{s}{2}}  $.
Hence, 
Theorem \ref{Lem14} and Corollary \ref{c3-29-1-q}
have a larger applicability beyond the i.i.d. case.

\section{Secret key generation with error correction}\Label{s5}
\subsection{Protocol}\Label{s5-1}
Next, 
we apply the above discussions to secret key generation with public communication.
Alice is assumed to have an initial random variable $a\in {\cal A}$, which generates with the probability $p_a$,
and 
Bob and Eve are assumed to have 
their initial quantum states $\rho_{B|a}$ and $\rho_{E|a}$ on their quantum systems ${\cal H}_B$ and  ${\cal H}_E$, respectively.
The task for Alice and Bob is to share a common random variable almost independent of Eve's quantum state by using a public communication.
The quality is evaluated by three quantities:
the size of the final common random variable,
the probability of the disagreement of their final variables (error probability), and
the information leaked to Eve,
which can be quantified by 
the $L_1$ distinguishability criterion or
the modified mutual information criterion 
between Alice's final variables and Eve's state.

In order to construct a protocol for this task,
we assume that the set ${\cal A}$ is a vector space on a finite field $\FF_q$.
Indeed, even if the cardinality $|{\cal A}|$ is not a prime power, 
it become a prime power by adding elements with zero probability.
Hence, we can assume that the cardinality $|{\cal A}|$ is a prime power $q$ without loss of generality.
Then, the secret key agreement can be realized by the following two steps: 
The first is the error correction, and the second is the privacy amplification.
In the error correction, 
Alice and Bob prepare a linear subspace $C_1 \subset \cA$
and the representatives $a(x)$ of all cosets $x \in \cA/C_1$.
Alice sends the coset information $[A] \in \cA/C_1$ to Bob 
in stead of her random variable $A\in \cA$,
and Bob obtain his estimate $\hat{A}$ of $A \in \cA$ 
from his quantum state on ${\cal H}_B$
and $[A] \in \cA/C_1$.
Alice obtains her random variable $A_1:=A-a([A])\in C_1$,
and Bob obtains his random variable $\hat{A}_1:=\hat{A}-a([B])\in C_1$.
In the privacy amplification,
Alice and Bob prepare a common hash function $f$ on $C_1$.
Then, applying the hash function $f$ to the their variables $A_1$ and $\hat{A}_1$,
they obtain their final random variables $f(A_1)$ and $f(\hat{A}_1)$.

Indeed, the above protocol depends on the choice of estimator that gives
the estimate $\hat{A}$ from $[A] \in \cA/C_1$ and 
his random variable $B\in {\cal B}$ (or his quantum state on ${\cal H}_B$).
In the remaining part of this section, 
we give the estimator depending on the setting and 
discuss the performance of this protocol.

\subsection{Error probability}\Label{s5-4}
In the following, we 
give the concrete form of the estimator and 
evaluate the error probability when 
Bob's information is quantum.
In this case, we construct an estimator for $\hat{A}$ in the following way.
For a given code $C_1 \subset \cA$
and a normalized c-q state $\rho_{A,B}=
\sum_{a}P_A(a)|a\rangle \langle a| \otimes \rho_{B|a}$,
our decoder is given as follows:
First, we define the projection:
\begin{align}
P_a:= \{P_A(a) \rho_{B|a} - \frac{q^t}{|\cA|} \rho_B \ge 0\},
\end{align}
where $t$ is the dimension of $C_1$.
When Bob receives the coset $[A]$,
he applies the POVM $\{P_a'\}$: 
\begin{align*}
P_a'  := Q_{[A]}^{-1/2} P_a Q_{[A]}^{-1/2} , \quad
Q_{[A]} :=
\sum_{a\in [A]} P_a.
\end{align*}
Then, Bob chooses the outcome $a$ as the estimate $\hat{A}$. 

Next, we evaluate the performance of the error probability.
Using the operator inequality \cite[Lemma 4.5]{Hayashi-book},
we obtain
\begin{align}
I- P_a' \le
2(I-P_a)+ 4 \sum_{a' \in C_1+a \setminus \{a\} }P_{a'}.
\end{align}
Thus, the error probability 
$P_{e}[\rho_{A,B},C_1]$
is evaluated as
follows.
\begin{align*}
&P_{e}[\rho_{A,B},C_1] 
=
\sum_{a}P_A(a) \Tr \rho_{B|a}  (I-P_a') 
\le 
2\sum_{a}P_A(a) \Tr \rho_{B|a} (I-P_a) 
+ 4 
\sum_{a}P_A(a) \Tr \rho_{B|a} \sum_{a' \in C_1+a \setminus \{a\} }P_{a'} .
\end{align*}
Now, we choose the code $C_1$ from $\varepsilon$-almost universal$_2$
code ensemble $\{C_{\bX}\}$ with the dimension $t$.
Then, the average of the error probability can be evaluated as
\begin{align}
&\rE_{\bX} P_{e}[\rho_{A,B},C_{\bX}] 
\le 
2\sum_{a}P_A(a) \Tr \rho_{B|a} (I-P_a) 
+ 4 
\rE_{\bX} \sum_{a}P_A(a) \Tr \rho_{B|a} \sum_{a' \in C_1+a \setminus \{a\} }P_{a'} \nonumber\\
\le &
2\sum_{a}P_A(a) \Tr \rho_{B|a} (I-P_a) 
+ 4 
\sum_{a}P_A(a) \Tr \rho_{B|a} \varepsilon \frac{q^t}{|\cA|} \sum_{a' \neq a  }P_{a'} \nonumber \\
\le &
2\sum_{a}P_A(a) \Tr \rho_{B|a} (I-P_a) 
+ 4 
\varepsilon \frac{q^t}{|\cA|} \sum_{a}P_A(a) \Tr \rho_{B|a}  \sum_{a'  }P_{a'} 
= 
2\sum_{a} \Tr P_A(a) \rho_{B|a} (I-P_a) 
+ 4 
\varepsilon \frac{q^t}{|\cA|} \sum_{a'} \Tr \rho_B P_{a'} \nonumber \\
\le &
2
\sum_{a} \Tr (P_A(a) \rho_{B|a})^{1-s} \rho_B^{s} 
(\frac{q^t}{|\cA|})^s 
+ 4 
\varepsilon 
\sum_{a'} \Tr (P_A(a') \rho_{B|a'})^{1-s} \rho_B^{s} 
(\frac{q^t}{|\cA|})^s \nonumber \\
=&
(2+ 4 \varepsilon )
\sum_{a'} \Tr (P_A(a') \rho_{B|a'})^{1-s} \rho_B^{s} 
(\frac{q^t}{|\cA|})^s 
=
(2+ 4 \varepsilon )
(\frac{q^t}{|\cA|})^s e^{sH_{1-s}(A|B|\rho_{A,B}) }
.\Label{12-23-10}
\end{align}


\subsection{Leaked information with fixed error correction}\Label{s5-5}
As is mentioned in the previous sections,
we have two criteria for quality of secret random variables.
Given a code $C_1\subset \cA$ and a hash function $f$,
the first criterion is $d_1'(f( A_1) |[A],E|\rho_{A,E})$, and
the second criterion is $I'(f( A_1) |[A],E|\rho_{A,E})$.
Note that the random variable $A$ can be written by the pair of $A_1$ and $[A]$
given in Subsection \ref{s5-1}.
\begin{thm}\Label{L3-20-7}
When
$\{f_{\bX}\}$ is a universal$_2$ ensemble of hash functions from $\cA/C_1$ to $\{1, \ldots, \sM\}$,
the relations
\begin{align}
\rE_{\bX} 
d_1'(f_{\bX}( A_1 ) |[A],E|\rho_{A,E})  
\le &
(4+\sqrt{v'}) (|\cA|/\sL)^{s/2}
e^{\frac{-s}{2}H_{1+s}^{\rG}(A|E|\rho_{A,E})}, \Label{12-23-2-q}\\
\rE_{\bX} 
I'(f_{\bX}( A_1 ) |[A],E|\rho_{A,E}) 
\le &
 \eta (
(4+\sqrt{v'}) (|\cA|/\sL)^{s/2}
e^{\frac{-s}{2}H_{1+s}^{\rG}(A|E|\rho_{A,E})},
\log |{\cal A}|d_E
)
\Label{12-23-3-q2}
\end{align}
hold for $s \in (0,1]$,
where $v'$ is the number of eigenvalues of 
$\Tr_A \rho_{A,E}^{1+s}$,
$v$ is the number of eigenvalues of $\rho_E$, $\sL$ is the amount of sacrifice information $|C_1|/\sM$,
and $\tilde{\sM} := \max\{\sM,d_E\}$.
\end{thm}

\begin{proof}
The relations (\ref{8-26-13-g}) and (\ref{12-20-5-q}) 
guarantee that
\begin{align}
& \rE_{\bX} 
d_1'(f_{\bX}( A_1 ) |[A],E|\rho_{A,E}) 
\le 
(4+\sqrt{v'}) \sM^{s/2} e^{\frac{-s}{2}H_{1+s}^{\rG}(A_1|[A],E|\rho_{A,E})} 
\le 
(4+\sqrt{v'}) \sM^{s/2} (|\cA|/|C_1|)^{s/2}
e^{\frac{-s}{2}H_{1+s}^{\rG}(A_1,[A]|E|\rho_{A,E})} \nonumber \\
= &
(4+\sqrt{v'}) (\sM |\cA|/|C_1|)^{s/2}
e^{\frac{-s}{2}H_{1+s}^{\rG}(A|E|\rho_{A,E})} 
= 
(4+\sqrt{v'}) (|\cA|/\sL)^{s/2}
e^{\frac{-s}{2}H_{1+s}^{\rG}(A|E|\rho_{A,E})} \nonumber
\end{align}
for $s \in (0,1]$,
which implies (\ref{12-23-2-q}).
A simple combination of (\ref{8-26-9-q}) and (\ref{12-23-2-q}) 
yields (\ref{12-23-3-q2}).
\end{proof}

Similarly, Lemma \ref{Lem9-q}, (\ref{8-26-13-g}), and (\ref{12-20-5-q})
 yield the following theorem.
\begin{thm}\Label{L3-20-8}
When $\{f_{\bX}\}$ is an $\varepsilon$-almost dual universal$_2$ ensemble 
of hash functions from $\cA/C_1$ to $\{1, \ldots, \sM\}$,
the relation
\begin{align}
\rE_{\bX} 
d_1'(f_{\bX}( A_1 ) |[A],E|\rho_{A,E}) 
\le & 
(4+\sqrt{\varepsilon v'}) (|\cA|/\sL)^{s/2}
e^{\frac{-s}{2}H_{1+s}^{\rG}(A|E|\rho_{A,E})} ,
\Label{12-23-4-q}\\
 \rE_{\bX} 
I'(f_{\bX}( A_1 ) |[A],E|\rho_{A,E})
\le &
 \eta (
(4+\sqrt{\varepsilon v'}) (|\cA|/\sL)^{s/2}
e^{\frac{-s}{2}H_{1+s}^{\rG}(A|E|\rho_{A,E})},
\log |{\cal A}|d_E
)
\Label{12-23-3-q2b}
\end{align}
holds for $s \in (0,1]$,
where $v'$ is the number of eigenvalues of 
$\Tr_A \rho_{A,E}^{1+s}$,
$v$ is the number of eigenvalues of $\rho_E$,
$\sL$ is the amount of sacrifice information $|C_1|/\sM$,
and
$\tilde{\sM} := \max\{\sM,d_E\}$.
\end{thm}

\begin{rem}
Similar to (\ref{12-23-2-q}) and (\ref{12-23-4-q}),
using (\ref{8-26-13-g2}),
we can show formulas with the logarithmic ratio $\lambda$ between the maximum and minimum eigenvalues
of $\Tr_A \rho_{A,E}^{1+s}$.
\end{rem}

\subsection{Leaked information with randomized error correction code}\Label{s5-5-2}
Next, we evaluate leaked information when the error correcting code $C_1$ is chosen from an $\varepsilon_1$-almost universal$_2$ code ensemble.
In this case, the evaluation for the average of the modified mutual information criterion can be improved to the following way.

\begin{thm}\Label{L12-31-2}
We choose the code $C_1$ from an $\varepsilon_1$-almost universal$_2$
code ensemble $\{C_{\bX}\}$ with the dimension $t$.
Assume that
$\{f_{\bY}\}$ is 
an $\varepsilon_2$-almost dual universal$_2$ ensemble of hash functions from $\cA/C_{\bX}$ to $\{1, \ldots, \sM\}$,
the random variables $\bX$ and $\bY$ are independent of each other,
and $\varepsilon_2 \ge 2$.
\begin{align}
\rE_{\bX,\bY} I' (f_{\bY}( A_1 ) |[A]_{C_{\bX}},E
|\rho_{A,E}' ) 
\le &
2 \eta(
(2 (\frac{|\cA|\sM }{q^t})^{\frac{s}{2-s}}
e^{-\frac{s}{2-s} H_{1+s}(A |E  |\rho_{A,E}'  )}  
,\log \tilde{\sM}+ \frac{v\varepsilon_2}{2\varepsilon_1})
 +\log \varepsilon_1 .
\Label{12-23-5-q-2} 
\end{align}
for $s \in (0,1]$,
where $v$ is the number of eigenvalues of $\rho_E$
and $\tilde{\sM}:=\max\{\sM,d_E\}$.
Similarly, when 
$\{f_{\bY}\}$ is 
universal$_2$ ensemble of hash functions from $\cA/C_{\bX}$ to $\{1, \ldots, \sM\}$,
\begin{align}
\rE_{\bX,\bY} I' (f_{\bY}( A_1 ) |[A]_{C_{\bX}},E
|\rho_{A,E}' ) 
\le &
2 \eta(
(2 (\frac{|\cA|\sM }{q^t})^{\frac{s}{2-s}}
e^{-\frac{s}{2-s} H_{1+s}(A |E  |\rho_{A,E}'  )}  
,\log \tilde{\sM}+ \frac{v}{4\varepsilon_1})
+\log \varepsilon_1 .
\Label{12-23-5-4}
\end{align}
\end{thm}
\begin{proof}
We choose a sub cq-state $\rho_{A,E}'
= \sum_{a}|a \rangle \langle a|\otimes \rho_{E|a}'$ 
such that 
${\rho'}_{E}\le \rho_{E} $ and ${\rho'}_{A}\le \rho_{A} $.
Due to (\ref{Lem6-3-q-eq2}), 
we obtain
\begin{align*}
& \rE_{\bY} e^{-\overline{H}_2(f_{\bY}( A_1 ) |[A]_{C_{\bX}},E
|\rho_{A,E}' 
\|\rho_{\mix,[A]_{C_{\bX}}} \otimes \rho_E )} \\
\le &
\varepsilon_2
(1-\frac{1}{\sM})
e^{-\overline{H}_2(A_1| [A]_{C_{\bX}},E  |\rho_{A,E}' 
\|\rho_{\mix,[A]_{C_{\bX}}} \otimes \rho_E )} 
+
\frac{1}{\sM}
e^{\underline{\psi}(1|{\rho'}_{[A]_{C_{\bX}},E} 
\|\rho_{\mix,[A]_{C_{\bX}}} \otimes \rho_E)}
\\
= &
\varepsilon_2
(1-\frac{1}{\sM})
\frac{|\cA|}{q^t}
e^{-\overline{H}_2(A_1,[A]_{C_{\bX}} |E  |\rho_{A,E}' \| \rho_E )} 
+
\frac{1}{\sM}
e^{\underline{\psi}(1|{\rho'}_{[A]_{C_{\bX}},E} 
\|\rho_{\mix,[A]_{C_{\bX}}} \otimes \rho_E)}
\\
= &
\varepsilon_2 (\frac{|\cA|}{q^t})
e^{-\overline{H}_2(A |E  |\rho_{A,E}' \| \rho_E )} 
+
\frac{1}{\sM}
\frac{|\cA|}{q^t}
(
e^{-\overline{H}_2([A]_{C_{\bX}} |E  |\rho_{A,E}' \| \rho_E )} 
-
\varepsilon_2 e^{-\overline{H}_2(A |E  |\rho_{A,E}' \| \rho_E )} ) .
\end{align*}
Since the matrix $\rho_{A,E}'$ satisfies 
\begin{align*}
& 
e^{-\overline{H}_2([A]_{C_{\bX}} |E  |\rho_{A,E}' \| \rho_E )} 
-
e^{-\overline{H}_2(A |E  |\rho_{A,E}' \| \rho_E )} 
=
\sum_{a} 
\Tr_E
{\rho'}_{E|a}
\rho_{E}^{-\frac{1}{2}}
( \sum_{a'\in  C_{\bX}+a \setminus \{a\} } 
{\rho'}_{E|(a')})
\rho_{E}^{-\frac{1}{2}}
\end{align*}
and
\begin{align*}
& 
\rE_{\bX}
\sum_{a} 
\Tr_E {\rho'}_{E|a}
\rho_{E}^{-\frac{1}{2}}
( \sum_{a'\in  C_{\bX}+a \setminus \{a\} }
{\rho'}_{E|(a')} )
\rho_{E}^{-\frac{1}{2}}
\le 
\sum_{a} 
\Tr_E 
{\rho'}_{E|a}
\rho_{E}^{-\frac{1}{2}}
( \varepsilon_1 
\frac{q^t}{|\cA|}
\sum_{a'\neq a}
{\rho'}_{E|(a')} )
\rho_{E}^{-\frac{1}{2}}
\\
\le &
\varepsilon_1
\frac{q^t}{|\cA|}
\Tr_E
\sum_{a} {\rho'}_{E|a}
\rho_{E}^{-\frac{1}{2}}
(  \sum_{a'}
{\rho'}_{E|(a')}
)
\rho_{E}^{-\frac{1}{2}}
= 
\varepsilon_1
\frac{q^t}{|\cA|}
e^{\underline{\psi}(1|{\rho'}_{E} \| \rho_E)}
\le 
\varepsilon_1
\frac{q^t}{|\cA|},
\end{align*}
we have
\begin{align*}
& \rE_{\bX}
e^{-\overline{H}_2([A]_{C_{\bX}} |E  |\rho_{A,E}' \| \rho_E )} 
-
\varepsilon_2 e^{-\overline{H}_2(A |E  |\rho_{A,E}' \| \rho_E )} 
\le 
\rE_{\bX}
e^{-\overline{H}_2([A]_{C_{\bX}} |E  |\rho_{A,E}' \| \rho_E )} 
-
e^{-\overline{H}_2(A |E  |\rho_{A,E}' \| \rho_E )} 
\le  
\varepsilon_1
\frac{q^t}{|\cA|},
\end{align*}
where the first inequality follows from $\varepsilon_2 \ge 1$.

Hence,
we obtain
\begin{align*}
& \rE_{\bX,\bY} e^{-\overline{H}_2(f_{\bY}( A_1 ) |[A]_{C_{\bX}},E
|\rho_{A,E}' 
\|\rho_{\mix,[A]_{C_{\bX}}} \otimes \rho_E )} 
\le 
\varepsilon_2 (\frac{|\cA|}{q^t})
e^{-\overline{H}_2(A |E  |\rho_{A,E}' \| \rho_E )} 
+
\frac{1}{\sM} \varepsilon_1 
=
\frac{1}{\sM}
\varepsilon_1 
(
1+
\frac{\varepsilon_2}{\varepsilon_1} 
\frac{|\cA|}{q^t}
\sM e^{-\overline{H}_2(A |E  |\rho_{A,E}' \| \rho_E )} 
).
\end{align*}
Applying Jensen's inequality to $x \mapsto \log x$,
we obtain
\begin{align*}
 \rE_{\bX,\bY} -\overline{H}_2(f_{\bY}( A_1 ) |[A]_{C_{\bX}},E
|\rho_{A,E}' 
\|\rho_{\mix,[A]_{C_{\bX}}} \otimes \rho_E ) 
\le &
-\log \sM+\log \varepsilon_1 
+\log 
(
1+
\frac{\varepsilon_2}{\varepsilon_1} 
\frac{|\cA|}{q^t}
\sM e^{-\overline{H}_2(A |E  |\rho_{A,E}' \| \rho_E )} ).
\end{align*}
Using (\ref{12-6-5-q}), (\ref{12-31-5}), and (\ref{8-15-14}),
we obtain
\begin{align}
& I' (f_{\bY}( A_1 ) |[A]_{C_{\bX}},E
|\rho_{A,E} ) 
= 
\log \sM-
 H (f_{\bY}( A_1 ) |[A]_{C_{\bX}},E
|\rho_{A,E} ) \nonumber \\
\le &
2 \eta(\|\rho_{A,E}-\rho_{A,E}'\|_1,\log \tilde{\sM} ) 
+
\log \sM-
 H (f_{\bY}( A_1 ) |[A]_{C_{\bX}},E
|\rho_{A,E}' ) \nonumber \\
\le &
2 \eta(\|\rho_{A,E}-\rho_{A,E}'\|_1,\log \tilde{\sM} ) 
+\log \sM-
 H (f_{\bY}( A_1 ) |[A]_{C_{\bX}},E
|\rho_{A,E}' \| \rho_{\mix,[A]_{C_{\bX}}} \otimes \rho_E ) \nonumber \\
\le &
2 \eta (\|\rho_{A,E}-\rho_{A,E}'\|_1,\log \tilde{\sM}) 
+
\log \sM-
 \overline{H}_2 (f_{\bY}( A_1 ) |[A]_{C_{\bX}},E
|\rho_{A,E}' \| \rho_{\mix,[A]_{C_{\bX}}} \otimes \rho_E ) .
\end{align}
Hence, we obtain
\begin{align}
& \rE_{\bX,\bY} I' (f_{\bY}( A_1 ) |[A]_{C_{\bX}},E
|\rho_{A,E} ) 
\le 
2 \eta(\|\rho_{A,E}-\rho_{A,E}'\|_1,\log \tilde{\sM} ) 
+\log \varepsilon_1 
+\log 
( 1
+ \frac{\varepsilon_2}{\varepsilon_1} 
\frac{|\cA|}{q^t}
\sM e^{-\overline{H}_2(A |E  |\rho_{A,E}' \| \rho_E )} ) \nonumber \\
\le &
2 \eta (\|\rho_{A,E}-\rho_{A,E}'\|_1,\log \tilde{\sM} ) 
+
\log \varepsilon_1 
+ \frac{\varepsilon_2}{\varepsilon_1} 
\frac{|\cA|}{q^t}
\sM e^{-\overline{H}_2(A |E  |\rho_{A,E}' \| \rho_E )}  .
\end{align}
Applying the same discussion as the proof of Theorem \ref{Lem12-q},
we obtain
\begin{align}
& \rE_{\bX,\bY} I' (f_{\bY}( A_1 ) |[A]_{C_{\bX}},E
|\rho_{A,E}' ) 
\le 
2 \eta(
(2 (\frac{|\cA|\sM }{q^t})^{\frac{s}{2-s}}
e^{-\frac{s}{2-s} H_{1+s}(A |E  |\rho_{A,E}'  )}  
,\log \tilde{\sM}+ \frac{v\varepsilon_2}{4\varepsilon_1})
+\log \varepsilon_1 .
\end{align}
\end{proof}

\subsection{Asymptotic analysis}\Label{s5-6}
Next, we consider the case when 
the c-q state is given as the $n$-fold independent and identical extension 
$\rho_{A,B,E}^{\otimes n}$
of a c-q normalized state $\rho_{A,B,E}$,
where $\cA$ is $\FF_q$.
Now, we fix codes $C_{1,n}$ in $\FF_q^n$ 
with the dimension $\lfloor n \frac{R_1}{\log q}\rfloor $.
Then, we obtain the following theorem.
\begin{thm}\Label{t3-20-20}
When $\{f_{\bX}\}$ is a universal$_2$ ensemble of hash functions
from $\FF_q^n/C_{1,n}$ to $\FF_q^{\lfloor n \frac{1-R_1-R_2}{\log q}\rfloor } $,
the relations
\begin{align}
\liminf_{n\to \infty}\frac{-1}{n}\log 
\rE_{\bX} d_1'(f_{\bX}( A_{1,n} )|[A_n],E_n|
\rho_{A,E}^{\otimes n}) 
\ge & 
\max_{0\le s \le 1} \frac{s}{2} (R_2 -\log q ) +
\frac{s}{2}H_{1+s}^{\rG}(A|E|\rho_{A,E}) 
= e_{\rG,\rq}(\rho_{A,E}|\log q -R_2),
\Label{12-23-6-q}\\
\liminf_{n\to \infty}\frac{-1}{n}\log 
\rE_{\bX} 
I'(f_{\bX}( A_{1,n} ) |[A_n],E_n|\rho_{A,E}^{\otimes n}) 
\ge &
e_{\rG,\rq}(\rho_{A,E}|\log q -R_2)
\Label{12-23-7-q}
\end{align} 
hold.
\end{thm}
\begin{proof}
(\ref{12-23-2-q}) and (\ref{12-23-3-q2}) yield 
(\ref{12-23-6-q}) and (\ref{12-23-7-q}), respectively.
\end{proof}

Similarly, we have the following theorem.
\begin{thm}\Label{t3-20-21}
When $P(n)$ is an arbitrary polynomial
and
$\{f_{\bX}\}$ is a $P(n)$-almost dual universal$_2$ ensemble of hash functions
from $\FF_q^n/C_{1,n}$ to $\FF_q^{\lfloor n \frac{1-R_1-R_2}{\log q}\rfloor } $,
the relations 
(\ref{12-23-6-q}) and (\ref{12-23-7-q}) hold.
\end{thm}

\begin{proof}
(\ref{12-23-4-q}) and (\ref{12-23-3-q2b}) yield 
Inequalities (\ref{12-23-6-q}) and (\ref{12-23-7-q}), respectively.
\end{proof}

Next, we consider the case when the error correcting code is chosen randomly.
In this case, the exponential decreasing rate for 
$I'(f_{\bX}( A_{1,n} ) |[A_n],E_n|P_{A,E}^n)$
can be improved as follows.

\begin{thm}\Label{p3-16-1c}
For independent random variables $\bX,\bY$,
we assume that
the code ensemble $\{C_{\bX}\}$ 
with the dimension $\lfloor n \frac{R_1}{\log q}\rfloor $
is universal$_2$ 
and 
$\{f_{\bY}\}$ is 
universal$_2$ ensemble of hash functions from $\FF_q^n /C_{\bX}$ to 
$\FF_q^{\lfloor n \frac{1-R_1-R_2}{\log q}\rfloor } $,
the relations (\ref{12-23-6-q}), (\ref{12-23-3-q2}), and 
\begin{align}
\liminf_{n\to \infty}\frac{-1}{n}\log \rE_{\bX} 
P_{e}[\rho_{A,B}^{\otimes n},C_{\bX}] 
\ge &
\max_{0\le s \le 1} s(\log q -R_1 )-s H_{1-s}(A|B|\rho_{A,B}),
 \Label{3-16-11b}\\
 \liminf_{n\to \infty}\frac{-1}{n}\log 
\rE_{\bX,\bY} 
I'(f_{\bY}( A_{1,n} ) |[A_n]_{C_{\bX}},E_n|\rho_{A,E}^{\otimes n}) 
\ge &
e_{\rH,\rq}(P_{A,E}|\log q -R_2)
\Label{12-31-2-c}
\end{align}
hold.
\end{thm}

\begin{proof}
Theorem \ref{L12-31-2} implies that
\begin{align*}
 \liminf_{n\to \infty}\frac{-1}{n}\log 
\rE_{\bX,\bY} 
I'(f_{\bY}( A_{1,n} ) |[A_n]_{C_{\bX}},E_n|\rho_{A,E}^{\otimes n}) 
\ge 
\max_{0\le s \le 1} \frac{s}{2-s} (R_2 -\log q   + H_{1+s}(A|E|\rho_{A,E}) ) 
=& 
e_{\rH,\rq}(P_{A,E}|\log q -R_2),
\end{align*}
which yields  (\ref{12-31-2-c}).
Due to (\ref{12-23-10}), the error probability can be bounded as
\begin{align*}
& \rE_{\bX_n} P_{e}[\rho_{A,B}^{\otimes n},C_{\bX_n}] 
\le
P(n) e^{n (s(R_1 -\log q )+s H_{1-s}(A|B|\rho_{A,B}))}
\end{align*}
for $s\in [0,1]$, which implies (\ref{3-16-11b}).
\end{proof}

Similarly, we obtain the following theorem.
\begin{thm}\Label{t3-20-11b}
For an arbitrary polynomial $P(n)$ and the independent random variables $\bX,\bY$,
we assume that
the code ensemble $\{C_{\bX}\}$ with the dimension $\lfloor n \frac{R_1}{\log q}\rfloor $ is universal$_2$ 
and 
$\{f_{\bY}\}$ is 
a $P(n)$-almost dual universal$_2$ ensemble of hash functions from $\FF_q^n /C_{\bX}$ to 
$\FF_q^{\lfloor n \frac{1-R_1-R_2}{\log q}\rfloor } $,
the relations (\ref{12-23-6-q}), 
(\ref{12-23-3-q2}),
(\ref{3-16-11b}), and (\ref{12-31-2-c}) hold.
\end{thm}
For a comparison between two exponents
$e_{\rG,\rq}(\rho_{A,E}|R)$ and $e_{\rG,\rq}(\rho_{A,E}|R) $,
see Lemma \ref{8-29-10}.

\section{Application to generalized Pauli channel}\Label{s10}
\subsection{General case}
In order to apply the above result to quantum key distribution,
we treat the quantum state generated by transmission by 
a generalized Pauli channel in the $p$-dimensional system ${\cal H}$.
First, we define 
the discrete Weyl-Heisenberg representation $W$ for $\FF_p^2$:
\begin{align*}
\sX &:= \sum_{j=0}^{p-1} |j+1 \rangle \langle j| , \quad
\sZ := \sum_{j=0}^{p-1} \omega^j |j \rangle \langle j| ,\quad
\sW(x,z) := \sX^x \sZ^z,
\end{align*}
where $\omega$ is the root of the unity with the order $p$.
Using this representation and a probability distribution $P_{XZ}$ on $\FF_p^2$, 
we can define the generalized Pauli channel:
\begin{align*}
{\cal E}_{P}(\rho):= \sum_{(x,z)\in \FF_p^2}
P_{XZ}(x,z) \sW(x,z) \rho \sW(x,z)^\dagger .
\end{align*}
In the following, we assume that the eavesdropper can access all of the environment of the channel ${\cal E}_{P}$.
When the state $|j\rangle$ is input to the channel ${\cal E}_{P}$, 
the environment system is spanned 
by the basis $\{|x,z \rangle_E \}$. 
Then, 
the state $\rho_{E|j}$ of the environment (Eve's state)
and Bob's state $\rho_{B|j}$
are given as
\begin{align*}
\rho_{E|j} &=\sum_{z=0}^{p-1}
P_Z(z) |j,z:P_{XZ}\rangle \langle j,z:P_{XZ}|,
\quad 
|j,z:P_{XZ}\rangle := \sum_{x=0}^{p-1} \omega^{j x}
\sqrt{P_{X|Z}(x|z)} |x,z \rangle_E \\
\rho_{B|j} &=\sum_{x=0}^{p-1} 
P_X(x) |j +x\rangle_B ~_B\langle j+x|.
\end{align*}
Thus, the relation
\begin{align*}
& \sum_{a=0}^{p-1}
|a,z:P_{XZ}\rangle \langle j,z:P_{XZ}| 
=
p \sum_{x} 
P_{X|Z}(x|z) |x,z\rangle_E ~_E \langle x,z|
\end{align*}
holds.
Hence,
\begin{align}
\rho_E= \sum_{x,z}P_{X,Z}(x,z) 
|x,z\rangle_E ~_E \langle x,z|.
\end{align}
Then,
we obtain the following state after 
the quantum state transmission via the generalized Pauli channel.
\begin{align*}
\rho_{A,B,E}:=
\sum_{j=0}^{p-1}
\frac{1}{p}
|j\rangle \langle j|\otimes \rho_{B|j}
\otimes \rho_{E|j}.
\end{align*}
In this setting, the joint state $\rho_{A,B}$ is
classical, we can apply the classical theory for error probability.
Since $P_{A,B}(a,b)=
\sum_{a}\frac{1}{p}P_X(b-a)$,
we have
\begin{align*}
&e^{-s H_{\frac{1}{1-s}}^{\rG}(A|B|\rho_{A,B})}
=
\sum_{e} \frac{1}{p} 
(\sum_a P_X(b-a)^{1/(1-s)})^{1-s} 
= 
(\sum_x P_X(x)^{1/(1-s)})^{1-s} 
= e^{(1-s) \frac{-s}{1-s}H_{\frac{1}{1-s}}(X|P_X)} 
= e^{-s H_{\frac{1}{1-s}}(X|P_X)} .
\end{align*}

Now, we choose the rate $R_1$ of size of code $C_1$.
When $\{C_{\bX_n}\}$ is the $P(n)$-almost universal$_2$ code ensemble in $\FF_q^n$ 
with the dimension $\lfloor n \frac{R_1}{\log p}\rfloor $,
due to \cite[(243)]{H-arxiv}, the decoding error probability can be bounded as
\begin{align*}
\rE_{\bX_n} P_{e}[\rho_{A,B}^{\otimes n},C_{\bX_n}] 
\le &
P(n) e^{n (s(R_1 -\log q )-s H_{\frac{1}{1+s}}(X|P_X) )}.
\end{align*}
That is,
\begin{align*}
\liminf_{n\to \infty}\frac{-1}{n}\log \rE_{\bX_n} 
P_{e}[\rho_{A,B}^{\otimes n},C_{\bX_n}] 
\ge &
\max_{0\le s \le 1} s(\log q -R_1 )+s H_{\frac{1}{1+s}}(X|P_X).
\end{align*}

Next, we treat the leaked information.
In the following discussion, we fix codes $C_{1,n}$ in 
$\FF_p^n$.
Since $\rho_{A,E}=
\sum_{a}\frac{1}{q}|a\rangle \langle a| \otimes \rho_{E|a}$,
we have
\begin{align}
& e^{-s H_{\frac{1}{1-s}}^{\rG}(A|E|\rho_{A,E})} 
=
\Tr_E 
(\Tr_A
(\sum_{a}\frac{1}{p}|a\rangle \langle a| \otimes \rho_{E|a}
)^\frac{1}{1-s})^{1-s} 
= 
\frac{1}{p}\Tr_E 
(\sum_{a} (\rho_{E|a})^\frac{1}{1-s})^{1-s}\nonumber  \\
=& 
\frac{1}{p}\Tr_E 
(\sum_{a} 
\sum_{z=0}^{p-1}
P_Z(z)^\frac{1}{1-s}
 |a,z:P_{XZ}\rangle \langle a,z:P_{XZ}|
)^{1-s} 
= 
\frac{1}{p}\Tr_E 
(
\sum_{z=0}^{p-1}
P_Z(z)^\frac{1}{1-s}
\sum_{a} 
 |a,z:P_{XZ}\rangle \langle a,z:P_{XZ}|
)^{1-s}\nonumber  \\
=& 
\frac{1}{p}\Tr_E 
(
\sum_{z=0}^{p-1}
P_Z(z)^\frac{1}{1-s}
p \sum_{x} 
P_{X|Z}(x|z) |x,z\rangle_E ~_E \langle x,z|
)^{1-s}
=
p^{-s} 
\Tr_E 
\sum_{z=0}^{p-1}
\sum_{x} 
P_Z(z)
P_{X|Z}(x|z)^{1-s} |x,z\rangle_E ~_E \langle x,z| 
\nonumber \\
=&
p^{-s} 
e^{s H_{1-s}(X|Z|P_{X,Z})}\Label{1-3-5}
\end{align}
and
\begin{align}
&e^{-sH_{1+s}(A|E|\rho_{A,E})} 
=
\Tr
(\sum_{a}\frac{1}{p}|a\rangle \langle a| \otimes \rho_{E|a}
)^{1+s} \rho_E^{-s} 
=
\frac{1}{p^{1+s}}
\sum_a 
\rho_{E|a}^{1+s}\rho_E^{-s} \nonumber \\
=&
\frac{1}{p^{1+s}}
\sum_a 
\sum_z P_Z(z)
\Tr |j,z :P_{XZ}\rangle \langle j,z :P_{XZ}|^{1+s} 
\cdot (\sum_{x}P_{X|Z}(x|z) |x,z\rangle_E ~_E \langle x,z|)^{-s} \nonumber \\
=&
\frac{1}{p^{1+s}}
\sum_a 
\sum_z P_Z(z)
\sum_x P_{X|Z}(x|z)^{1-s} 
=
\frac{1}{p^{s}}
\sum_z P_Z(z)
\sum_x P_{X|Z}(x|z)^{1-s} 
=
p^{-s} 
e^{s H_{1-s}(X|Z|P_{X,Z})}.\Label{1-3-6}
\end{align}
That is, we have
\begin{align}
H_{1+s}(A|E|\rho_{A,E})
=H_{\frac{1}{1-s}}^{\rG}(A|E|\rho_{A,E})
=\log p -  H_{1-s}(X|Z|P_{X,Z}).
\Label{8-30-b}
\end{align}
Now, we consider the case with randomized error correction.
Given a sequence of fixed codes $C_{1,n}$,
we focus on a sequence of ensembles of hash functions of
$\FF_p^n/C_{1,n}$ with the rate $R_2$ of sacrifice information
(i.e., with the sacrifice bit length $L=n R_2$).

In this case, the numbers of eigenvalues of 
$\rho_E^{\otimes n}$ and 
$\Tr_A (\rho_{A,E}^{\otimes n})^{1+s}$ 
are less than $ (n+1)^{(p^2-1)}$.
Thus,
when the code ensemble $\{C_{\bX}\}$ with the dimension $\lfloor n \frac{R_1}{\log q}\rfloor $ is universal$_2$ 
and 
$\{f_{\bY}\}$ is 
a $\varepsilon$-almost dual universal$_2$ ensemble of hash functions from $\FF_q^n /C_{\bX}$ to 
$\FF_q^{\lfloor n \frac{1-R_1-R_2}{\log q}\rfloor } $,
(\ref{12-23-4-q}), (\ref{12-23-3-q2b}), and (\ref{12-23-5-q-2}) yield that
\begin{align}
\rE_{\bX,\bY} d_1'(f_{\bY}( A_{1,n} ) |[A_n]_{C_{\bX}},E_n|
\rho_{A,E}^{\otimes n}) 
\le &
(4+(n+1)^{(p^2-1)/2} \sqrt{\varepsilon})
e^{ n \frac{s}{2} (-R_2 +H_{\frac{1}{1+s}}(X|Z|P_{X,Z}))} 
\Label{12-24-2}, \\ 
\rE_{\bX,\bY} I'(f_{\bY}( A_{1,n} ) |[A_n]_{C_{\bX}},E_n|
\rho_{A,E}^{\otimes n}) 
\le &
\eta
\bigl( 
(4+(n+1)^{(p^2-1)/2} \sqrt{\varepsilon})
e^{ n \frac{s}{2} (-R_2 +H_{\frac{1}{1+s}}(X|Z|P_{X,Z}))} ,
 n \log p \bigr),
\Label{12-24-1b} \\
\rE_{\bX,\bY} I'(f_{\bY}( A_{1,n} ) |[A_n]_{C_{\bX}},E_n|
\rho_{A,E}^{\otimes n}) 
\le &
2 \eta
\bigl( 2 e^{ n \frac{s}{2-s} (-R_2 +H_{1-s}(X|Z|P_{X,Z}))} , 
{\tiny{\frac{\varepsilon (n+1)^{(p^2-1)}}{4}}} 
+ n \log p\bigr).
\Label{12-24-1}
\end{align} 
In particular, when
$\{f_{\bY}\}$ is a universal$_2$ ensemble of hash functions,
due to (\ref{12-23-2-q}), (\ref{12-23-3-q2}), and (\ref{12-23-5-4}),
the real number $\varepsilon$ can be replaced by $1$ 
in the above inequalities.

Here, we need a remark for (\ref{12-24-1b}).
The second input of the function $\eta$ in (\ref{12-24-1b})
is $n \log p $ not $2n \log p$.
In this case, the state $\rho_A$ is the uniform distribution, 
we can use (\ref{8-26-9-a-q}) instead of (\ref{8-26-9-q}).
Hence, we can replace $2n \log p$ by $n \log p $.

The exponents 
$e_{\rG,\rq}(\rho_{A,E}|\log p -R_2)$
and 
$e_{\rH,\rq}(\rho_{A,E}|\log p -R_2)$
are calculated as 
\begin{align}
e_{\rG,\rq}(\rho_{A,E}|\log p -R_2)
=& \max_{0\le s \le 1} \frac{s}{2} 
(R_2 -H_{1-\frac{s}{1+s}}(X|Z|P_{X,Z})),
\Label{12-23-6-qq2}\\
e_{\rH,\rq}(\rho_{A,E}|\log p -R_2)
=&\max_{0\le s \le 1} 
\frac{s}{2-s} (R_2 -H_{1-s}(X|Z|P_{X,Z})) 
=
\max_{0\le t \le 1} 
\frac{t}{2} (R_2 -H_{1-\frac{2t}{2+t}}(X|Z|P_{X,Z})) 
\Label{12-23-7-qq2},
\end{align} 
where $\frac{s}{2-s}=\frac{t}{2}$.
In fact,
our bound in (\ref{12-23-6-qq}) is the same as 
the bound obtained by the recent paper \cite[(60)]{Tsuru} via the phase error correction approach.
This fact seems the goodness of our bound and our approach. 

Since $\frac{s}{1+s} \le \frac{2s}{2+s}$ for $s \in [0,1]$, 
Lemma \ref{L22-1} guarantees that $
H_{1-\frac{2s}{2+s}}(X|Z|P_{X,Z})
\ge H_{1-\frac{s}{1+s}}(X|Z|P_{X,Z})$, which implies
$e_{\rH,\rq}(\rho_{A,E}|\log p -R_2)
\le e_{\rG,\rq}(\rho_{A,E}|\log p -R_2)$.
That is, (\ref{12-24-1b}) gives a better exponent than (\ref{12-24-1}).
Since the relation (\ref{ineq-7-23-2}) holds due to (\ref{8-30-b}),
this case can be regard as a special case of Lemma \ref{8-29-10}.
Thus, we obtain
\begin{align}
\liminf_{n\to \infty}\frac{-1}{n}\log 
\rE_{\bX,\bY} d_1'(f_{\bY}( A_{1,n} ) |[A_n]_{C_{\bX}},E_n|
\rho_{A,E}^{\otimes n}) 
& \ge 
e_{\rG,\rq}(\rho_{A,E}|\log p -R_2)
\Label{12-23-6-qq}\\
\liminf_{n\to \infty}\frac{-1}{n}\log 
\rE_{\bX,\bY} I'(f_{\bY}( A_{1,n} ) |[A_n]_{C_{\bX}},E_n|
\rho_{A,E}^{\otimes n}) 
& \ge 
e_{\rG,\rq}(\rho_{A,E}|\log p -R_2)
\Label{12-23-7-qq}.
\end{align} 
However, there still exists a possibility that
the evaluation (\ref{12-24-1}) gives a better evaluation than 
(\ref{12-24-1b})
in the finite length setting.

\subsection{Independent case}\Label{s10-2}
Next, we consider the case when the two random variables $X$ and $Z$ are independent,
Eve's state $\rho_{E|j}$ has the following form:
\begin{align*}
\rho_{E|j} =
 |j:P_{X}\rangle \langle j:P_{X}|
\otimes \sum_{z=0}^{p-1}P_Z(z) |z\rangle_Z~_Z\langle z| ,
\quad
 |j:P_{X}\rangle :=& \sum_{x=0}^{p-1} \omega^{j x}
\sqrt{P_{X}(x)} |x \rangle_{X} .
\end{align*}
In this case, the system spanned by $\{|z\rangle_Z\}$
has no correlation with $j$,
and only 
the system spanned by $\{|x\rangle_X\}$
has correlation with $j$.
So, we can replace $\rho_{E|j}$ by the following way:
\begin{align*}
\rho_{E|j} =
 |j:P_{X}\rangle \langle j:P_{X}|.
\end{align*}
In this case,
the numbers of eigenvalues of $\rho_E$ and 
$\Tr_A \rho_{A,E}^{1+s}$ are less than $p$.
Hence, 
the numbers of eigenvalues of 
$\rho_E^{\otimes n}$ and 
$\Tr_A (\rho_{A,E}^{\otimes n})^{1+s}$ 
are less than $ (n+1)^{(p-1)}$.
When we choose $\varepsilon =1$ for simplicity, 
the inequalities  
(\ref{12-24-2}), (\ref{12-24-1b}), and (\ref{12-24-1}) can be 
simplified to
\begin{align}
\rE_{\bX,\bY} d_1'(f_{\bY}( A_{1,n} ) |[A_n]_{C_{\bX}},E_n|
\rho_{A,E}^{\otimes n}) 
\le &
(4+(n+1)^{(p-1)/2} 
)
e^{ n \frac{s}{2} (-R_2 +H_{\frac{1}{1+s}}(X|P_{X}))} ,
\Label{12-24-2-q} \\ 
\rE_{\bX,\bY} I'(f_{\bY}( A_{1,n} )|[A_n]_{C_{\bX}},E_n|\rho_{A,E}^{\otimes n}) 
\le &
\eta ((4+(n+1)^{(p-1)/2} 
)
e^{ n \frac{s}{2} (-R_2 +H_{\frac{1}{1+s}}(X|P_{X}))} ,
n (\log p) ),\Label{1-2-1} \\
\rE_{\bX,\bY} I'(f_{\bY}( A_{1,n} )|[A_n]_{C_{\bX}},E_n|\rho_{A,E}^{\otimes n}) 
\le &
2 \eta
( 2 e^{ n \frac{s}{2-s} (-R_2 +H_{1-s}(X|P_{X}))} ,
(n+1)^{(p-1)}/4 + n \log p ).\Label{12-24-1-q}
\end{align} 
Hence, we obtain
\begin{align}
& \liminf_{n \to \infty}
-\frac{1}{n}\log 
\rE_{\bX,\bY} d_1'(f_{\bY}( A_{1,n} ) |[A_n]_{C_{\bX}},E_n|
\rho_{A,E}^{\otimes n}) 
\ge 
e_{\rG,\rq}(\rho_{A,E}|\log p- R_2)
\nonumber\\
=&\max_{0 \le t \le 1}
\frac{t}{2} (R_2 -H_{\frac{1}{1+t}}(X|P_{X}))  
= \max_{0 \le s \le 1/2}
\frac{s R_2-sH_{1-s}(P_X) }{2(1-s)}
\Label{8-18-1}.
\end{align} 

Here, we compare the evaluations (\ref{1-2-1}) and (\ref{12-24-1-q}).
As is explained in the previous subsection,
the exponent of (\ref{1-2-1}) is better than (\ref{12-24-1-q}).
This relation can be numerically checked in 
Fig. \ref{f1} with the parameters $p=2$, $P_X(0)=0.9$, $P_X(1)=0.1$,
and $R \in (0.53,0.58)$.
However, in the case of a finite $n$,
$-\frac{1}{n}\log \min_{0\le s \le 1} 
\hbox{(RHS of (\ref{1-2-1}))} $
is not necessarily larger than 
$-\frac{1}{n}\log \min_{0\le s \le 1} 
\hbox{(RHS of (\ref{12-24-1-q}))}$.
The relation between these two quantities is also numerically 
demonstrated in Fig. \ref{f1} with the same parameters 
when $n=10,000$.
This numerical result suggests that
the exponents can not necessarily decide 
the order of advantages with the finite size $n$
when $n$ is not sufficiently large.

\begin{figure}[htbp]
\begin{center}
\scalebox{0.7}{\includegraphics[scale=1]{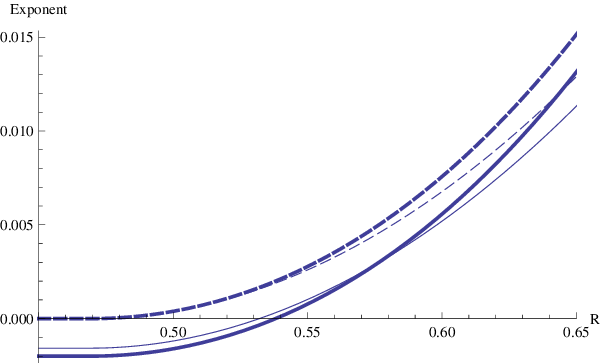}}
\end{center}
\caption{
Lower bounds of exponent.
Thick dashed line: 
$e_{\rG,\rq}(\rho_{A,E}|\log p -R_2)
= \max_{0\le s \le 1} \frac{s}{2} 
(R_2 -H_{1-\frac{s}{1+s}}(X|P_{X}))$
Normal dashed line: 
$e_{\rH,\rq}(\rho_{A,E}|\log p -R_2)
=\max_{0\le s \le 1} 
\frac{s}{2-s} (R_2 -H_{1-s}(X|P_{X})) 
=
\max_{0\le t \le 1} 
\frac{t}{2} (R_2 -H_{1-\frac{2t}{2+t}}(X|P_{X})) $
Thick line: $-\frac{1}{n}\log \min_{0\le s \le 1}$ (RHS of (\ref{1-2-1})),
Normal line: $-\frac{1}{n}\log \min_{0\le s \le 1}$ (RHS of (\ref{12-24-1-q})) 
with $n=10,000$, $p=2$, $P_X(0)=0.9$, $P_X(1)=0.1$.}
\Label{f1}
\end{figure}%

Next, we consider the case when there is no error in $Z$ basis.
In this case, it is sufficient to apply only privacy amplification.
Hence, we evaluate the upper bounds
$\Delta_{d,2}(e^{nR}, \varepsilon_1|\rho_{A,E}^{\otimes n}) $
as follows.

\begin{lem}\Label{L8-16-6}
When $p=2$ and $\rho_{A,E}=\sum_{x \in \FF_2} \frac{1}{2} |x\rangle \langle x|\otimes |x :P_X\rangle \langle x :P_X|$,
we have
\begin{align}
\lim_{n \to \infty}\frac{-1}{n} \log 
\Delta_{d,2}(e^{nR}, \varepsilon_1|\rho_{A,E}^{\otimes n}) 
&=e_{\rG,\rq}(\rho_{A,E}|R)
=\max_{0 \le s \le 1/2}
\frac{-sH_{1-s}(P_X) +s (\log 2-R)}{2(1-s)}.\Label{8-16-20}
\end{align}
\end{lem}

Lemma \ref{L8-16-6} is proven in Appendix \ref{aL8-16-6}.
 
\section{Conclusion}
We have derived upper bounds for the leaked information in the modified mutual information
criterion and the $L_1$ distinguishability criterion
in the quantum case
when we apply 
a family of universal$_2$ hash functions or
a family of $\varepsilon$-almost dual universal$_2$ hash functions for privacy amplification 
(Theorems \ref{Lem14} and \ref{Lem12-q} in Section \ref{s4-1}). 
Then, we have derived lower bounds on their exponential decreasing rates in the i.i.d. setting.
(Theorems \ref{t-3-16-2} and \ref{t-3-16-3} in Section \ref{s4-1-b}). 
The obtained bound for  the $L_1$ distinguishability criterion
has been shown to be tight in the qubit case
when the state is generated by transmission via Pauli channel (Appendix \ref{aL8-16-6}).
The obtained exponents are summarized in Table \ref{table2}.
We have also applied our result to the case when we need error correction.
In this case, we apply the privacy amplification after error correction as given in Subsection \ref{s5-1}.
Then, we have derived upper bounds for 
the information leaked with respect to the final keys in the respective criteria 
as well as upper bounds for the probability for disagreement in the final keys 
(Theorems \ref{L3-20-7}, \ref{L3-20-8}, and \ref{L12-31-2} in Section \ref{s5}). 
Applying them to the i.i.d. setting,
we have derived lower bounds on their exponential decreasing rates.
(Theorems \ref{t3-20-20}, \ref{t3-20-21}, \ref{p3-16-1c}, and \ref{t3-20-11b} in Section \ref{s5}).

\begin{table}[htb]
  \caption{Summary of obtained lower bounds on exponents.}
\begin{center}
  \begin{tabular}{|l|c|c|c} \hline
Task & L1  & MMI   \\ \hline
\multirow{2}{*}{PV (R\'{e}nyi)} 
& \multirow{2}{*}{$e_{\rG,\rq}(\rho_{A,E}|R)$}  & $e_{\rH,\rq}(\rho_{A,E}|R)$, \\
&   &  $e_{\rG,\rq}(\rho_{A,E}|R)$ \\ \hline
{PV \& fixed EC} 
& {$e_{\rG,\rq}(\rho_{A,E}|\log q-R_2)$} 
& $e_{\rG,\rq}(\rho_{A,E}|\log q-R_2)$ \\ 
\hline
\multirow{2}{*}{PV \& randomized EC} 
& \multirow{2}{*}{no improvement} 
& $e_{\rH,\rq}(\rho_{A,E}|\log q-R_2)$ \\ 
&& $e_{\rG,\rq}(\rho_{A,E}|\log q-R_2)$\\ \hline
  \end{tabular}
\end{center}

\vspace{2ex}
$R$ is the key generation rate.
$R_2$ is the sacrifice rate.
PV (R\'{e}nyi) is 
the exponent for privacy amplification via our approximate smoothing of R\'{e}nyi entropy of order 2.
EC is error correction.
L1 is the $L_1$ distinguishability criterion.
MMI is the modified mutual information criterion. 
\Label{table2}
\end{table}

Since a family of $\varepsilon$-almost dual universal$_2$ hash functions 
is a larger family of liner universal$_2$ hash functions, 
the obtained result suggests a possibility of the existence of 
an effective privacy amplification protocol 
with a smaller calculation time than known privacy amplification protocols.
In fact, as shown in the forthcoming paper \cite{H-T},
there exists an example of $\varepsilon$-almost dual universal$_2$ hash functions
with a smaller calculation amount and smaller number of random variables 
than the concatenation of Toeplitz matrix and the identity matrix.
Hence, it is expected that the obtained evaluation has a future application from an applied viewpoint. 

In fact, our bounds have polynomial factors in the quantum setting.
When the order of these polynomial factors are large,
the bounds do not work well when the number $n$ is not sufficiently large.
Fortunately, 
as is discussed in Subsection \ref{s4-1-b},
some of them have the order $n^{3/2}$ at most.
We can expect that these types of bounds work well 
even when 
the number $n$ is not sufficiently large.
These types of bounds and these discussions have been extended to the case when error correction is needed.
Further, as is discussed in Subsubsection \ref{s4-1-0-b-qb},
we can expect that some of obtained bounds work well even in the 
non-i.i.d. case.

In Section \ref{s10},
we have applied our result to the case when
Eve obtains the all information leaked to the environment via Pauli channel.
In this case,
our bounds can be described by using the joint classical distribution 
with respect to the bit error and the phase error.
We have numerically compared the obtained lower bounds on 
the exponential decreasing rates for leaked information.


Due to Pinsker inequality and Inequality (\ref{8-26-9-q}), 
the exponential convergence of one criterion yields the exponential convergence of the other criterion.
However, we have shown that
better exponential decreasing rates can be obtained by separate derivations. 
Our approximate smoothing of R\'{e}nyi entropy of order 2 yields the lower bound
$e_{\rG,\rq}(P_{A,E}|R)$ 
of the exponent of the $L_1$ distinguishability criterion,
which yields the lower bound $e_{\rG,\rq}(P_{A,E}|R)$ 
of the exponent of the modified mutual information criterion by using Pinsker inequality.
Similarly, our approximate smoothing of R\'{e}nyi entropy of order 2 yields the lower bound
$e_{\rH,\rq}(P_{A,E}|R)$ 
of the exponent of the modified mutual information criterion,
which yields the lower bound
$\frac{e_{\rH,\rq}(P_{A,E}|R)}{2}$ 
of the exponent of the $L_1$ 
distinguishability criterion by Inequality (\ref{8-19-14-q}).
Since $e_{\rG,\rq}(P_{A,E}|R) \ge \frac{e_{\rH,\rq}(P_{A,E}|R)}{2}$, 
we can conclude that
the evaluation of the $L_1$ distinguishability criterion
becomes worse if it goes through another criterion.
However, 
since we have not derived the definitive relation between $e_{\rH,\rq}(P_{A,E}|R)$ and $e_{\rG,\rq}(P_{A,E}|R)$,
we cannot say the same thing for the modified mutual information criterion.
The relation is also a future problem. 

\section*{Acknowledgments}
The author is grateful to Dr. Toyohiro Tsurumaru,
Dr. Shun Watanabe, 
Dr. Marco Tomamichel,
Dr. William Henry Rosgen,
Dr. Li Ke,
and Dr. Markus Grassl for a helpful comments.
He would like to express his appreciation to the referees of this paper for their helpful comments.
He is also grateful to the
referee of the first version of \cite{Tsuru}
for informing the literatures \cite{DS05,FS08}.
He also is partially supported by a MEXT Grant-in-Aid for Young Scientists (A) No. 20686026 and Grant-in-Aid for Scientific Research (A) No. 23246071.
He is partially supported by the National Institute of Information and Communication Technology (NICT), Japan.
The Centre for Quantum Technologies is funded by the
Singapore Ministry of Education and the National Research Foundation
as part of the Research Centres of Excellence programme. 

\appendices

\section{Modified mutual information criterion}\Label{s8-24}
It is natural to adopt a quantity expressing the difference between the true state and the ideal state 
$\rho_{\mix,A} \otimes \rho_{E}$ as a security criterion.
However, there are several quantities  expressing the difference between two states.
Both $d_1'(A|E|\rho)$ and $I'(A|E|\rho)$
are characterized in this way.
Here, we show that the modified mutual criterion $I'(A|E|\rho)$ can be derived in a natural way.

It is natural assume the following condition for the security criterion $C(A;E|\rho)$
as well as the unitary invariance on ${\cal H}_E$ and the permutation invariance on ${\cal H}_A$. 
\begin{description}
\item[\bf C1]{\bf Chain rule}
$C(A,B|E|\rho)=C(B|E|\rho)+C(A|B,E|\rho)$.

\item[\bf C2]{\bf Linearity}
When two states $\rho_1$ and $\rho_2$ are distinghuishable on ${\cal H}_E$,
$C(A|E|\lambda \rho_1+(1-\lambda) \rho_2)=\lambda C(A|E|\rho_1)+(1-\lambda)C(A|E|\rho_2)$.


\item[\bf C3]{\bf Range}
$ \log d_A \ge C(A|E|\rho) \ge 0$.

\item[\bf C4]{\bf Ideal case}
$C(A|E|\rho_{\mix,A} \otimes \rho_E)=0$.

\item[\bf C5]{\bf Normalization}
$C(A|E||a\rangle \langle a|\otimes \rho_E)=\log d_A$.
\end{description}
Unfortunately, 
the $L_1$ distinguishability does not satisfies {\bf C1 Chain rule}.
However, 
we have the following lemma.
\begin{lem}\Label{l8-24-1}
The modified mutual information criterion $I'(A|E|\rho)= \log d_A -H(A|E|\rho)$
satisfies all of these conditions.
\end{lem}
Further, we have the following theorem.
\begin{thm}\Label{l8-24-2}
When 
$C(A|E|\rho)$ satisfies all of the above properties and 
$\rho'$ is written as $\sum_{a,e} P_{A,E}(a,e)|a,e\rangle \langle a,e|$,
$C(A|E|\rho')= I'(A|E|\rho')= \log d_A -H(A|E|\rho')$.
\end{thm}

That is,
in the classical case, the security criterion is written by using the conditional entropy.
In the quantum case, the above theorem cannot determine uniquely the security criterion.
Since the most natural quantum extension of the conditional entropy
is the quantum conditional entropy $H(A|E|\rho)$.
Hence, it is natural to adopt
the modified mutual information criterion $I'(A|E|\rho)$ as a security criterion.
In particular, 
if one emphasizes {\bf C1 Chain rule} 
rather than the universal composability,
it is better employ the modified mutual information criterion $I'(A|E|\rho)$.

\begin{proofof}{Lemma \ref{l8-24-1}}
We can trivially check the conditions {\bf C4 Ideal case} 
and {\bf C5 Normalization}.
We show other conditions.

{\bf C1 Chain rule} can be shown as follows.
\begin{align*}
&I'(A,B|E|\rho)
=\log d_A+\log d_B
- H(A,B,E|\rho) + H(E|\rho) \\
=&\log d_A+\log d_B
- H(B,E|\rho) + H(E|\rho)
- H(A,B,E|\rho) + H(B,E|\rho) \\
=& \log d_A+\log d_B
- H(B|E|\rho)
- H(A|B,E|\rho)
= I'(A|B,E|\rho) +I'(B|E|\rho).
\end{align*}

When two states $\rho_1$ and $\rho_2$ are distinghuishable on ${\cal H}_E$,
\begin{align*}
& I'(A|E|\lambda \rho_1+(1-\lambda) \rho_2)
=\log d_A
- H(A,E|\lambda \rho_1+(1-\lambda) \rho_2) + H(E|\lambda \rho_1+(1-\lambda) \rho_2) \\
=&\log d_A
- \lambda  H(A,E|\rho_1)
-(1-\lambda) H(A,E|\rho_2) 
- h(\lambda)
+ \lambda  H(E|\rho_1)
+(1-\lambda) H(E|\rho_2) 
+ h(\lambda) \\
=&\log d_A
- \lambda  H(A,E|\rho_1)
-(1-\lambda) H(A,E|\rho_2) 
+ \lambda  H(E|\rho_1)
+(1-\lambda) H(E|\rho_2) \\
=&\lambda I'(A|E|\rho_1)+(1-\lambda)I'(A|E|\rho_2),
\end{align*}
which implies {\bf C2 Linearity}.

$I'(A|E|\rho)=D(\rho\|\rho_{\mix,A}\otimes \rho_E)\ge 0$.
Since $\rho$ is separable, $H(A,E|\rho)\ge 0 $ \cite{NK}.
Hence, $I'(A|E|\rho)$ satisfies {\bf C3 Range}.
\end{proofof}

\begin{proofof}{Theorem \ref{l8-24-2}}
We discuss 
$\tilde{H}(A|E|\rho):=\log d_A - C(A|E|\rho)$.
Due to {\bf C2 Linearity},
we have
\begin{align*}
\tilde{H}(A|E|\rho)=
\sum_e P_E(e) \tilde{H}(A|E|\sum_{a} P_{A|E}(a|e)|a,e\rangle \langle a,e|).
\end{align*}
Further, we see that the quantity $\tilde{H}(A|E|\sum_{a} P_{A|E}(a|e)|a,e\rangle \langle a,e|)$ satisfies 
Khinchin's axioms \cite{Khinchin} for entropy due to the remaining conditions.
Hence, we find that $\tilde{H}(A|E|\sum_{a} P_{A|E}(a|e)|a,e\rangle \langle a,e|)=H(P_{A|E=e})$.
Thus, $\tilde{H}(A|E|\rho)$ is equal to the conditional entropy ${H}(A|E|\rho)$.
Hence, $C(A|E|\rho)=I'(A|E|\rho)$.
\end{proofof}

\section{Proof of Lemma \ref{L31}}\Label{sL31}
First, 
we focus on the 
spectral decomposition of $\sigma$:
$\sigma= \sum_i s_i E_i$.
Since $x \mapsto x^\frac{1+s}{2}$ is operator concave,
\begin{align}
E_i \rho^\frac{1+s}{2} E_i 
\le 
(E_i \rho E_i )^\frac{1+s}{2} .
\Label{8-29-r1}
\end{align}
When $v$ is the number of eigenvectors of $\sigma$
Inequality (\ref{8-15-23}) implies
\begin{align}
\rho^\frac{1+s}{2} 
\le
v \sum_i E_i \rho^\frac{1+s}{2} E_i .
\Label{8-29-r2}
\end{align}
Since $E_i$ and $E_{i'}$ are orthogonal to each other 
for $i\neq i'$,
\begin{align}
\sum_i (E_i \rho E_i )^\frac{1+s}{2} 
=
( \sum_i E_i \rho E_i )^\frac{1+s}{2} .
\Label{8-29-r3}
\end{align}
Combining (\ref{8-29-r1}), (\ref{8-29-r2}), and (\ref{8-29-r3}), 
we obtain
\begin{align*}
& \sigma^{-\frac{s}{4}} \rho^\frac{1+s}{2} \sigma^{-\frac{s}{4}}
\le
v \sigma^{-\frac{s}{4}} 
\sum_i E_i \rho^\frac{1+s}{2} E_i \sigma^{-\frac{s}{4}} \\
\le &
v \sum_i
\sigma^{-\frac{s}{4}} 
(E_i \rho E_i )^\frac{1+s}{2} 
\sigma^{-\frac{s}{4}}
=
v \sigma^{-\frac{s}{4}} 
({\cal E}_{\sigma}(\rho))^\frac{1+s}{2} 
\sigma^{-\frac{s}{4}}.
\end{align*}
Thus, (\ref{8-21-7-q}) implies 
\begin{align}
& e^{\underline{\psi}(s|\rho\|\sigma)} 
= \Tr (\sigma^{-\frac{s}{4}} \rho^\frac{1+s}{2} \sigma^{-\frac{s}{4}})^2 
\le 
v\Tr (\sigma^{-\frac{s}{4}} 
({\cal E}_{\sigma}(\rho))^\frac{1+s}{2} 
\sigma^{-\frac{s}{4}})^2 
=
v e^{\underline{\psi}(s|{\cal E}_{\sigma}(\rho) \|\sigma)} 
=
v e^{\psi(s|{\cal E}_{\sigma}(\rho) \|\sigma)} 
\le 
v e^{\psi(s|\rho\|\sigma)} .
\Label{8-27-1}
\end{align}
That is, $\underline{\psi}(s|\rho\|\sigma)\le
\log v+\psi(s|\rho\|\sigma)$.
When we denote the number of eigenvalues of $\sigma^{\otimes n}$
by $v_n$,
we have
\begin{align}
n \underline{\psi}(s|\rho\|\sigma)
=\underline{\psi}(s|\rho^{\otimes n}\|\sigma^{\otimes n})
\le
\log v_n+\psi(s|\rho^{\otimes n}\|\sigma^{\otimes n})
=
\log v_n+ n \psi(s|\rho\|\sigma).\Label{8-29-r4}
\end{align}
Dividing (\ref{8-29-r4}) by $n$ and taking the limit $n \to \infty$, we obtain (\ref{8-26-2}).

\section{Proof of Lemma \ref{L3-26-1}}\Label{sL3-26-1}
The convexity of $\psi(s|\rho\|\sigma)$ is shown in \cite[Exercises 2.24]{Hayashi-book}.
Using this fact, we obtain the desired argument with respect to
$\psi(s|\rho\|\sigma)$.
The convexity of $\underline{\psi}(s|\rho\|\sigma)$ 
can be shown in the following way:
\begin{align*}
\frac{d \underline{\psi}(s|\rho\|\sigma)}{ds} 
=&
\frac{\Tr (\log \rho-\log \sigma) \rho^{\frac{1+s}{2}} \sigma^{-s/2}\rho^{\frac{1+s}{2}} \sigma^{-s/2}}
{\Tr\rho^{\frac{1+s}{2}} \sigma^{-s/2}\rho^{\frac{1+s}{2}} \sigma^{-s/2}} ,\\
\frac{d^2 \underline{\psi}(s|\rho\|\sigma)}{ds^2} 
=&
\frac{\Tr (\log \rho-\log \sigma) \rho^{\frac{1+s}{2}} (\log \rho-\log \sigma) \sigma^{-s/2}\rho^{\frac{1+s}{2}} \sigma^{-s/2}}
{2 \Tr\rho^{\frac{1+s}{2}} \sigma^{-s/2}\rho^{\frac{1+s}{2}} \sigma^{-s/2}} \\
&+
\frac{\Tr (\log \rho-\log \sigma) \rho^{\frac{1+s}{2}} \sigma^{-s/2} \rho^{\frac{1+s}{2}} (\log \rho-\log \sigma) \sigma^{-s/2}}
{2 \Tr\rho^{\frac{1+s}{2}} \sigma^{-s/2}\rho^{\frac{1+s}{2}} \sigma^{-s/2}} 
-
(\frac{\Tr (\log \rho-\log \sigma ) \rho^{\frac{1+s}{2}} \sigma^{-s/2}\rho^{\frac{1+s}{2}} \sigma^{-s/2}}
{\Tr\rho^{\frac{1+s}{2}} \sigma^{-s/2}\rho^{\frac{1+s}{2}} \sigma^{-s/2}})^2 .
\end{align*}
Now, we consider two kinds of inner products between two matrices $X$ and $Y$:
\begin{align*}
\langle Y, X \rangle_1 
:=
\Tr X \rho^{\frac{1+s}{2}} Y^\dagger \sigma^{-s/2}\rho^{\frac{1+s}{2}} \sigma^{-s/2} , \quad
 \langle Y, X \rangle_2 
:= 
\Tr X \rho^{\frac{1+s}{2}} \sigma^{-s/2} \rho^{\frac{1+s}{2}} Y^\dagger \sigma^{-s/2}.
\end{align*}
Applying Schwarz inequality to the case of $X=(\log \rho-\log \sigma)$ and $Y=I$,
we obtain
\begin{align*}
& \Tr (\log \rho-\log \sigma) \rho^{\frac{1+s}{2}} (\log \rho-\log \sigma) \sigma^{-s/2}\rho^{\frac{1+s}{2}} \sigma^{-s/2} 
\cdot \Tr\rho^{\frac{1+s}{2}} \sigma^{-s/2}\rho^{\frac{1+s}{2}} \sigma^{-s/2}
\ge 
(\Tr (\log \rho-\log \sigma ) \rho^{\frac{1+s}{2}} \sigma^{-s/2}\rho^{\frac{1+s}{2}} \sigma^{-s/2})^2 
\end{align*}
and
\begin{align*}
\Tr (\log \rho-\log \sigma) \rho^{\frac{1+s}{2}} \sigma^{-s/2} \rho^{\frac{1+s}{2}} (\log \rho-\log \sigma) \sigma^{-s/2}
\cdot \Tr\rho^{\frac{1+s}{2}} \sigma^{-s/2}\rho^{\frac{1+s}{2}} \sigma^{-s/2} 
\ge &
(\Tr (\log \rho-\log \sigma ) \rho^{\frac{1+s}{2}} \sigma^{-s/2}\rho^{\frac{1+s}{2}} \sigma^{-s/2})^2 .
\end{align*}
Therefore,
\begin{align*}
\frac{\Tr (\log \rho-\log \sigma) \rho^{\frac{1+s}{2}} (\log \rho-\log \sigma) \sigma^{-s/2}\rho^{\frac{1+s}{2}} \sigma^{-s/2}}
{\Tr\rho^{\frac{1+s}{2}} \sigma^{-s/2}\rho^{\frac{1+s}{2}} \sigma^{-s/2}} 
\ge &
(\frac{\Tr (\log \rho-\log \sigma ) \rho^{\frac{1+s}{2}} \sigma^{-s/2}\rho^{\frac{1+s}{2}} \sigma^{-s/2}}
{\Tr\rho^{\frac{1+s}{2}} \sigma^{-s/2}\rho^{\frac{1+s}{2}} \sigma^{-s/2}})^2 ,\\
\frac{\Tr (\log \rho-\log \sigma) \rho^{\frac{1+s}{2}} \sigma^{-s/2} \rho^{\frac{1+s}{2}} (\log \rho-\log \sigma) \sigma^{-s/2}}
{\Tr\rho^{\frac{1+s}{2}} \sigma^{-s/2}\rho^{\frac{1+s}{2}} \sigma^{-s/2}} 
\ge &
(\frac{\Tr (\log \rho-\log \sigma ) \rho^{\frac{1+s}{2}} \sigma^{-s/2}\rho^{\frac{1+s}{2}} \sigma^{-s/2}}
{\Tr\rho^{\frac{1+s}{2}} \sigma^{-s/2}\rho^{\frac{1+s}{2}} \sigma^{-s/2}})^2 ,
\end{align*}
which implies
\begin{align*}
\frac{d^2 \underline{\psi}(s|\rho\|\sigma)}{ds^2} \ge 0.
\end{align*}
In particular, when $\rho$ and $\sigma$ are not completely mixed,
the above inequalities are strict.
Hence, 
the functions $s \mapsto
\psi(s|\rho\|\sigma),
\underline{\psi}(s|\rho\|\sigma)$ are strictly convex

\section{Proof of Lemma \ref{cor1-q}}\Label{scor1-q}
Assume that $s \in (0,\infty)$.
For 
two non-negative matrices $X$ and $Y$, 
the reverse operator H\"{o}lder inequality
\begin{align*}
\Tr X Y \ge
(\Tr X^{1/(1+s)})^{1+s}
(\Tr Y^{-1/s})^{-s}
\end{align*}
holds.
Substituting 
$\sum_a P_A(a)^{1+s} \rho_{E|a}^{1+s}$
and $\sigma_E^{-s}$
to $X$ and $Y$, we obtain  
\begin{align*}
& e^{-s H_{1+s}(A|E|\rho_{A,E}\| \sigma_E )} 
= 
\Tr 
\sum_a 
( P_{A}(a) \rho_{E|a})^{1+s}
\sigma_E^{-s} 
\ge 
(\Tr (\sum_a ( P_{A}(a) \rho_{E|a} )^{1+s})^{1/(1+s)})^{1+s}
(\Tr \sigma_E^{-s\cdot -1/s})^{-s} \\
= &
(\Tr 
(\sum_a ( P_{A}(a) \rho_{E|a} )^{1+s})^{1/(1+s)})^{1+s} 
= e^{-s H_{1+s}^{\rG}(A|E|\rho_{A,E})} .
\end{align*}
Since the equality holds
when $\sigma_E=
(\sum_{a} ( P_{A}(a) \rho_{E|a} )^{1+s})^{1/(1+s)} /
\Tr  (\sum_{a} ( P_{A}(a) \rho_{E|a} )^{1+s})^{1/(1+s)}$,
we obtain
\begin{align*}
\min_{\sigma_E}e^{-s H_{1+s}(A|E|\rho_{A,E}\| \sigma_E )} 
= 
e^{-s H_{1+s}^{\rG}(A|E|\rho_{A,E})} ,
\end{align*}
which implies (\ref{8-26-8}).

When $s \in (-1,0)$, 
applying the operator H\"{o}lder inequality 
$\Tr X Y \le
(\Tr X^{1/(1+s)})^{1+s}
(\Tr Y^{-1/s})^{-s}$
instead of the reverse operator H\"{o}lder inequality,
we obtain
\begin{align*}
 e^{-s H_{1+s}(A|E|\rho_{A,E}\| \sigma_E )} 
\le & e^{-s H_{1+s}^{\rG}(A|E|\rho_{A,E})} .
\end{align*}
The equality can be shown in the same way.

\section{Proof of Lemma \ref{L8-16-6}}\Label{aL8-16-6}
\subsection{Outline of the proof}
Since (\ref{12-18-9b}) implies
\begin{align}
\lim_{n \to \infty}\frac{-1}{n} \log 
\Delta_{d,2}(e^{nR}, \varepsilon_1|\rho_{A,E}^{\otimes n}) 
& \ge e_{\rG,\rq}(\rho_{A,E}|R)
=\max_{0 \le s \le 1/2}
\frac{-sH_{1-s}(P_X) +s (\log 2-R)}{2(1-s)},
\end{align}
it is enough to show the opposite inequality.
For this purpose, we will show the following lemma.

\begin{lem}\Label{L8-16-3}
When we choose an $\lceil n(1-R) \rceil$-dimensional subspace $C_{\bbZ}\subset \FF_2^n$ with equal probability,
we obtain
\begin{align}
& \lim_{n \to \infty}\frac{-1}{n}\log
\rE_{\bbZ}
d_1'([A]_{C_{\bbZ}}|E|\rho_{A,E})
=
\max_{0 \le s \le 1/2}
\frac{-sH_{1-s}(P_X) +s (1-R)\log 2}{2(1-s)}.\Label{8-16-9}
\end{align}
\end{lem}

Here, we prove Lemma \ref{L8-16-6} by using Lemma \ref{L8-16-3}.
When we choose an $\lceil n(1-\frac{R'}{\log 2}) \rceil$-dimensional subspace $C_{\bbZ}\subset \FF_2^n$ with equal probability,
since the hash function $X \mapsto [X]_{C_{\bbZ}}$ satisfies the universal$_2$ condition,
we obtain
\begin{align}
\rE_{\bbZ}
d_1'([A]_{C_{\bbZ}}|E|\rho_{A,E})
\le 
\Delta_{d,2}(2^{\lfloor n \frac{R'}{\log 2} \rfloor},1|\rho_{A,E})
\le 
\Delta_{d,2}(e^{nR'},1|\rho_{A,E}),
\end{align}
which implies that
\begin{align}
& \limsup_{n \to \infty}\frac{-1}{n}\log
\Delta_{d,2}(e^{nR'},1|\rho_{A,E})
\le
\max_{0 \le s \le 1/2}
\frac{-sH_{1-s}(P_X) +s (\log 2-R')}{2(1-s)}.
\end{align}
Since Inequality (\ref{8-18-1}) is the opposite inequality,
we obtain (\ref{8-16-20}).

In the following, we prepare two lemmas for the proof of Lemma \ref{L8-16-3}.
Given a code $C \subset \FF_p^n$,
we can define its orthogonal space $C^{\perp} \subset \FF_p^n$. 
Then, for $[x_2]_{C^{\perp}}\in \FF_p^n/C^{\perp}$ and $x_1 \in [x_2]_{C^{\perp}}$,
we define the conditional distribution 
$P_{X|[X]_{C^{\perp}}}(x_1|[x_2]_{C^{\perp}}):= \frac{P_X(x_1)}{P_{[X]_{C^{\perp}}}([x_2]_{C^{\perp}})}$,
where $P_{[X]_{C^{\perp}}}([x_2]_{C^{\perp}}):
= \sum_{x_1 \in [x_2]_{C^{\perp}}} P_X(x_1)$.
Then, we define a pure state 
$|[a]_{C},[x_2]_{C^{\perp}}\rangle  
:= $\par\noindent$\sum_{x_1 \in [x_2]_{C^{\perp}}} \omega^{a x_1} \sqrt{P(x_1|[x_2]_{C^{\perp}})}|x_1 \rangle$
for $[a]_{C} \in \FF_p^n/C$ and $[x_2]_{C^{\perp}}\in \FF_p^n/C^{\perp}$.
Note that the definition of the state 
$|[a]_{C},[x_2]_{C^{\perp}}\rangle $ does not depend on the choice of the representatives of
$[a]_{C}$ and $[x_2]_{C^{\perp}}$ except for the phase factor.
Then, the relation
\begin{align*}
\rho_{E|[a]_{C}}:= \sum_{y \in C} \frac{1}{|C|} |a + y : P_X \rangle 
\langle j + y : P_X |
= \sum_{[x_2]_{C^{\perp}} \in \FF_p^n/C^{\perp}} 
P_X([x_2]_{C^{\perp}}) |[a]_{C},[x_2]_{C^{\perp}}\rangle 
\langle [a]_{C},[x_2]_{C^{\perp}}|.
\end{align*}
holds.
In order to describe the maximum likelihood estimator 
of the code $C^{\perp}$ under the distribution $P_X$,
we define $x([x_2]_{C^{\perp}}):= \argmax_{x_1 \in [x_2]_{C^{\perp}}} P_X(x_1) $.
Then, the decoding error probability is given as
\begin{align}
P_e(C^{\perp}):= 
1-\sum_{[x_2] \in \FF_p^n/C^{\perp}} P_X(x([x_2]))
1-\sum_{[x_2] \in \FF_p^n/C^{\perp}}
\max_{x_1 \in [x_2]}P_X(x_1).
\end{align}

\begin{lem}\Label{L8-16-1}
The relation
\begin{align}
& 2 \sum_{[x_2]_{C^{\perp}} \in \FF_p^n/C^{\perp}}
\sqrt{P_X(x([x_2]_{C^{\perp}}))(P_X([x_2]_{C^{\perp}})- P_X(x([x_2]_{C^{\perp}}))) } \nonumber \\
\le &
d_1'([A]_{C}|E|\rho_{A,E})=\| \rho_E - \rho_{E|[a]_{C}}\|_1 
\Label{8-16-1}\\
\le & 2 \sqrt{2 P_e({C^{\perp}})} \Label{8-16-2}
\end{align}
holds for $a \in \FF_2^n$.
\end{lem}
The proof of Lemma \ref{L8-16-1} is given in Appendix \ref{sL8-16-1}.

Now, we consider the binary case, i.e., the case of $\FF_2^n$.
We choose an $m$-dimensional subspace $C_{\bX}\subset \FF_2^n$ with equal probability.
That is, there are $G(m):=\prod_{i=0}^{m-1}\frac{2^n-2^i}{2^m-2^i}$ distinct 
$m$-dimensional subspaces in $\FF_2^n$.
Hence, we chose each of them with the probability $1/G(m)$.
\begin{lem}\Label{L8-16-2}
When 
we choose an $\lceil nR \rceil$-dimensional subspace $C_{\bX}\subset \FF_2^n$ with equal probability,
\begin{align}
& \lim_{n \to \infty}\frac{-1}{n}\log
\rE_{\bX}
\sum_{[x_2]_{C_{\bX}} \in \FF_2^n/C_{\bX}}
\sqrt{P_X^n(x([x_2]_{C_{\bX}}))(P_X^n([x_2]_{C_{\bX}})- P_X^n(x([x_2]_{C_{\bX}}))) } 
\nonumber \\
=& \lim_{n \to \infty}\frac{-1}{n}\log
\rE_{\bX} P_e(C_{\bX}) \nonumber \\
=& \frac{1}{2} \min_{Q:\log 2(1 -R) \ge H(Q)} D(Q\|P_X) + \log 2(1- R) - H(Q) \nonumber \\
=& \frac{1}{2} 
\max_{0 \le s \le 1/2}
\frac{-sH_{1-s}(P_X) +s \log 2(1-R)}{1-s}.
\Label{8-13-h}
\end{align}
\end{lem}
The proof of Lemma \ref{L8-16-2} is given in Appendix \ref{sL8-16-2}.

\begin{proofof}{Lemma \ref{L8-16-3}}
We apply Lemma \ref{L8-16-2} to the case $C_{\bX}=C_{\bbZ}^{\perp}$.
Then, the exponential decreasing rates of the upper and lower bounds given in Lemma \ref{L8-16-1}
are $\max_{0 \le s \le 1/2} \frac{sH_{1-s}(P_X) -s (1-R)\log 2}{1-s}$,
which implies (\ref{8-16-9}).
\end{proofof}

\subsection{Proof of Lemma \ref{L8-16-1}}\Label{sL8-16-1}
In this proof, we abbreviate $[x]_{C^{\perp}}$ by $[x]$.
Since 
\begin{align}
\| \rho_E - \rho_{E|[a]_C}\|_1
=
\| \sW(0,z) (\rho_E - \rho_{E|[a]_C}) \sW(0,z)^\dagger\|_1
=\| \rho_E - \rho_{E|[a+z]_C}\|_1
\end{align}
for $z,a \in \FF_p^n$,
we have
$d_1'([A]_C|E|\rho_{A,E})=\| \rho_E - \rho_{E|[a]_C}\|_1$.

Next, we prove the inequality (\ref{8-16-2}).
For this purpose, we define the fidelity as
$F(\rho_E,\rho_{E|[a]_C}):=\Tr |\sqrt{\rho_E} \sqrt{\rho_{E|[a]_C}}|$.
The fidelity satisfies that
\begin{align}
\| \rho_E - \rho_{E|[a]_C}\|_1
\le 2 \sqrt{1-F(\rho_E,\rho_{E|[a]_C})^2},\Label{8-16-5}
\end{align}
and is characterized as
\begin{align}
F(\rho_E,\rho_{E|[a]_C})^2=
\Bigl(\sum_{[x_2] \in \FF_p^n/C^{\perp}} P_X([x_2])
\sqrt{\sum_{x_1 \in [x_2]} P_{X|[X]}(x_1|[x_2])^2}
\Bigr)^2
= e^{-H_2^{\rm G}(X|[X]|P_X)}.\Label{8-16-6}
\end{align}
Since $e^{-H_2^{\rm G}(X|[X]|P_X)} \ge  
\Bigl(\sum_{[x_2] \in \FF_p^n/C^{\perp}} P_X([x_2]) 
\max_{x_1 \in [x_2]} P_{X|[X]}(x_1|[x_2])
\Bigr)^2
=(1-P_e ({C^{\perp}}) )^2$,
we have
\begin{align}
1- e^{-H_2(X|[X]|P_X)}
\le 1-(1-P_e ({C^{\perp}}) )^2
=1 -1 + 2P_e ({C^{\perp}})-P_e ({C^{\perp}})^2
\le 2P_e ({C^{\perp}})
\Label{8-16-4}.
\end{align}
Combining (\ref{8-16-5}), (\ref{8-16-6}), and (\ref{8-16-4}),
we obtain (\ref{8-16-2}).

Next, we show (\ref{8-16-1}).
For $x_1 \in [x_2] \setminus \{x([x_2])\}$, we define the operator
$K_{x_1}:= 
|x_1 \rangle \langle x_1|+ 
\sqrt{\frac{P_X(x_1)}{P_X([x_2])- P_X(x([x_2])) }}
|x([x_2])\rangle \langle x([x_2])|$.
Then, we have the relation
$\sum_{[x_2] \in \FF_p^n/C^{\perp}}\sum_{x_1 \in [x_2] \setminus \{x([x_2])\}} K_{x_1}^2=I$.
Hence, we can define the TP-CP map
$\Lambda:\rho \mapsto 
\sum_{[x_2] \in \FF_p^n/C^{\perp}}\sum_{x_1 \in [x_2] \setminus \{x([x_2])\}} 
K_{x_1} \rho K_{x_1} \otimes |x_1 \rangle_R ~_R\langle x_1|$,
where $\{|x_1 \rangle_R\}$ is a CONS on another system.
Thus,
\begin{align*}
K_{x_1} \rho_E K_{x_1}
&= 
P_X(x_1)|x_1 \rangle \langle x_1|
+ 
\frac{P_X(x_1)P_X(x([x_2])) }{P_X([x_2])- P_X(x([x_2])) }
|x([x_2])\rangle \langle x([x_2])| \\
K_{x_1} \rho_{E|[a]} K_{x_1}
&= 
(\sqrt{P_X(x_1)}|x_1 \rangle 
+ 
\sqrt{\frac{P_X(x_1)P_X(x([x_2])) }{P_X([x_2])- P_X(x([x_2])) }}
|x([x_2])\rangle )
(\langle x_1 \sqrt{P_X(x_1)} | 
+ 
\langle x([x_2])|
\sqrt{\frac{P_X(x_1)P_X(x([x_2])) }{P_X([x_2])- P_X(x([x_2])) }}).
\end{align*}
Hence,
\begin{align}
& \| K_{x_1} \rho_E K_{x_1}- K_{x_1} \rho_{E|[a]} K_{x_1}\|_1
= 
2
\sqrt{P_X(x_1)}
\sqrt{\frac{P_X(x_1)P_X(x([x_2])) }{P_X([x_2])- P_X(x([x_2])) }} \nonumber\\
= &
2
\frac{P_X(x_1)}{P_X([x_2])- P_X(x([x_2]))}
\sqrt{P_X(x([x_2]))(P_X([x_2])- P_X(x([x_2])))}\Label{8-16-7}.
\end{align}
Using the relation
$\sum_{x_1 \in [x_2] \setminus \{x([x_2])\}} \frac{P_X(x_1)}{P_X([x_2])- P_X(x([x_2]))}=1$
and (\ref{8-16-7}),
we obtain
\begin{align*}
& 
\| \rho_E - \rho_{E|[a]_C}\|_1
\ge
\| \Lambda(\rho_E) - \Lambda(\rho_{E|[[a]})\|_1 \\
\ge &
\sum_{[x_2] \in \FF_p^n/C^{\perp}}\sum_{x_1 \in [x_2] \setminus \{x([x_2])\}} 
\| K_{x_1} \rho_E K_{x_1} - K_{x_1}  \rho_{E|[a]_C} K_{x_1} \|_1 \\
\ge &
\sum_{[x_2] \in \FF_p^n/C^{\perp}}\sum_{x_1 \in [x_2] \setminus \{x([x_2])\}} 
\| K_{x_1} \rho_E K_{x_1} - K_{x_1}  \rho_{E|[a]_C} K_{x_1} \|_1 \\
= &
2 \sum_{[x_2] \in \FF_p^n/C^{\perp}}\sum_{x_1 \in [x_2] \setminus \{x([x_2])\}} 
\frac{P_X(x_1)}{P_X([x_2])- P_X(x([x_2]))}
\sqrt{P_X(x([x_2]))(P_X([x_2])- P_X(x([x_2])))} \\
= &
2 \sum_{[x_2] \in \FF_p^n/C^{\perp}}
\sqrt{P_X(x([x_2]))(P_X([x_2])- P_X(x([x_2]))) },
\end{align*}
which implies (\ref{8-16-1}).

\subsection{Proof of Lemma \ref{L8-16-2}}\Label{sL8-16-2}
In this proof, we abbreviate $[x]_{C_{\bX}}$ by $[x]$.
It was shown in \cite[Theorem 7]{Tsuru} that
\begin{align}
& \lim_{n \to \infty}\frac{-1}{n}\log
\rE_{\bX} P_e(C_{\bX}) 
\ge \max_{0 \le s \le 1/2}
\frac{-sH_{1-s}(P_X) +s \log 2(1-R)}{1-s}
\Label{8-13-i}.
\end{align}
We can show the following lemma.
\begin{lem}\Label{L8-22-1}
\begin{align}
&\max_{0 \le s \le 1/2}
\frac{-sH_{1-s}(P_X) +s \log 2(1-R)}{1-s}
\nonumber \\
=& \min_{Q:\log 2 (1-R) \ge H(Q)} D(Q\|P_X) + \log 2 (1-R) - H(Q) 
\Label{8-13-i2}.
\end{align}
\end{lem}
Lemma \ref{L8-22-1} is shown in Appendix \ref{sL8-22-1}.

Hence, it is enough to show that
\begin{align}
& \lim_{n \to \infty}\frac{-1}{n}\log
\rE_{\bX}
\sum_{[x_2] \in \FF_2^n/C_{\bX}}
\sqrt{P_X^n(x([x_2]))(P_X^n([x_2])- P_X^n(x([x_2]))) }\nonumber \\
\le & \frac{1}{2} \min_{Q:\log 2 (1-R) \ge H(Q)} D(Q\|P_X) + \log 2(1-R) - H(Q) 
\Label{8-13-j}
\\
& \lim_{n \to \infty}\frac{-1}{n}\log
\rE_{\bX}
\sum_{[x_2] \in \FF_2^n/C_{\bX}}
\sqrt{P_X^n(x([x_2]))(P_X^n([x_2])- P_X^n(x([x_2]))) } \nonumber\\
\ge & \frac{1}{2} 
\max_{0 \le s \le 1/2}
\frac{-sH_{1-s}(P_X) +s \log 2(1-R)}{1-s}.
\Label{8-13-k}
\end{align}

Now, we denote the set of empirical distributions on $\FF_2$
with $n$ trials by ${\cal T}_n$.
The cardinality $|{\cal T}_n|$ is $n+1$ \cite{CKbook}.
When $T_n(Q)$ represents the set of $n$-trial data whose empirical distribution is $Q$,
the cardinality of $T_n(Q)$ can be evaluated as \cite{CKbook}:
\begin{align}
\lceil \frac{e^{nH(Q)}}{n+1} \rceil \le |T_n(Q)| \le \lfloor e^{nH(Q)} \rfloor,
\Label{9-7-5}
\end{align}
where $ \lceil x \rceil$ is the minimum integer $m$ satisfying $m \ge x $,
and $\lfloor x \rfloor$ is the maximum $m$ satisfying $m \le x $.
Since any element $\vec{a} \in T_n(Q)$ 
satisfies 
\begin{align}
P_X^n(\vec{a})=e^{-n (D(Q\|P_X)+H(Q)) }, \Label{9-7-13}
\end{align}
we obtain an important formula
\begin{align}
\frac{1}{n+1} e^{-n D(Q\|P_X)}
\le P_X^n(T_n(Q)) \le 
e^{-n D(Q\|P_X)}
\Label{9-7-14}.
\end{align}

Now, we prepare the following lemma in the finite-length case.
\begin{lem}\Label{L8-13-2}
Assume that we choose an $m$-dimensional subspace $C_{\bX}\subset \FF_2^n$ with equal probability.
When 
$Q_1,Q_2 \in {\cal T}_n$ satisfies that
$H(Q_1) \le 2^{n-m}$
and $D(Q_1\|P_X)+H(Q_1) < D(Q_2\|P_X)+H(Q_2)$, 
we have
\begin{align}
& \rE_{\bX}
\sum_{[x_2] \in \FF_2^n/C_{\bX}}
\sqrt{P_X^n(x([x_2]))(P_X^n([x_2])- P_X^n(x([x_2]))) } 
\ge B_{n,m}(Q_1,Q_2) \Label{8-13-e} \\
& \rE_{\bX}
\sum_{[x_2] \in \FF_2^n/C_{\bX}}
\sqrt{P_X^n(x([x_2]))(P_X^n([x_2])- P_X^n(x([x_2]))) } 
\le 
e^{-\frac{1}{2} \max_{0 \le s \le 1}
\frac{n sH_{1-s}(P_X) -s (n-m)\log 2}{1-s}  },
\Label{8-13-f}
\end{align}
where
\begin{align}
B_{n,m}(Q_1,Q_2) := e^{-\frac{n}{2}(D(Q_1\|P_X)+H(Q_1)+D(Q_2\|P_X)+H(Q_2))} 
\frac{|T_n(Q_2)|(|T_n(Q_1)|-1)(1-\frac{(2^m-2)(|T_n(Q_1)|-2)}{2(2^n-2)} )2^{m-n}}
{(1+ \frac{|T_n(Q_2)|-1}{\frac{2^n-2}{2^m-2}-\frac{|T_n(Q_1)|-2}{2}} 
)^{\frac{1}{2}}}  .
\end{align}
\end{lem}
The proof of Lemma \ref{L8-13-2} is given in Appendix \ref{sL8-13-2}.

Since (\ref{8-13-f}) shows (\ref{8-13-k}),
we will show (\ref{8-13-j}) by using (\ref{8-13-e}).
When $ \log 2 (1-R) < H(Q_1)$
and $D(Q_1\|P_X)+H(Q_1) < D(Q_2\|P_X)+H(Q_2)$, 
\begin{align}
\lim_{n \to \infty}\frac{-1}{n}\log
B_{n,\lfloor nR \rfloor}(Q_1,Q_2)
=\frac{1}{2}
(D(Q_1\|P_X) +D(Q_2\|P_X) +\log 2 (1-R) - H(Q_1))
\Label{8-13-a}.
\end{align}

Choosing $Q_1=P_X$, we have
\begin{align}
&\inf_{Q_1,Q_2: \log 2 (1-R) < H(Q_1),D(Q_1\|P_X)+H(Q_1) < D(Q_2\|P_X)+H(Q_2)}
\frac{1}{2}
(D(Q_1\|P_X) +D(Q_2\|P_X) +\log 2 (1-R) - H(Q_1)) \nonumber\\
=&\inf_{Q_1: \log 2 (1-R) < H(Q_1)}
\frac{1}{2} (D(Q_1\|P_X) +\log 2 (1-R) - H(Q_1)) .
\Label{8-13-b}
\end{align}
Thus,
\begin{align*}
&\inf_{Q_1: \log 2 (1-R) < H(Q_1)}
D(Q_1\|P_X) +\log 2 (1-R) - H(Q_1) \\
=&\min_{Q_1: \log 2 (1-R) \le H(Q_1)}
D(Q_1\|P_X) +\log 2 (1-R) - H(Q_1) 
\ge 
\min_{Q_1: \log 2 (1-R) \le H(Q_1)}
D(Q_1\|P_X) .
\end{align*}
Since the minimum $\min_{Q_1: \log 2 (1-R) < H(Q_1)}D(Q_1\|P_X) $ can be realized with
$Q_1^*$ satisfying $\log 2 (1-R) = H(Q_1^*)$,
we have
\begin{align}
&\inf_{Q_1: \log 2 (1-R) < H(Q_1)}
D(Q_1\|P_X) +\log 2 (1-R) - H(Q_1) \nonumber \\
=&
\min_{Q_1: \log 2 (1-R) \le H(Q_1)}
D(Q_1\|P_X) 
=
D(Q_1^*\|P_X) .\Label{8-13-g}
\end{align}
It is known that
this quantity is the optimal error exponent with the source coding with the compression rate 
$\log 2 (1-R) $,
which is equal to 
$\max_{0 \le s \le 1/2}
\frac{sH_{1-s}(P_X) -s \log 2(1-R)}{1-s}$.
Hence, combining (\ref{8-13-e}), (\ref{8-13-e}), (\ref{8-13-g}),
and the above mentioned fact,
we obtain (\ref{8-13-h}).

\begin{proofof}{(\ref{8-13-a})}
Since $ \log 2 (1-R) < H(Q_1)$, we have
\begin{align*}
\lim_{n \to \infty}\frac{-1}{n}\log
(1-\frac{(2^m-2)(|T_n(Q_1)|-2)}{2(2^n-2)} ) &=0 \\
\lim_{n \to \infty}\frac{-1}{n}\log
\frac{|T_n(Q_2)|-1}{\frac{2^n-2}{2^m-2}-\frac{|T_n(Q_1)|-2}{2}} 
&=H(Q_2)-\log 2 (1-R) .
\end{align*}
Hence,
\begin{align*}
& \lim_{n \to \infty}\frac{-1}{n}\log
\frac{|T_n(Q_2)|(|T_n(Q_1)|-1)2^{m-n}}
{(1+ \frac{|T_n(Q_2)|-1}{\frac{2^n-2}{2^m-2}-\frac{|T_n(Q_1)|-2}{2}} 
)^{\frac{1}{2}}}  \\
=&
 H(Q_2)+H(Q_1)-\log 2(1-R) -\frac{1}{2}(H(Q_2)-\log 2 (1-R) ) 
= 
 \frac{H(Q_2)+R-\log 2}{2} +H(Q_1) .
\end{align*}
Thus, 
\begin{align*}
& \lim_{n \to \infty}\frac{-1}{n}\log
B_{n,\lfloor nR \rfloor}(Q_1,Q_2) 
= \frac{1}{2}(D(Q_1\|P_X)+H(Q_1)+D(Q_2\|P_X)+H(Q_2))
- \frac{H(Q_2)-\log 2(1-R)}{2} +H(Q_1) \\
=& \frac{1}{2}(D(Q_1\|P_X) +D(Q_2\|P_X) +\log 2 (1-R) - H(Q_1)),
\end{align*}
which implies (\ref{8-13-a}).
\end{proofof}

\subsection{Proof of Lemma \ref{L8-13-2}}\Label{sL8-13-2}
In this proof, we abbreviate $[x]_{C_{\bX}}$ by $[x]$.
In Lemma \ref{L8-13-2},
we choose an $m$-dimensional subspace $C_{\bX}\subset \FF_2^n$ with equal probability.
That is, there are $G(m):=\prod_{i=0}^{m-1}\frac{2^n-2^i}{2^m-2^i}$ distinct 
$m$-dimensional subspace in $\FF_2^n$.
Hence, we chose each of them with the probability $1/G(m)$.

Now, we show (\ref{8-13-f}). 
Since $x \mapsto \sqrt{x}$ is concave for $s \in [0,1]$,
we have
\begin{align}
& \rE_{\bX}
\sum_{[x_2] \in \FF_2^n/C_{\bX}}
\sqrt{P_X^n(x([x_2]))(P_X^n([x_2])- P_X^n(x([x_2]))) } \nonumber \\
= &
\rE_{\bX}
\sum_{[x_2] \in \FF_2^n/C_{\bX}}
P_X^n(x([x_2]))
\sqrt{\frac{(P_X^n([x_2])- P_X^n(x([x_2])))}{P_X^n(x([x_2]))} } \nonumber \\
\le &
\rE_{\bX}
\sqrt{
\sum_{[x_2] \in \FF_2^n/C_{\bX}}
P_X^n(x([x_2]))
\frac{(P_X^n([x_2])- P_X^n(x([x_2])))}{P_X^n(x([x_2]))} } \nonumber \\
= &
\rE_{\bX}
\sqrt{
\sum_{[x_2] \in \FF_2^n/C_{\bX}} P_X^n([x_2])- P_X^n(x([x_2]))} \nonumber \\
\le &
\sqrt{
\rE_{\bX}
\sum_{[x_2] \in \FF_2^n/C_{\bX}} P_X^n([x_2])- P_X^n(x([x_2]))} .\Label{8-13-l}
\end{align}
Since the quantity $\sum_{[x_2] \in \FF_2^n/C_{\bX}} P_X^n([x_2])- P_X^n(x([x_2]))$
is the average error probability when we apply maximum likelihood decoder,
it can be evaluated as
\begin{align}
\sum_{[x_2] \in \FF_2^n/C_{\bX}} P_X^n([x_2])- P_X^n(x([x_2]))
\le
2^{\frac{s(n- m)}{1-s}} (\sum_{x \in \FF_2} P_X(x)^{1-s})^{\frac{n}{1-s}}\Label{8-13-m}
\end{align}
with $s \in [0,\frac{1}{2}]$.
Combining (\ref{8-13-l}) and (\ref{8-13-m}), we obtain (\ref{8-13-f}).

Next, we proceed to the proof of (\ref{8-13-e}). 
For distinct elements $y_1, \ldots, y_l, y_{l+1}, \ldots, y_{l+k} \in \FF_2^n$,
we define the number $M(y_1, \ldots, y_l|y_{l+1}, \ldots, y_{l+k})$ as 
the number of cases
that one of $y_1, \ldots, y_l$ belongs to $C_{\bX}$ and one of $y_{l+1}, \ldots, y_{l+k}$ belongs to $C_{\bX}$.
In particular, 
$M(y_1, \ldots, y_l|\emptyset)$ denotes the number of cases that one of $y_1, \ldots, y_l$ belongs to $C_{\bX}$.
Then, we prepare the following lemma.
\begin{lem}\Label{L8-13-1}
\begin{align}
 M(y_1,y_2,\ldots, y_l|y_{l+1}) 
\le &
l \prod_{i=2}^{m-1}\frac{2^n-2^i}{2^m-2^i} \Label{8-13-1}\\
M(y_1,y_2,\ldots, y_l|\emptyset)
\ge &
l(\frac{2^n-2^1}{2^m-2^1}-\frac{l-1}{2} )
\prod_{i=2}^{m-1}\frac{2^n-2^i}{2^m-2^i} \Label{8-13-2}.
\end{align}
\end{lem}
The proof of Lemma \ref{L8-13-1} is given in the end of this subsection.

Now, using Lemma \ref{L8-13-1}, we show (\ref{8-13-e}). 
We define
$a([x_2])$ to be $1$ if $[x_2]\cap T_n(Q_1) \neq \emptyset$,
and to be $0$ otherwise.
We also define $N([x_2])$ the number of elements of $[x_2] \cap T_n(Q_2)$.
Then, for any code $C$ and any $[x_2] \in \FF_2^n/C$, 
\begin{align}
& \sqrt{P_X^n(x([x_2]))(P_X^n([x_2])- P_X^n(x([x_2]))) }\nonumber  \\
\ge & 
a([x_2]) e^{-\frac{n}{2}(D(Q_1\|P_X)+H(Q_1))}
e^{-\frac{n}{2}(D(Q_2\|P_X)+H(Q_2))}\sqrt{N([x_2]) } \nonumber \\
= &
a([x_2]) e^{-\frac{n}{2}(D(Q_1\|P_X)+H(Q_1))}
e^{-\frac{n}{2}(D(Q_2\|P_X)+H(Q_2))} N([x_2]) N([x_2])^{-\frac{1}{2}} .
\end{align}
Next, for $x \in \FF_2^n$, 
we define 
$b(C,x)$ to be $1$ if $(x+C) \cap T_n(Q_1) \neq \emptyset$,
and to be $0$ otherwise.
We also define the number $N(C,x)$
as the number of elements of $(x+C) \cap T_n(Q_2)$.
Hence,
\begin{align*}
& \sum_{[x_2] \in \FF_2^n/C}
\sqrt{P_X^n(x([x_2]))(P_X^n([x_2])- P_X^n(x([x_2]))) } \\
\ge &
\sum_{[x_2] \in \FF_2^n/C}
a([x_2]) e^{-\frac{n}{2}(D(Q_1\|P_X)+H(Q_1))}
e^{-\frac{n}{2}(D(Q_2\|P_X)+H(Q_2))} N([x_2]) N([x_2])^{-\frac{1}{2}} \\
= &
\sum_{x \in T_n(Q_2)}
b(C,x) e^{-\frac{n}{2}(D(Q_1\|P_X)+H(Q_1))}
e^{-\frac{n}{2}(D(Q_2\|P_X)+H(Q_2))} N(C,x)^{-\frac{1}{2}} .
\end{align*}
Thus,
\begin{align}
& \rE_{\bX}
\sum_{[x_2] \in \FF_2^n/C_{\bX}}
\sqrt{P_X^n(x([x_2]))(P_X^n([x_2])- P_X^n(x([x_2]))) } \nonumber \\
\ge &
\rE_{\bX}
\sum_{x \in T_n(Q_2)}
b(C_{\bX},x) e^{-\frac{n}{2}(D(Q_1\|P_X)+H(Q_1))}
e^{-\frac{n}{2}(D(Q_2\|P_X)+H(Q_2))} N(C_{\bX},x)^{-\frac{1}{2}} \nonumber  \\
\ge &
\sum_{x \in T_n(Q_2)}
\rP_{\bX} (b(C_{\bX},x)=1 )
e^{-\frac{n}{2}(D(Q_1\|P_X)+H(Q_1))}
e^{-\frac{n}{2}(D(Q_2\|P_X)+H(Q_2))} (\rE_{\bX|b(C_{\bX},x)=1} N(C_{\bX},x))^{-\frac{1}{2}} .\Label{8-16-11}
\end{align}
Now, we evaluate the values
$\rP_{\bX} (b(C_{\bX},x)=1 )$ and 
$\rE_{\bX|b(C_{\bX},x)=1} N(C_{\bX},x)$.

The condition
$(x+C_{\bX}) \cap T_n(Q_1) \neq \emptyset$ is equivalent with
the condition
$C_{\bX} \cap (T_n(Q_1) -x)
\neq \emptyset$,
where $(T_n(Q_1) -x):=\cup_{y\in T_n(Q_1)} (y-x) $.
When $y_1, \ldots, y_l$ are all non-zero elements of $(T_n(Q_1) -x)$ for a fixed $x$,
the number of cases that $C_{\bX} \cap (T_n(Q_1) -x) \neq \emptyset$
is 
$M(y_1, \ldots, y_l|\emptyset)$.
Lemma \ref{L8-13-1} guarantees that
$M(y_1, \ldots, y_l|\emptyset) \ge (|T_n(Q_1)|-1)(\frac{2^n-2}{2^m-2}-\frac{|T_n(Q_1)|-2}{2} )
\prod_{i=2}^{m-1}\frac{2^n-2^i}{2^m-2^i}$.
Thus,
\begin{align}
\rP_{\bX} (b(C_{\bX},x)=1 )
\ge &
(|T_n(Q_1)|-1)(\frac{2^n-2}{2^m-2}-\frac{|T_n(Q_1)|-2}{2} )
\frac{2^m(2^m-2)}{2^n(2^n-2)} \nonumber \\
=&
(|T_n(Q_1)|-1)(1-\frac{(2^m-2)(|T_n(Q_1)|-2)}{2(2^n-2)} )2^{m-n}
\Label{8-16-12}.
\end{align}
The number $N(C,x)$ is the number of elements of $C \cap (T_n(Q_2)-x)$.
For any non-zero element $y'\in (T_n(Q_2)-x)$, 
$M(y_1, \ldots, y_l| y')
\le l \prod_{i=2}^{m-1}\frac{2^n-2^i}{2^m-2^i} $.
Hence, we have
\begin{align}
& \rE_{\bX|b(C_{\bX},x)=1} N(C_{\bX},x) 
=1+\sum_{y'\in (T_n(Q_2)-x) \setminus \{0\}}
\rP_{\bX|b(C_{\bX},x)=1} (y' \in C_{\bX})\nonumber \\
=&1+\sum_{y'\in (T_n(Q_2)-x) \setminus \{0\}}
\frac{M(y_1, \ldots, y_l| y')}{M(y_1, \ldots, y_l| \emptyset)}
\le 
1+ \sum_{y'\in (T_n(Q_2)-x) \setminus \{0\}}
\frac{1}{\frac{2^n-2}{2^m-2}-\frac{|T_n(Q_1)|-2}{2}} 
= 
1+ \frac{|T_n(Q_2)|-1}{\frac{2^n-2}{2^m-2}-\frac{|T_n(Q_1)|-2}{2}} .
\Label{8-16-13}
\end{align}

Combining (\ref{8-16-11}), (\ref{8-16-12}), and (\ref{8-16-13}), we obtain
\begin{align}
& \rE_{\bX}
\sum_{[x_2] \in \FF_2^n/C_{\bX}}
\sqrt{P_X^n(x([x_2]))(P_X^n([x_2])- P_X^n(x([x_2]))) } \nonumber \\
\ge &
\sum_{x \in T_n(Q_2)}
e^{-\frac{n}{2}(D(Q_1\|P_X)+H(Q_1)+D(Q_2\|P_X)+H(Q_2))} 
\frac{(|T_n(Q_1)|-1)(1-\frac{(2^m-2)(|T_n(Q_1)|-2)}{2(2^n-2)} )2^{m-n}}
{(1+ \frac{|T_n(Q_2)|-1}{\frac{2^n-2}{2^m-2}-\frac{|T_n(Q_1)|-2}{2}} 
)^{\frac{1}{2}}} \nonumber  \\
= &
e^{-\frac{n}{2}(D(Q_1\|P_X)+H(Q_1)+D(Q_2\|P_X)+H(Q_2))} 
\frac{|T_n(Q_2)|(|T_n(Q_1)|-1)(1-\frac{(2^m-2)(|T_n(Q_1)|-2)}{2(2^n-2)} )2^{m-n}}
{(1+ \frac{|T_n(Q_2)|-1}{\frac{2^n-2}{2^m-2}-\frac{|T_n(Q_1)|-2}{2}} 
)^{\frac{1}{2}}}  ,
\end{align}
which implies (\ref{8-13-e}).

\begin{proofof}{Lemma \ref{L8-13-1}}
We fix the one-dimensional subspace spanned by a non-zero element $y_1\in \FF_2^n$.
We count the number of $m-1$ dimensional subspaces that are orthogonal to $y_1$ and belong to $C_{\bX}$.
Hence, $M(y_1|\emptyset)$ is $\prod_{i=1}^{m-1}\frac{2^n-2^i}{2^m-2^i}$.

Next, we consider two elements $y_1$ and $y_2$.
We fix the two-dimensional subspace spanned by $y_1$ and $y_2$ in $\FF_2^n$.
We count the number of $m-2$ dimensional subspaces
that are orthogonal to the two-dimensional subspace and belong to $C_{\bX}$.
Hence, $M(y_1|y_2)$ is $\prod_{i=2}^{m-1}\frac{2^n-2^i}{2^m-2^i}$.
Thus, 
\begin{align*}
& M(y_1,y_2|\emptyset)=
M(y_1|\emptyset)+M(y_2|\emptyset)-M(y_1|y_2) \nonumber \\
=&
2\prod_{i=1}^{m-1}\frac{2^n-2^i}{2^m-2^i}
-\prod_{i=2}^{m-1}\frac{2^n-2^i}{2^m-2^i}
=(2\frac{2^n-2^1}{2^m-2^1}-1)
\prod_{i=2}^{m-1}\frac{2^n-2^i}{2^m-2^i}.\nonumber 
\end{align*}

We consider $l+1$ elements $y_1,y_2, \ldots, y_l,y_{l+1} \in \FF_2^n \setminus \{0\}$.
We focus on the two-dimensional subspace $C'$ spanned by $y_{l+1}$ and one of $y_1, \ldots, y_l$.
The number of choices of $C'$ is at most $l$.
When we fix the subspace $C'$,
we consider the number of cases 
what $m-2$ dimensional space of the orthogonal space belongs $C_{\bX}$.
This number of cases is $\prod_{i=2}^{m-1}\frac{2^n-2^i}{2^m-2^i}$.
Hence, we obtain (\ref{8-13-1}).

Using (\ref{8-13-1}), we can show (\ref{8-13-2}) with $l=3$ as follows.
\begin{align*}
& M(y_1,y_2,y_3|\emptyset)=
M(y_1|\emptyset)+M(y_2|\emptyset)+M(y_3|\emptyset)
-M(y_1|y_2) -M(y_1,y_2|y_3) \\
\ge &
3\prod_{i=1}^{m-1}\frac{2^n-2^i}{2^m-2^i}
-(1+2)\prod_{i=2}^{m-1}\frac{2^n-2^i}{2^m-2^i}
=(3\frac{2^n-2^1}{2^m-2^1}-3)
\prod_{i=2}^{m-1}\frac{2^n-2^i}{2^m-2^i}.
\end{align*}
Similarly, using (\ref{8-13-1}), we can show (\ref{8-13-2}) in the general case as follows.
\begin{align*}
& M(y_1,y_2,\ldots, y_l|\emptyset)=
M(y_1|\emptyset)+M(y_2|\emptyset)+\ldots+M(y_l|\emptyset)
-M(y_1|y_2) -M(y_1,y_2|y_3) -\ldots -M(y_1,\ldots,y_{l-1}|y_l) \\
\ge &
l\prod_{i=1}^{m-1}\frac{2^n-2^i}{2^m-2^i}
-(1+2+\ldots+(M-1) )\prod_{i=2}^{m-1}\frac{2^n-2^i}{2^m-2^i}
=(l\frac{2^n-2^1}{2^m-2^1}-\frac{l(l-1)}{2} )
\prod_{i=2}^{m-1}\frac{2^n-2^i}{2^m-2^i}.
\end{align*}
\end{proofof}

\subsection{Proof of Lemma \ref{L8-22-1}}\Label{sL8-22-1}
It is enough to show that 
\begin{align}
\max_{0 \le s \le 1/2}
\frac{-f (s)  +s r}{1-s}
=& 
\min_{Q:r \ge H(Q)} D(Q\|P_X) + r - H(Q) 
\Label{8-13-ic},
\end{align}
where $f(s):=sH_{1-s}(P_X)$.
Since both quantities are zero when $r \le H(P_X)$,
it is enough to show (\ref{8-13-ic}) with $r > H(P_X)$.

We define the distribution $P_s(x):= P_X(x)^{1-s}/\sum_{x'}P_X(x')^{1-s}$.
Since $f(s)$ is strictly convex,
$f'(s)$ is strictly increasing.
Hence, we can define the function $s(t)$ as the inverse function 
$s \mapsto f'(s)$.
Since 
\begin{align}
\frac{d}{dt} (1-s(t)) t +f(s(t))
= 1-s(t) -s'(t) t + s'(t) f'(s(t))
= 1-s(t) >0
\end{align}
for $s(t) \in [0,1)$,
we can define $t_r$ as
\begin{align}
r=(1-s(t_r)) t_r +f(s(t_r)).\Label{8-22-9}
\end{align}
Then,
we have
$s(H(P_X))=0$, 
$t_{H(P_X)}=H(P_X)$,
and
$t_{H(P_s)}=f'(s)$.

Hence, 
when $r \in [H(P_X),H(P_1)]$,
we obtain
\begin{align}
t_r -r= t_r s(t_r) - f(s(t_r))
=\frac{s(t_r)r -f(s(t_r)) }{1-s(t_r)}
=\max_{s \in [0,1]} 
\frac{s r -f(s) }{1-s} \Label{8-22-1},
\end{align}
which is shown below.
In the following, we denote the above value by $g(r)$.
Hence, we obtain
\begin{align}
\max_{s \in [0,1/2]} 
\frac{s r -f(s) }{1-s} 
=
\left\{
\begin{array}{ll}
\frac{s(t_r)r -f(s(t_r)) }{1-s(t_r)}
& \hbox{if } s(t_r) \ge 1/2 \\
\frac{H(P_{1/2})/2 -f(1/2) }{1-1/2}
 r - H(P_{1/2}) & \hbox{if } s(t_r) < 1/2 .
\end{array}
\right.
\Label{8-22-6}
\end{align}

We can also show
\begin{align}
\frac{d}{dr} g(r)
= \frac{s(t_r)}{1-s(t_r)}.
\Label{8-22-3}
\end{align}
Its proof is given below.
By simple calculation, 
we obtain
\begin{align}
D(P_s\|P_X)= s f'(s) -f (s). 
\Label{8-22-2}
\end{align}
When $H(Q)=H(P_s)$, 
we can show
\begin{align}
D(Q\|P_X) -D(P_s\|P_X)= 
\frac{D(Q\|P_s)}{1-s}.
\Label{8-22-4}
\end{align}
Its proof is given below.
Combining (\ref{8-22-1}), (\ref{8-22-2}), and (\ref{8-22-4}), we obtain
\begin{align}
\max_{s \in [0,1]} 
\frac{s r -f(s) }{1-s} 
=
\min_{Q:r = H(Q)} D(Q\|P_X) 
=D(P_{s(t_r)}\|P_X)= s(t_r) t_r -f (s(t_r)). 
\Label{8-22-8}
\end{align}
Hence, (\ref{8-22-3}) and (\ref{8-22-8}) yield that
\begin{align}
\min_{Q:r \ge H(Q)} D(Q\|P_X) + r - H(Q) 
=
\min_{r':r \ge r'}
g(r') + r - r' 
=
\left\{
\begin{array}{ll}
g(r') & \hbox{if } s(t_r) \ge 1/2 \\
g(H(P_{1/2})+ r - H(P_{1/2}) & \hbox{if } s(t_r) < 1/2 .
\end{array}
\right.
\Label{8-22-5}
\end{align}
Therefore, combination of (\ref{8-22-6}) and 
(\ref{8-22-5}) yields (\ref{8-13-ic}).

\begin{proofof}{(\ref{8-22-1})}
The first equation follows from (\ref{8-22-9}).
The second equation can be shown
by substituting $t_r= \frac{r-f(s(t_r))}{1-s(t_r)}$.
Now, we show the final equation.
We have
\begin{align}
\frac{d}{ds}\frac{s r -f(s) }{1-s} 
&= \frac{(1-s)(r-f'(s))+sr-f(s)  }{(1-s)^2} .
\end{align}
Since
\begin{align}
\frac{d}{ds}(1-s)(r-f'(s))+sr-f(s) 
= -f''(s)(1-s),
\end{align}
$(1-s)(r-f'(s))+sr-f(s)$ is monotonically increasing for $s$,
Hence, the maximum $\max_{s \in [0,1]} \frac{s r -f(s) }{1-s}$
is realized when
$(1-s)(r-f'(s))+sr-f(s)=0$,
which is equivalent with 
$s=s(t_r)$ because of (\ref{8-22-9}).
Therefore, we obtain the final equation.
\end{proofof}

\begin{proofof}{(\ref{8-22-3})}
Thanks to the proof of (\ref{8-22-1}),
we have $\frac{d}{dr}\frac{s r -f(s) }{1-s}|_{r=s(t_r)}=0$.
Hence,
\begin{align}
\frac{d}{dr}
\frac{s(t_r)r -f(s(t_r)) }{1-s(t_r)}
=
\frac{s(t_r)-f(s(t_r)) }{1-s(t_r)}
+
\frac{d s(t_r)}{dr}
\frac{d}{dr}\frac{s r -f(s) }{1-s}|_{r=s(t_r)}
=
\frac{s(t_r)-f(s(t_r)) }{1-s(t_r)}.
\end{align}
\end{proofof}

\begin{proofof}{(\ref{8-22-4})}
We have
\begin{align*}
&D(Q\|P_X) -D(P_s\|P_X)
=\sum_{x} Q(x)(\log Q(x) - \log P_X(x))
- \sum_{x} P_s(x) (\log P_s(x) - \log P_X(x))\\
=&\sum_{x} Q(x)(\log Q(x) - \log P_s(x))
+ \sum_{x}(Q(x) - P_s(x)) (\log P_s(x) - \log P_X(x))\\
=&D(Q\|P_s) - s  \sum_x (Q(x) - P_s(x)) \log P_X(x)
\end{align*}
and
\begin{align*}
&-H(Q)+H(P_s)
=\sum_{x} Q(x)(\log Q(x) - \log P_s(x))
+ \sum_{x}(Q(x) - P_s(x)) \log P_s(x)\\
=&D(Q\|P_s) +(1- s)  \sum_x (Q(x) - P_s(x)) \log P_X(x).
\end{align*}
Since $H(Q)=H(P_s)$, we obtain (\ref{8-22-4}).
\end{proofof}

\bibliographystyle{IEEE}

\end{document}